\documentclass[11pt]{article}
\usepackage{amsthm,amsmath,amssymb,amsfonts}
\usepackage{epsfig}
\usepackage{xspace}
\usepackage{color,fancybox,graphicx,subfigure,fullpage}
\usepackage[top=1in, bottom=1in, left=1in, right=1in]{geometry}
\usepackage{hyperref}
\usepackage{pdfsync}
\usepackage{multicol}
\usepackage{cite,cleveref}

\usepackage[most]{tcolorbox, empheq}
\newtcbox{\mymath}[1][]{%
nobeforeafter, math upper, tcbox raise base,
enhanced, colframe=blue!30!black,
colback=blue!30, boxrule=1pt,
#1}

\usepackage{verbatim}

\theoremstyle{definition}

\usepackage{algorithmic}
\usepackage{algorithm,caption}

\usepackage{mhequ}
\def \be{\begin{equs}}
\def \ee{\end{equs}}

\newtheorem{theorem}{Theorem}[section]
\newtheorem{lemma}[theorem]{Lemma}

\newtheorem{proposition}[theorem]{Proposition}
\newtheorem{fact}[theorem]{Fact}

\newtheorem*{theorem*}{Theorem}
\newtheorem{remark}[theorem]{Remark}

\usepackage{latexsym,nicefrac,bbm}
\newcommand{\nfrac}{\nicefrac}

\usepackage{rotating}
\usepackage{xcolor}
\usepackage{array}
\usepackage{longtable}

\usepackage{mathrsfs}
\newcommand{\cE}{\mathscr{E}}
\newcommand{\cEC}{\overline{\mathscr{E}}}
\newcommand{\cH}{\mathscr{H}}
\newcommand{\cHC}{\overline{\mathscr{H}}}
\newcommand{\iid}{\text{iid} }
\newcommand{\proofHspace}{\ }
\colorlet{RED}{red}
\colorlet{BLACK}{black}
\usepackage{amsmath}
\usepackage{mathtools}
\usepackage{pifont}
\usepackage{enumitem}
\usepackage{subfiles}

\newcommand{\opt}{{\sc Opt}}
\newcommand{\uncons}{{\sc Uncons}}
\newcommand{\cons}{{\sc Cons}}
\newcommand{\picons}{x_{\text{cons}}}
\newcommand{\piuncons}{x_{\text{uncons}}}
\newcommand{\piopt}{x^\star}

\newcommand{\Bracket}[1]{\bigg[#1\bigg]}
\newcommand{\bracket}[1]{\big[#1\big]}
\newcommand{\Paren}[1]{\bigg(#1\bigg)}
\newcommand{\paren}[1]{\big(#1\big)}

\newcommand{\ohalf}{O\paren{n^{-\nfrac{3}{8}}}}
\newcommand{\oexp}{O\paren{e^{-n^{\nfrac{1}{4}}}}}
\newcommand{\onexp}{O\paren{ne^{-n^{\nfrac{1}{4}}}}}
\newcommand{\ohalfd}{O\paren{n^{-\nfrac{1}{2}+\delta}}}
\newcommand{\oexpd}{O\paren{e^{-n^{2\delta}}}}

\newcommand{\cR}{\mathbb{R}}
\newcommand{\cN}{\mathbb{N}}
\newcommand{\cZ}{\mathbb{Z}}
\newcommand{\cU}{\mathcal{U}}

\newcommand{\floor}[1]{\left\lfloor#1\right\rfloor}
\newcommand{\ceil}[1]{\left\lceil#1\right\rceil}
\newcommand{\myeqref}[2]{(\hyperref[#1]{#2})}

\newcommand{\white}[1]{ \textcolor{white} {#1}}
\newcommand\labelthis[1]{\addtocounter{equation}{1}\tag{{#1},\ \theequation}}
\newcommand\emptylabel[1]{\addtocounter{equation}{1}\tag{\theequation}}

\newcommand\numberthis{\addtocounter{equation}{1}\tag{\theequation}}
\DeclareMathOperator*{\eE}{\mathbb{E}}
\newcommand{\argmax}{\operatornamewithlimits{argmax}}

\renewcommand{\epsilon}{\varepsilon}

\title{Interventions for Ranking in the Presence of Implicit Bias}

\usepackage{authblk}

\author[1]{L. Elisa Celis}
\author[2]{Anay Mehrotra}
\author[3]{Nisheeth K. Vishnoi}
\affil[1,3]{\small Yale University}
\affil[2]{\small Indian Institute of Technology, Kanpur}
\date{}

\begin{document}

\maketitle

\thispagestyle{empty}

\begin{abstract}
  Implicit bias is the unconscious attribution of certain qualities (or lack thereof) to a member from a particular social group (e.g., defined by gender or race).
  Studies on implicit bias have shown that these unconscious stereotypes can have adverse outcomes in various social contexts, such as job screening, teaching, or policing.
  Recently, \cite{KleinbergR18} considered a mathematical model for implicit bias and studied the effectiveness of the Rooney Rule as a constraint to improve the  utility of the outcome for certain cases of the subset selection problem.
  Here we study the problem of designing interventions for a  generalization of  subset selection  -- ranking --
 which requires an \emph{ordered} set and is a central primitive in many social and computational contexts.
  We present a family of simple and interpretable constraints and show that they can optimally mitigate the effect of implicit bias for a generalization of the model studied in~\cite{KleinbergR18}.
  Subsequently, we prove that under natural distributional assumptions on the utilities of items, surprisingly, simple Rooney Rule-like constraints can recover almost all of the utility lost due to implicit biases.
  Finally, we augment our theoretical results with empirical findings on real-world distributions from the IIT-JEE (2009) dataset and the Semantic Scholar Research corpus.
\end{abstract}

\newpage

\thispagestyle{empty}

\setcounter{tocdepth}{2}
\tableofcontents
\newpage

\section{Introduction}
Implicit bias is the unconscious attribution of certain qualities (or lack thereof) to a member from a particular social group defined by characteristics such as gender, origin, or race~\cite{greenwald2006implicit}.
It is well understood that implicit bias is a factor in adverse effects against sub-populations in many societal contexts~\cite{munoz2016big, acm2017statement, bendick2012developing} as also highlighted by recent events in the popular press
\cite{star_bucks_incident, implicit_bias_la_times, implicit_bias_forbes}.
For instance, in employment decisions, men are perceived as more competent and given a higher starting salary even when qualifications are the same~\cite{uhlmann2005constructed}, and in managerial jobs, it was observed that women had to show roughly twice as much evidence of competence as men to be seen as equally competent~\cite{williams2014double, lyness2006fit}.
In education, implicit biases have been shown to exist in ways that exacerbate the achievement gap for racial and ethnic minorities~\cite{van2010implicit} and female students~\cite{moss2012science}, and add to the large racial disparities in school discipline which particularly affect black students' school performance and future prospects~\cite{okonofua2015two}.
Beyond negatively impacting social opportunities, implicit biases have been shown to put lives at stake as they are a factor in police decisions to shoot, negatively impacting people who are black~\cite{correll2007influence} and of other racial or ethnic minorities \cite{sadler2012world}.
Furthermore, decision making that relies on biased measures of quantities such as utility can not only adversely impact those perceived more negatively, but can also lead to sub-optimal outcomes for those harboring these unconscious biases.

To combat this, a significant effort has been placed in developing anti-bias training with the goal of eliminating or reducing implicit biases~\cite{zestcott2016examining, training_facebook, mcgregor2017race}.
However, such programs have been shown to have limited efficacy~\cite{noon2018pointless}.
Furthermore, as algorithms increasingly take over prediction and ranking tasks, e.g., for ranking and shortlisting candidates for interview~\cite{upturn_help_wanted},
algorithms can learn from and encode existing  biases present in past hiring data, e.g., against gender~\cite{amazonRecruitingTool} or race~\cite{term_of_stay_correlated_with_race}, resulting in algorithmic biases of their own.
Hence, it is important to develop interventions that can mitigate implicit biases and hence result in better outcomes.

As a running example  we will consider hiring, although the interventions  we describe would apply to any domain in which people are selected or ranked.
Hiring usually works in multiple stages, where the candidates are first identified and ranked in order of their perceived relevance, a shortlist of these candidates are then interviewed with more rigor, and finally a one or few of them are hired~\cite{upturn_help_wanted}.
The Rooney Rule is an intervention to combat biases during the shortlisting phase, it requires the ``shortlister'' (an entity that shortlists) to select at least {\rm one} candidate from an underprivileged group for interview.
It was originally introduced for coach positions in the National Football League~\cite{collins2007tackling}, and subsequently adopted by other industries~\cite{passariello-implicit-Rooney-wsj, obama_rr, facebook_rr,cavicchia-implicit-bias-Rooney}.
The idea is that including the underprivileged candidates would give opportunity to these candidates, with a higher (hidden or latent) potential.
Whereas without the Rooney Rule these candidates may not have been selected for interview.
While the Rooney Rule appears to have been effective\footnote{The representation of African-American coaches in NFL increased from 6\% to 22\% since Rooney Rule's introduction in 2002 to 2006~\cite{collins2007tackling}.} it is just one of many possible interventions one could design. How can we theoretically compare two proposed interventions?

\cite{KleinbergR18} study the Rooney Rule under a theoretical model of implicit bias, with two disjoint groups $G_a,G_b\subseteq [m]$ of $m$ candidates, where
$G_b$ is the underprivileged group.
Each candidate $i\in[m]$ has a true, {\em latent} utility $w_i\in \cR$, which is the utility they would generate if hired, and an {\em observed} utility $\hat{w_i}\leq w_i$, which is the shortlister's (possibly biased) estimate of $w_i$.
The shortlister selects $n \leq  m$ candidates with the highest observed utility.
For example, in the context of peer-review, the latent utility of a candidate could be their total publications, and the observed utility could be the total weight the reviewer assigns to the publications  (``impact points'' in \cite{wenneras2001nepotism}).
They model implicit bias as a multiplicative factor $\beta\in [0,1]$, where the observed utility is $\hat{w_i}\coloneqq \beta\cdot w_i$ if $i\in G_b$, and $\hat{w_i}\coloneqq w_i$ if $i\in G_a$.
\cite{KleinbergR18} characterize conditions on $n,m,\beta$ and the distribution of $w_i$, where Rooney rule increases the total utility of the selection.

However, before the shortlisting phase, applicants or potential candidates must be identified and ranked.
For example, LinkedIn Recruiter predicts a candidate's ``likelihood of being hired'' from their activity, Koru Hire analyzes a candidates (derived) personality traits to generate a ``fit score'', and HireVue ``grades'' candidates to produce an ``insight score''~\cite{upturn_help_wanted}.
In a ranking, the candidates are ordered (in lieu of an intervention, by their observed utilities), and the utility of a ranking is defined by a weighted sum of the latent utilities of the ranked candidates where the weight decreases the farther down the ranking a candidate is placed.
Such weighting when evaluating rankings is common practice, and due to the fact that candidates placed lower in the list receive a lower attention as compared to those placed higher~\cite{DCG}; this translates into being less likely to be shortlisted, and contributing less to the total utility of the hiring process.
Therefore, it becomes important to consider interventions to mitigate bias in the ranking phase, and understand their effectiveness in improving the ranking's latent utility.

{\em Can we construct simple and interpretable interventions that increase the latent utility of rankings in the presence of implicit bias?}

\subsection{Our contributions}
We consider the setting where items to be ranked may belong to multiple intersectional groups and present a straightforward generalization of the implicit bias model in \cite{KleinbergR18} for this setting.
We consider a set of interventions (phrased in terms of constraints) for the ranking problem which require that a fixed number of candidates from the under privileged group(s) be represented in the top $k$ positions in the ranking for all $k$.
We show that for any input to the ranking problem -- i.e., any set of utilities or bias -- there is an intervention as above that leads to optimal latent utility (Theorem~\ref{thm_constraints_are_sufficient}).
We then prove a structural result about the optimal ranking when all the utilities are drawn from the same distribution, making no assumption on the distribution itself (Theorem~\ref{fact_spread_dist}).
This theorem gives a simple Rooney Rule-like intervention for ranking in the case when there are two groups: For some $\alpha$, require that in the top $k$ positions, there are at least $\alpha \cdot k$ items from the underprivileged group.
We then show that when the utilities are drawn from the uniform distribution, the latent utility of a ranking that maximizes the biased utility but is subject to our constraints (for an appropriate $\alpha$) is close to that of the optimal latent utility (Theorem~\ref{thm_const_increase_utility}).
We  evaluate the performance of ranking under our constraints empirically on two real-world datasets, the IIT-JEE 2009 dataset (see Section~\ref{sec:empirical_results})
and the Semantics Scholar dataset
(see Section~\ref{sec_extnd_empirical}).
In both cases we observe that our simple constraints significantly improve the latent utility of the ranking, and in fact attain near-optimal latent utility.
Finally, while we phrase these results in the context of implicit bias, such interventions would be effective whenever the observed utilities are systematically biased against a particular group.

\subsection{Related work}
There  is a large body of work on studying the effects of bias in rankings, and designing algorithms for `fair rankings'; see, e.g., \cite{SinghJ18,celis2018ranking,kuhlman2019fare,sapiezynski2019quantifying,narasimhan2019pairwise}.
We refer the reader to the excellent talk~\cite{castillo2019fairness} that gives an overview of work on fairness and transparency in rankings.
A significant portion of these works are devoted to generating unbiased rankings.
For example, several approaches strive to learn the latent utilities of items and output a ranking according to the learned values~\cite{YangS17, abs-1712-09752}.
In contrast, we do not strive to learn the latent utilities, rather, to find a ranking that is close to that given by the (unknown) ranking according to the (unknown) latent utilities.
A different approach instead considers constraints on the placement of individuals within the ranking depending on their group status, e.g., enforcing that at least $x\%$ of the top $k$ candidates be non-male~\cite{celis2018ranking, linkedin_ranking_paper}. These works take the constraints as input and develop algorithms to find the optimal feasible ranking.
While we also use constraints in our approach, our goal differs; we strive to show that such constraints can recover the optimal latent utility ranking, and, where possible, derive the appropriate constraints that achieve this.

Constraints which guarantee representation across various groups have been studied in a number of works on fairness across various applications~\cite{FatML, CKSDKV18, CKSV18, celis2019advertising, HJV19},
most relevantly in works on forms of subset selection \cite{celis2018multiwinner} and ranking~\cite{linkedin_ranking_paper, celis2018ranking, ZehlikeB0HMB17}.
The primary goals of these works is to design fast algorithms that satisfy the constraints towards satisfying certain definitions of fairness; these fairness goals are given exogenously and the utilities are assumed to be unbiased.
In contrast, we begin with the premise that the utilities are systematically incorrect due to implicit bias, and use the constraints to mitigate the effect of these biases when constructing a ranking.
Our goal is to determine how to construct effective interventions rather than on the algorithm for solving the constraints; in fact, we use some of the works above as a black box for the subroutine of finding a ranking once the constraints have been determined.
Studying implicit and explicit biases is a rich field in psychology~\cite{lyness2006fit, williams2014double, sadler2012world}, where studies propose several mechanisms for origins of biases~\cite{Payne11693} and analyze present-day factors which can influence them~\cite{Epstein2015}.
We point the reader to the seminal work on implicit biases~\cite{greenwald1995implicit}, and the excellent treatise~\cite{whitley2016psychology} for an overview of the field.
We consider one model of implicit bias inspired by \cite{KleinbergR18,wenneras2001nepotism}; however other relevant models may exist and exploring other kinds of bias and models for capturing them could lead to interesting expansions of this work.
\section{Ranking problem, bias model, and interventions}
\subsection{Ranking problem}
In the classical ranking problem, one is given $m$ items, where item $i$ has utility $w_i$, from which a ranked list of $n \leq m$ items has to be outputted.
A ranking is a one to one mapping from the set of items $[m]$ to the set of positions $[n]$.
It is sometimes convenient to  let $x\in \{0,1\}^{m\times n}$ denote a binary assignment matrix representing a ranking, where  $x_{ij}=1$ if the $i$-th item is placed at position $j$, and is $0$ otherwise.
Define a position-based discount $v\in \cR^{n}_{\geq 0}$, where an item placed at position $j\in [n]$, contributes a latent utility of $w_i\cdot v_j$.
The latent utility obtained by placing an item $i$ at position $j$ is then $w_i v_j$.
It is assumed that  $v_j \geq v_{j+1}$ for all $1 \leq j \leq n-1$ implying that the same item derives a higher utility at a lower position.
This is satisfied by popularly studied position-based discounts such as discounted cumulative gain (DCG)~\cite{DCG} where $v_k\coloneqq \nfrac{1}{\log (k+1)}$ (and its variants) and Zipfian where $v_k\coloneqq\nfrac{1}{k}$~\cite{kanoulas2009empirical}.
Then, given $v\in \cR^{n}_{\geq}$ we define the {\em latent} utility of a ranking $x$ as
\begin{align}
  \hspace{14mm}\mathcal{W}(x,v,w)\coloneqq \sum_{i\in[m],\ j\in[n]}x_{ij}w_iv_j. \labelthis{Latent utility}
\end{align}
The goal of the ranking problem is to find a ranking (equivalently, an assignment)  that maximizes the latent utility:
$$\argmax\nolimits_x \mathcal{W}(x,v,w).$$
The reason  the utility above is called ``latent'' is, as is shortly discussed, in the presence of implicit bias, the perceived utility may be different from the latent utility.
Note that subset selection is a special case of the ranking problem when $v_j=1$ for all $j \in [n]$.

\subsection{Groups and a model for implicit bias}\label{sec:implicitbias}
Items may belong to one or more of intersectional (i.e., not necessarily disjoint) groups $G_1, G_2,\ldots, G_p$.
Each $G_s \subseteq [m]$ and $G_s \cap G_t$ may not be empty for $s \neq t$.
The perceived or observed utility of group items in $G_s$ may be different from their latent utility.
And, items that belong to multiple groups may be subject to multiple implicit biases, as has been observed~\cite{bertrand2004emily, uhlmann2005constructed}.
To mathematically model this,  we consider a model of implicit bias introduced by  by~\cite{KleinbergR18}, which is  motivated from  empirical findings of~\cite{wenneras2001nepotism}:
For two disjoint groups, $G_a, G_b\subseteq [m]$, given an implicit bias parameter $\beta\in[0,1]$ they defined the observed utility as
\begin{align}
  \hspace{14mm}\hat{w}_i \coloneqq \begin{cases}
  w_i & \text{if } i\in G_a\\
  \beta w_i & \text{if } i\in G_b.
\end{cases}\labelthis{Observed utility}
\end{align}

\noindent
To extend this model to the case of multiple intersectional groups,   for each $s\in [p]$, we assume an implicit bias parameter $\beta_s\in[0,1]$.
Since, it is natural to expect that items at the intersection of multiple groups encounter a higher bias~\cite{williams2014double},
we define their implicit bias parameter as the product of the implicit biases of each group the item belongs to.
Formally, the observed utility $\hat{w}_i$ of item $i\in [m]$ is
\begin{align}
  \hat{w}_i\coloneqq \bigg(\prod_{s\in [p]\colon G_s\ni i}\beta_s \bigg)\cdot w_{i}.
\end{align}
It follows that the case of two disjoint groups $G_a,G_b$ is a special case;  let $\beta_a=1$ and $\beta_b=\beta$.

\subsection{Intervention constraints}\label{sec:interventions}
In the presence of implicit bias, the ranking problem then results in finding the assignment matrix $x$ that maximizes the observed utility $\mathcal{W}(x,v,\hat{w})$.
Thus, not only does it result in adverse outcomes for groups for which $\beta_s <1$, it also follows that optimizing this utility as such may be sub-optimal for the overall goal of finding a utility maximizing rank.
To see this note that if $x^\star$ is the assignment that maximizes $\mathcal{W}(x,v,\hat{w})$, the value of the latent utility, $\mathcal{W}(x^\star,v,w)$, derived from it may be much less.

Motivated by the Rooney Rule and its efficacy as an intervention in the subset selection problem \cite{KleinbergR18}, we investigate if there are constraints that can be added to the optimization problem of finding a ranking that maximizes the observed utilities, that results in a ranking in which the latent utility is much higher, possibly even  optimal: $\max_{x} \mathcal{W}(x,v,{w})$.
As a class of potential interventions, we consider lower bound constraints on the number of items from a particular group $s\in [p]$, selected in the top-$k$ positions of the ranking, for all positions $k\in [n]$.
More specifically, given  $L\in \cZ_{\geq 0}^{n\times p}$ we consider the following constraints on rankings (assignments):
\begin{align}\label{eq:deterministic_constraints}
  \hspace{-6mm}\forall \ k \in [n],\ s\in [p] \ \ \ \ \  L_{ks} \leq \sum_{j\in [k]}\sum_{i\in G_s}x_{ij}.
\end{align}
We will sometimes refer to these constraints  as $L$-constraints.
Let
\begin{align}
  \mathcal{K}(L)\coloneqq \{x\in \{0,1\}^{m\times n}\ |\ x \text{ satisfies } L\text{-constraints}\}
\end{align}
be the set of all rankings satisfying the $L$-constraint.
Our goal will be to consider various $L$-constraints and understand under what conditions on the input utilities and bias parameters does the ranking
\begin{align}
  \tilde{x} \coloneqq \argmax\nolimits_{x \in \mathcal{K}(L)} \mathcal{W}(x,v,\hat{w})
\end{align}
have the property that
$\mathcal{W}(\tilde{x},v,w)$ close to $\max_{x} \mathcal{W}(x,v,{w})$.

Constraints such as \eqref{eq:deterministic_constraints} have been studied by a number of works on fairness,
including by several works on ranking~\cite{celis2018ranking, linkedin_ranking_paper,ZehlikeB0HMB17}
While these constraints can encapsulate a variety of fairness and diversity metrics~\cite{celis2019classification}, their effectiveness as an intervention for implicit biases was not clear and the utility of the rankings generated remained ill-understood prior to this work.

\section{Theoretical results}
\noindent {\bf Notation.}
Let $Z$ be a random variable. We use $Z_{(k:n)}$ to represent the $k$-th order statistic (the $k$-th largest value) from $n$ \iid draws of $Z$.
For all $a<b$, define $\cU[a,b]$ to be the uniform distribution on $[a,b]$.
More generally, for an interval $I\subseteq \cR$, let $\cU I$ be the uniform distribution on $I$.
Let
\begin{align}
  x^\star\coloneqq \argmax\nolimits_{x}\mathcal{W}(x,v,w)
\end{align}
be the ranking that maximizes the latent utility.
\subsection{$L$-constraints are sufficient to recover optimal latent utility}
Our first result is structural and shows that the class of $L$-constraints defined above are expressive enough to recover the optimal latent utility while optimizing observed utility constrained to certain specific $L \in \mathbb{Z}^{n \times p}_{\geq 0}.$

\begin{theorem}\label{thm_constraints_are_sufficient}
  Given a set of latent utilities $\{w_i\}_{i=1}^{m}$,
  there exists constraints $L(w)\in \cZ^{n\times p}_{\geq 0}$,
  such that,
  for all implicit bias parameters $\{\beta_s\}_{s=1}^{p}\in(0,1)^{p}$,
  the optimal constrained ranking $\tilde{x}\coloneqq$ \\$\argmax_{x\in \mathcal{K}(L(w))} \mathcal{W}(x,v,\hat{w})$ satisfies
  \begin{align*}
    \mathcal{W}(\tilde{x}, v, w)=\max\nolimits_{x}\mathcal{W}(x,v,w).\numberthis
  \end{align*}
\end{theorem}
\noindent
Without additional assumptions, $L(w)$ necessarily depends on $w$; see Fact~\ref{fact_impossibility}.
A set of utility-independent constraints is often preferable due to its simplicity and interpretability; our next two results take  steps in this direction by making assumptions about distributions from which the utilities are drawn.

\smallskip
\noindent {\em Proof sketch of Theorem~\ref{thm_constraints_are_sufficient}.}
Consider the following constraints:
For all $s\in[p]$ and $k\in[n]$,
$$L_{ks}(w) \coloneqq \sum_{i\in G_s, j\in [k]}x_{ij}^\star.$$
Recall that $x^\star\coloneqq \argmax_{x} \mathcal{W}(x,v,w)$ is a function of $w$.
We claim that $L_{ks}(w)$ satisfy the claim in the theorem.
The proof proceeds in two steps.
First, we show that $\tilde{x}$ is the same as $x^\star$ up to the groups of items at each position.
Let $T_i \coloneqq \{s \colon i\in G_s \}$ be the set of groups $i$ belongs to.
This proof relies on the fact that the observed utility of an item is always smaller than its latent utility, and that for any two items $i_1, i_2\in[m]$, if $T_{i_2}\subsetneq T_{i_1}$ and $w_{i_1}=w_{i_2}$, then  $\hat{w}_{i_1}< \hat{w}_{i_2}$.
Using these we show that for each position $k\in[n]$, under the chosen $L$-constraints, it is optimal to greedily place item $i^\prime\in [m]$ that has the highest observed utility and satisfies the constraints.
In the next step, we show that $\tilde{x}$ has the same latent utility as $x^\star$ from a contradiction.
We show that if the claim is satisfied for the first $k\in [n]$ positions, then we can swap two candidates $i_1$ and $i_2$, such that $T_{i_1}=T_{i_2}$, to satisfy the claim for the first $(k+1)$ positions without loosing latent utility.
Here we use the fact for any two items $i_1, i_2\in[m]$, if $T_{i_2}= T_{i_1}$ and $w_{i_1}<w_{i_2}$, then  $\hat{w}_{i_1}<\hat{w}_{i_2}$, i.e.,
the relative order of items in the same set of groups does not change whether we rank them by their observed utility (as in $\tilde{x}$) or their latent utility (as in $x^\star$).\footnote{A minor technicality is that $\tilde{x}$ could swap two items $i_1, i_2\in[m]$, with $T_{i_2}= T_{i_1}$ and $w_{i_1}=w_{i_2}$, relative to $x^\star$.
But this does neither affects the latent utility nor the observed utility.}
The proof of Theorem~\ref{thm_constraints_are_sufficient} is presented in  Section~\ref{sec_proof_thm_constraints_are_sufficient}.

\subsection{Distribution independent constraints}\label{sec_distribution_independent_const}
We now study the problem of coming up with constraints that do not depend on the utilities.
Towards this, we consider the setting of two disjoint groups,  $G_a, G_b\subseteq [m]$.
Let $m_a\coloneqq |G_a|$ and $m_b\coloneqq |G_b|$ be the sizes of $G_a$ and $G_b$.
We assume that the latent utility $W_i$ for all items  $i\in G_a \cup G_b$ is i.i.d. and drawn from some distribution $\mathcal{D}$.
This model is equivalent to the one considered by \cite{KleinbergR18}, except that they fix $\mathcal{D}$ to be in the family of power-law distributions, whereas our result in this section holds for any distribution $\mathcal{D}$.

The optimal ranking will  sort the utilities $(w_i)_{i \in [m]}$ in a decreasing order (breaking ties arbitrarily if necessary).
For all $\ell\in [m_b]$, let $P_\ell\in [m]$ be the random variable representing the position of the $\ell$-th item from $G_b$ in the optimal ranking.
Let $N_{kb}$ be the random variable that counts the number of items belonging to $G_b$ in the first $k$ positions of the optimal ranking.
The following result reveals the structure of the optimal  ranking (when there is no implicit bias) and is used in the next subsection to design utility-independent constraints.
\begin{theorem}\label{fact_spread_dist}
  Let $\mathcal{D}$ be a continuous distribution, $\ell\leq m_b$, and $0<k < \min(m_a, m_b)$ be a position, then
  \begin{align}
    \forall\ \delta\geq 2 \hspace{5mm}&\Pr[\ N_{kb} \leq \eE[N_{kb}] - \delta \ ] \leq e^{-\frac{2(\delta^2-1)}{k}},\label{cons_of_nkb}\\
    &\eE[N_{kb}] = k\cdot\nfrac{m_b}{(m_a+m_b)},\label{expect_nkb}\\
    &\eE[P_\ell] = \ell\cdot\big(1+\nfrac{m_a}{(m_b+1)}\big).\label{eq:expectation_z_ell}
  \end{align}
\end{theorem}
\noindent
Note that this result is independent of the distribution $\mathcal{D}$ and only requires  $\mathcal{D}$ to be a continuous probability distribution.
The above equations show that with high probability  the optimal ranking has $k\cdot\nfrac{m_b}{(m_a+m_b)}$ items from $G_b$ in the top-$k$ positions,
and in expectation,  it places these items at equal distances in the first $k$ positions for all $k\in[n]$.
This observation motivates the following simple constraints.

\noindent {\bf  Simple constraints.}
Given a number $\alpha\in [0,1]$,
we define the  constraints $L(\alpha)$ as follows:
For all $k\in [n]$
\begin{align}\label{eq:simple_constraints}
  &\hspace{20mm}L_{ka} \coloneqq 0 \text{ and } L_{kb} \coloneqq \alpha k.
  \labelthis{Rooney Rule like constraints}
\end{align}
Note that the only non-trivial constraint is on $G_b$.
These constraints are easy to interpret and can be seen as generalization of the Rooney Rule to the ranking setting.

\noindent {\em Proof sketch of Theorem~\ref{fact_spread_dist}.}
Here, we discuss the distributional independence of Theorem~\ref{fact_spread_dist}, and present its proof in Section~\ref{sec_proof_spread_distribution}.
For all $i\in [m]$, $w_i\stackrel{d}{=}\mathcal{D}$ be the random utility of the $i$-th item drawn from $\mathcal{D}$ and $F_{\mathcal{D}}(\cdot)$ be the cumulative distribution function of $\mathcal{D}$.
Then the independence follows from the straightforward facts that for all $\mathcal{D}$, $F_{\mathcal{D}}(w_i)\stackrel{d}{\sim}\cU[0,1]$ (here is one place it is used that $\mathcal{D}$ is continuous) and that $F_{\mathcal{D}}(\cdot)$ (being a cdf) is a monotone function.
From these it follows that $\argmax_x\mathcal{W}(x,v,\{w_i\}_{i=1}^m)=\argmax_x\mathcal{W}(x,v,\{F_{\mathcal{D}}(w_i)\}_{i=1}^m)$.
Thus, we can replace the all $W_i$s by $F_{\mathcal{D}}(w_i)$ without changing the optimal ranking.

\subsection{Optimal latent utility from simple constraints for uniform distributions}\label{sec_utility_bound}
In this section we study the effect of the simple constraints mentioned in the previous subsection when $\mathcal{D}$ is  the uniform distribution on $[0,1]$, for  the setting of two disjoint groups  $G_a, G_b\subseteq [m]$ with $\beta_a=1$ and $\beta_b=\beta$.
We discuss how these arguments could be extended to other bounded distributions in the remarks following Theorem~\ref{thm_const_increase_utility}, and empirically study the increase in utility for a non-bounded distribution in Section~\ref{sec:empirical_results}.

Define the expected utility $U_{\mathcal{D},v}(\alpha, \beta)$ as
\begin{align*}
  U_{\mathcal{D},v}(\alpha, \beta)\coloneqq \eE_{w\gets \mathcal{D}^m}\big[\mathcal{W}(\tilde{x}, v, w)\big]\numberthis\label{expected_utility}
\end{align*}
where for each draw $w$, $\tilde{x}\coloneqq \argmax_{x\in \mathcal{K}(L(\alpha))}\mathcal{W}(x,v,\hat{w})$, and $\hat{w}_i=w_i$  if $i \in G_a$ and $\hat{w}_i=\beta w_i$ if $i \in G_b.$
Sometimes we drop the subscripts $\mathcal{D}$ and $v$ if they are clear from the context.
\begin{theorem}~\label{thm_const_increase_utility}
  Given a $\beta\hspace{-0.5mm}\in\hspace{-0.5mm}(0,1)$,
  if $\mathcal{D}\hspace{-0.5mm}\coloneqq\hspace{-0.5mm} \cU[0,1]$,
  $v$ satisfies Assumptions~\eqref{assumption} and \eqref{assumption_2} with $\epsilon>0$,
  and $n\leq \min(m_a,m_b)$,
  then for $\alpha^\star\coloneqq \frac{m_b}{m_a+m_b}$, then adding $L(\alpha)$-constraints achieve nearly optimal latent utility in expectation:
  \begin{align*}
    U_{\mathcal{D},v}(\alpha^\star, \beta)=U_{\mathcal{D},v}(0,1) \cdot (1-O(n^{-\epsilon/2}+n^{-1})).
  \end{align*}
\end{theorem}
\paragraph{Assumptions.}
\begin{align*}
  \frac{\sum_{k=1}^{n-1} v_k-v_{k+1}}{\sum_{k=1}^{n}v_k} &=\frac{v_1-v_n}{\sum_{k=1}^{n}v_k} = O(n^{-\epsilon})\label{assumption}\numberthis\\
  \forall\ k\in [n],\ v_k-v_{k+1}&\geq v_{k+1}-v_{k+2}.\label{assumption_2}\numberthis
\end{align*}
Roughly, these assumptions mean that the position discounts be much larger than the difference between between the discounts of two consecutive positions.
These are mild assumptions, and are satisfied by several commonly studied position discounts including DCG for $\epsilon=0.9\bar{9}$, and inverse polynomial discount where $v_k\coloneqq k^{-c}$ and $c\in(0,1)$~\cite{wang2013theoretical} for $\epsilon=1-c$.

\medskip

A few remarks are in order:
\begin{enumerate}[itemsep=0pt,wide, labelwidth=!,labelindent=0pt]
  \item When $v_k=1$ for all $k\in [n]$ then we can derive the following explicit expressions of the utility assuming $m_a,m_b\geq n$:
  \begin{align*}
    \hspace{0mm}U_{\mathcal{D},1}(\alpha^\star, \beta)\labelthis{Utility with constraints}\label{eq_utility_with_const}&=n\big(1-\frac{n}{2(m_a+m_b)}+\ohalf\big),\\
    \hspace{0mm}U_{\mathcal{D},1}(0, \beta) \labelthis{Utility without constraints}\label{eq_utility_without_const}&=\begin{cases}
    \frac{m_a(1-\beta^2)}{2} + \frac{m_a\beta^2+m_b}{2}\Bracket{1-\frac{(m_a+m_b-n)^2}{(m_a\beta+m_b)^2}+\ohalf}, &\hspace{-2mm}c=n-\omega(n^{\nfrac{5}{8}})\\
    n\big(1-\frac{n}{2m_a}+\ohalf\big),  &\hspace{-2mm}c\geq n+\Theta(n^{\nfrac{5}{8}}).
  \end{cases}
\end{align*}
Where, we define $c\coloneqq m_a(1-\beta)$.
\item Choosing $\beta=1$ in Equation~\myeqref{eq_utility_without_const}{17} we can see that the optimal latent utility from a ranking of $n$ items is $n(1-\nfrac{n}{(2m_a+2m_b)}+o(1))$, and that the constraints achieve this utility within a $(1-\Theta(n^{-1}))$ multiplicative factor {\em for any} $\beta\in (0,1)$.
Therefore, for all $0<\beta<1$ there exists an $n_0\in \cN$ such that for all $n\geq n_0$, the constraints achieve nearly optimal latent utility.
\begin{remark}\label{rem_constraint_ind_of_beta}
  We note that our choice of $\alpha^\star\coloneqq \frac{m_b}{m_a+m_b}$ is independent of $\beta$.
  Thus, we can achieve near optimal latent utility without the knowledge of $\beta$.
\end{remark}
\item  Extending $\mathcal{D}$ to $\cU[0,C]$, does not change the form of the theorem.
We can simply scale Equations~\myeqref{eq_utility_with_const}{16} and \myeqref{eq_utility_without_const}{17} by $C$.
\item
In the special case of subset selection ($v_k=1$ for all $k$), this theorem answers the open problem in  \cite{KleinbergR18} regarding the $\ell$-th order Rooney rule for the uniform distribution.
\end{enumerate}

\noindent {\em Proof sketch of Theorem~\ref{thm_const_increase_utility}.}
To calculate $U_{\mathcal{D},v}(0, \beta)$, we partition the items from $G_a$ into those with a ``high'' utility (in $(\beta,1]$) and all others (with utility in $[0,\beta]$).
Items with a high utility are always selected before any item from $G_b$.
The number of such items, $N_{a_1}$, is a sum $n$ random variables indicating if the utility is in $[\beta,1)$.
Therefore, $N_{a_1}$ has a binomial distribution $\text{Binomial}(n,1-\beta)$.
Conditioned on  $N_{a_1}$, we can show that the distribution of observed utilities of all remaining items, including those from $G_b$, is the same.
This uses the fact that a uniform random variable conditioned to lie in an sub-interval is uniform.
Using this symmetry we can show that the number of items selected from $G_b$, $N_b$, follows a hypergeometric distribution conditioned on $N_{a_1}$.

This gives us a value for $U_{\mathcal{D},v}(0, \beta)$ conditioned on $N_{a_1}$.
To do away with the conditioning, we derive approximations to the negative moments of the binomial distribution.

To upperbound the difference $U_{\mathcal{D},v}(0, 1)-U_{\mathcal{D},v}(\alpha^\star, \beta)$, we use a coupling argument and the concentration properties of the hypergeometric distribution to show that the difference in positions of any item $i$ between the constrained and the unconstrained ranking is $o(n^{\delta})$ with high probability.
Using Assumptions~\eqref{assumption} and \eqref{assumption_2} with the boundedness of the utility gives us a lower bound of
$$U_{\mathcal{D},v}(0, 1)-U_{\mathcal{D},v}(\alpha^\star, \beta)= O(n^{\delta-\epsilon}+n^{-1})\sum\nolimits_{k=1}^n v_k.$$

Further, using the fact that $U_{\mathcal{D},v}(0, 1)=\Omega(\sum_{k=1}^n v_k)$, the theorem follows by choosing $\delta=\nfrac{\epsilon}{2}$.

To find an explicit expression of $U_{\mathcal{D},v}(\alpha, \beta)$ in the special case when $v_k=1$,
we use a coupling argument to show that if the unconstrained ranking picks fewer than $\nfrac{n}{2}$ items from $G_b$ then the constrained selects exactly $\nfrac{n}{2}$ items from both $G_a$ and $G_b$.
Using the distribution of $N_b$ for the unconstrained case we can show that this occurs with high probability as long as $\beta<1$ (note the strict inequality).
Then, using the boundedness of the utility, we show that $U_{\mathcal{D},v}(\alpha^\star, \beta)$ is twice the sum of expected utility of the highest $\nfrac{n}{2}$ order statistics in $n$ draws from $\cU[0,1]$ which gives us the required expression.
The proof of Theorem~\ref{thm_const_increase_utility} in presented in Section~\ref{sec_proof_thm_const_increase_utility} of the supplementary material.

\begin{remark}
  We expect similar strategy would give us bounds on $U_{\mathcal{D},v}(0, \beta)$ and $U_{\mathcal{D},v}(\alpha^\star, \beta)$ with other bounded distributions as well.
  We provide some evidence in favor of this for a naturally occurring bounded distribution in Section~\ref{sec:empirical_results}.
  However, when $\mathcal{D}$ is an unbounded distribution we cannot ignore events with a low probability and other techniques might be required to estimate the utilities.
  On a positive note, we still observe an increase in utility for the (unbounded) log-normal distribution in our empirical study (Section~\ref{sec:empirical_results}).
\end{remark}

\section{Empirical observations}\label{sec:empirical_results}
We examine the effect of our constraints on naturally occurring distributions of utilities derived from scores in the {\em IIT-JEE 2009} dataset.\footnote{We observe similar results from the citation dataset from the {\em Semantic Scholar Open Research Corpus}, which we present in Section~\ref{sec_extnd_empirical}.}
We consider with two disjoint groups of candidates $G_a$ and $G_b$, representing {male} and female candidates respectively.\footnote{While there could be richer and non-binary gender categories, the above datasets code all of their entries as either male or female.}
First, we analyze the distributions of the scores, $\mathcal{D}_a$ and $\mathcal{D}_b$, attained by $G_a$ and $G_b$, and note that that the distributions of utilities of two groups are very similar in Section~\ref{sec_same_distribution}.
In Section~\ref{sec_exp_implicit_bias} we consider the situation in which these scores accurately capture a candidate's latent utility;\footnote{This may not be the case as examined further in Section~\ref{sec_supernumerary}}
yet implicit bias against candidates in $G_b$, say during a job interview,  affects the  interpretation of the the score of these candidates.
We simulate this implicit bias by shading the distribution of utilities for candidates in $G_b$ by a constant multiplicative factor $\beta\in(0,1]$.
We then measure the effectiveness of our constraints by comparing the latent utilities of:
\begin{enumerate}[itemsep=0pt,leftmargin=16pt]
  \item {\sc \cons{}}: Our proposed ranking, which uses constraints to correct for implicit bias.
  \item {\sc \uncons{}}: An unconstrained ranking.
  \item {\sc \opt{}}: The optimal (unattainable) ranking that maximizes the (unobserved, due to implicit bias) utilities.
\end{enumerate}

In Section~\ref{sec_supernumerary}, we consider the situation in which the scores themselves encode systematic biases, and consider the effectiveness of our constrained ranking as a potential intervention.
In particular, we contrast our approach with a recent intervention used in IIT admissions, {\sc Supernumerary}, which was created to increase the representation of women at IITs~\cite{baswana2015joint}.

\subsection{Dataset: JEE scores}
Indian Institutes of Technology (IITs) are a group of 23 institutes of higher-education, which are, arguably, the most prestigious engineering schools in India.\footnote{The number of IITs in 2009, the year the dataset is from, was 15.}
Currently, undergraduate admissions into the IITs are decided solely on the basis of the scores attained in the Joint Entrance Exam (JEE Advanced; known as IIT-JEE in 2009).
IIT-JEE is conducted once every year, and only students that have graduated from high school in the previous two years are eligible.
Out of the $468,280$, candidates who took IIT-JEE in {2011}, only 9627 candidates (2\%) were admitted to an IIT.
In the same year, $108,653$ women (23.2\%  of the total) appeared in the exam, yet only $926$ were admitted into an IIT (less than 10\% of the 9627 admitted)\cite{JEE2011report}.

\begin{figure}[h!]
  \begin{center}
    \includegraphics[width=0.5\linewidth]{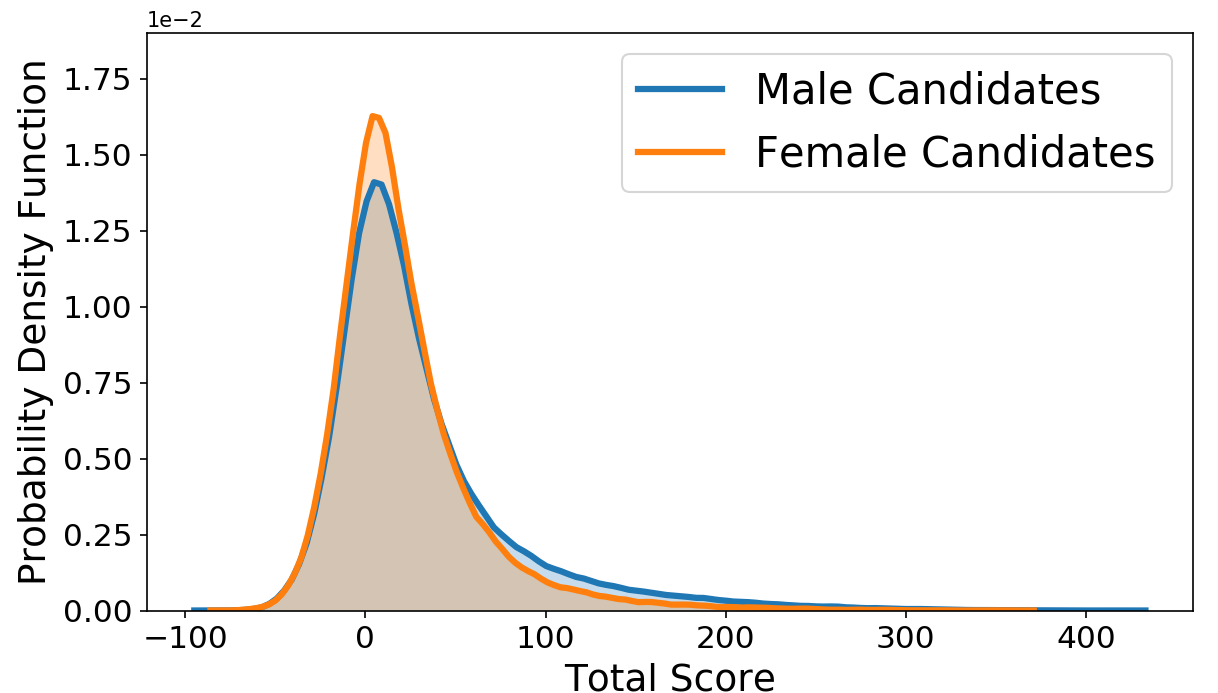}
    \hspace{2mm}
  \end{center}
  \vspace{-3mm}
  \caption{
  {\em Distributions of scores in IIT-JEE 2009:} Distribution of total scores of all male and all female candidates.
  Men and women have similar distributions of total scores, with a total variation distance of $\Delta_{\text{TV}}(\mathcal{D}_a,\mathcal{D}_b)=0.074.$
  }
  \label{jee_distribution_of_men_vs_women}
\end{figure}
The dataset consists of scores of candidates from IIT-JEE 2009 which was released in response to a Right to Information application filed in June 2009~\cite{rti_against_jee}.
This dataset contains the scores of 384,977 students in each of the Math, Physics, and Chemistry sections of IIT-JEE 2009, {along with the student's gender, given as a binary label} (98,028 women and 286,942 men), {their birth category (see \cite{baswana2015joint}), and zip code}.
The candidates are scored on a scale from $-35$ to 160 points in all three sections, with an average total score of $28.36$, a maximum score of 424 and a minimum score of $-86$.
While the statistics of IIT JEE 2009 admissions are not available, if roughly the same number of students were admitted in 2009 as in 2011,
then students with a score above 170 would have been admitted.

\subsubsection{Distribution of scores across groups}\label{sec_same_distribution}
We found that the Johnson's $S_U$-distribution gave a good fit for the distribution of scores.
These fitted distributions of scores of women and men, $\mathcal{D}_b$ and $\mathcal{D}_a$, are depicted in Figure~\ref{jee_distribution_of_men_vs_women}.
The two distributions are very similar; their total-variation distance is $\Delta_{\text{TV}}(\mathcal{D}_a, \mathcal{D}_b)=0.074$, i.e., the two distributions differ on less than 8\% of their probability mass.
However, the mean of men $\hat{\mu}_{a}=30.79$ (standard deviation $\hat{\sigma}_{a}=51.80$) is considerably higher than the mean of women $\hat{\mu}_{b}=21.24$ (standard deviation $\hat{\sigma}_{a}=39.27$).

\begin{table*}[h!]
  \begin{tabular}{|c|ccc|}
    \hline
    & $\nfrac{m_b}{m}=\nfrac{1}{4}$  & $\nfrac{m_b}{m}=\nfrac{1}{3}$  & $\nfrac{m_b}{m}=\nfrac{1}{2}$  \\
    \hline
    &&&\\
    &&&\\
    \vspace{-20mm}\begin{sideways}$\beta=1$\white{........}\end{sideways}&&&\\
    &
    \includegraphics[width=0.32\linewidth, trim={-0.1cm 0cm 0.5cm 1.7cm},clip]{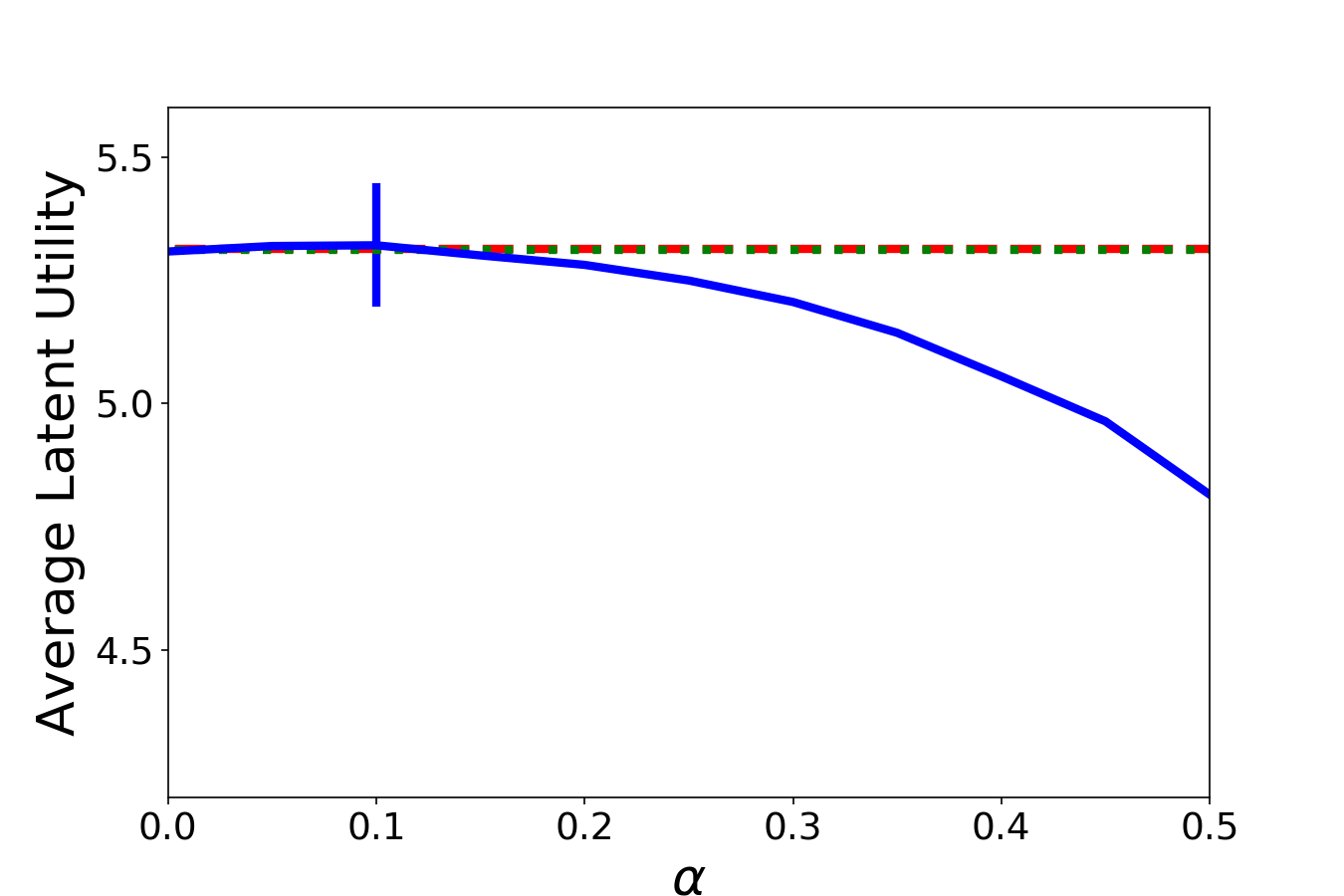}\hspace{-6mm} &
    \includegraphics[width=0.32\linewidth, trim={-0.1cm 0cm 0.5cm 1.7cm},clip]{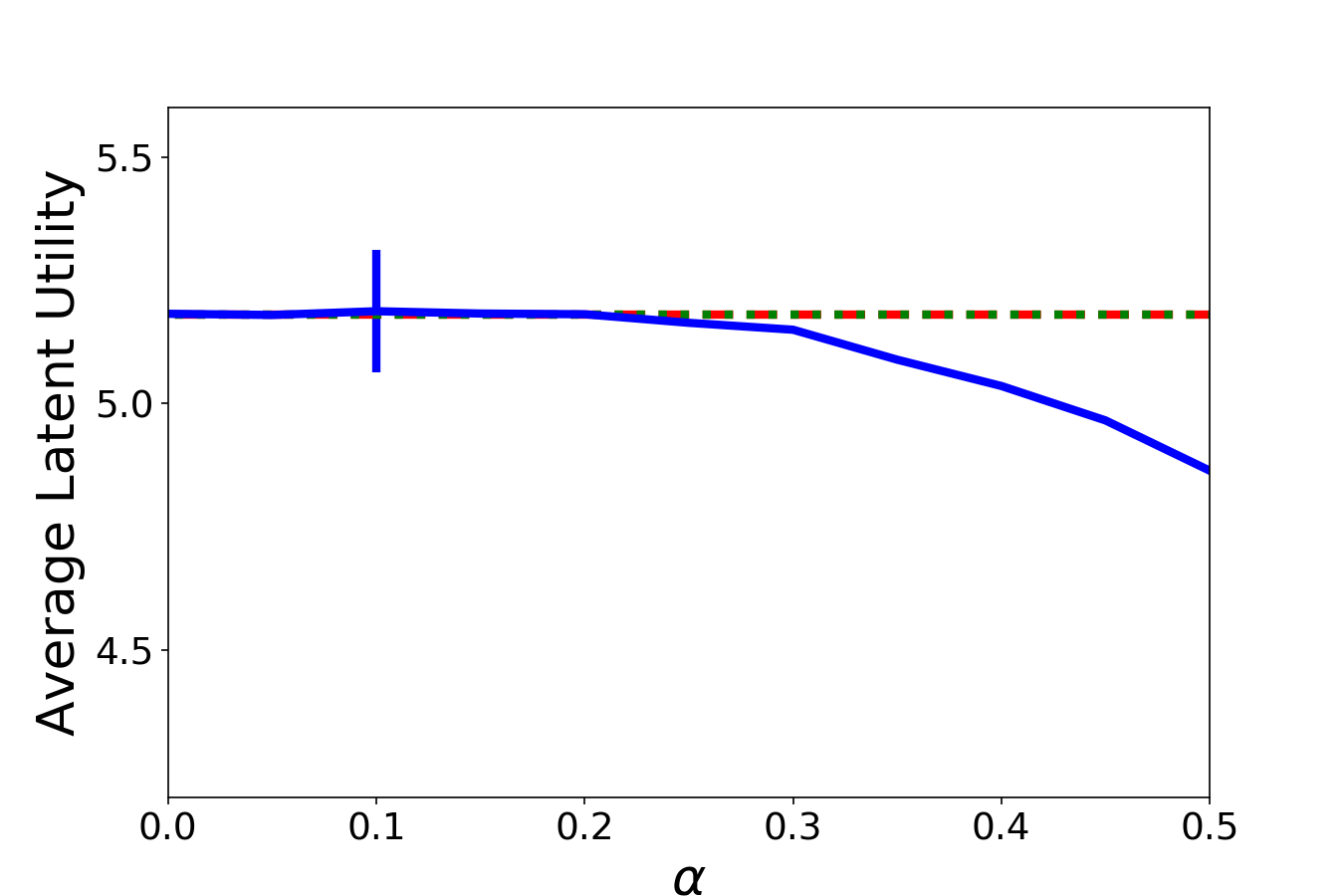}\hspace{-6mm} &
    \includegraphics[width=0.32\linewidth, trim={-0.1cm 0cm 0.5cm 1.7cm},clip]{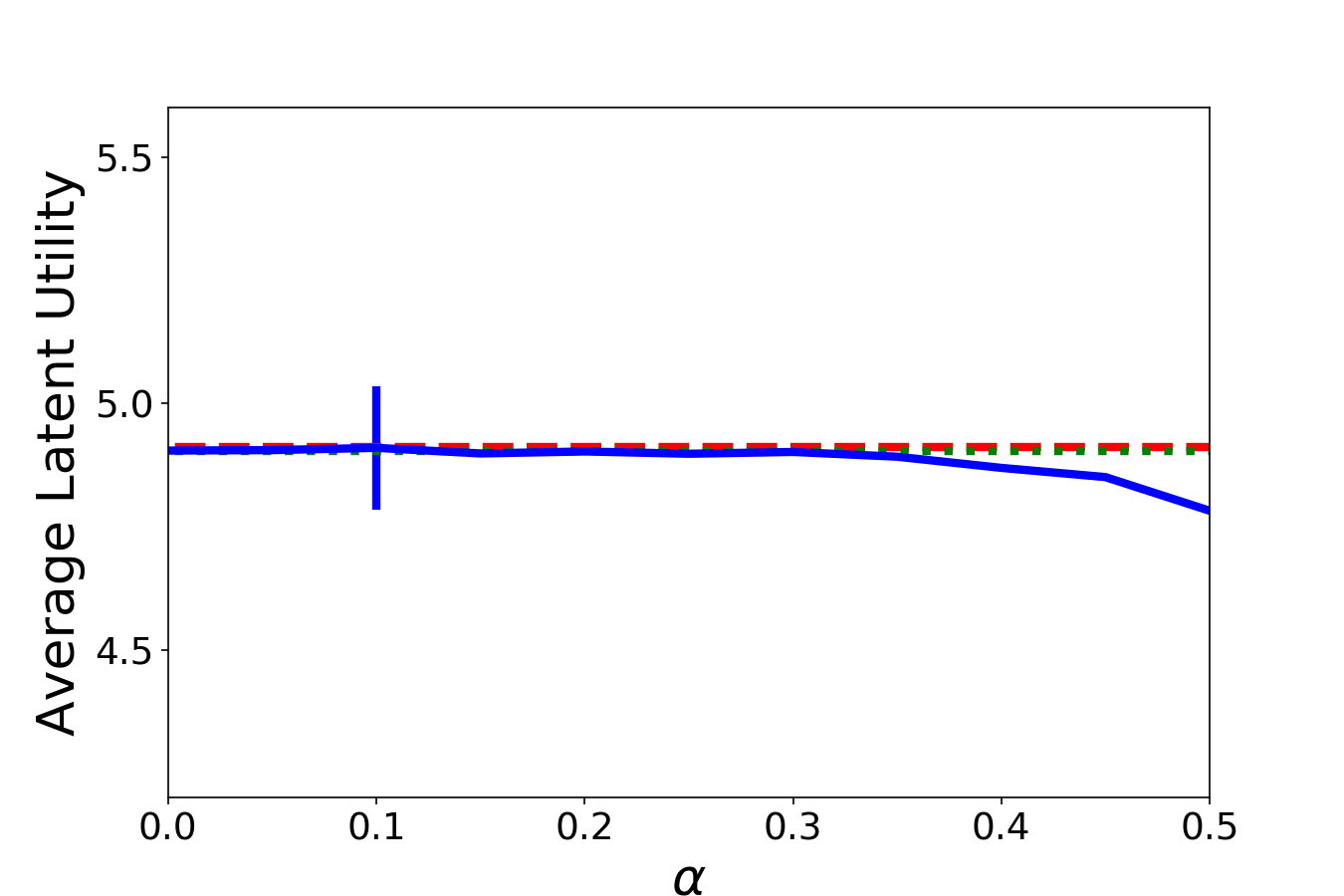}\hspace{-6mm}\\
    \vspace{-19mm}\begin{sideways}$\beta=\nfrac{1}{2}$\white{........}\end{sideways}&&&\\
    &
    \includegraphics[width=0.32\linewidth, trim={-0.1cm 0cm 0.5cm 1.7cm},clip]{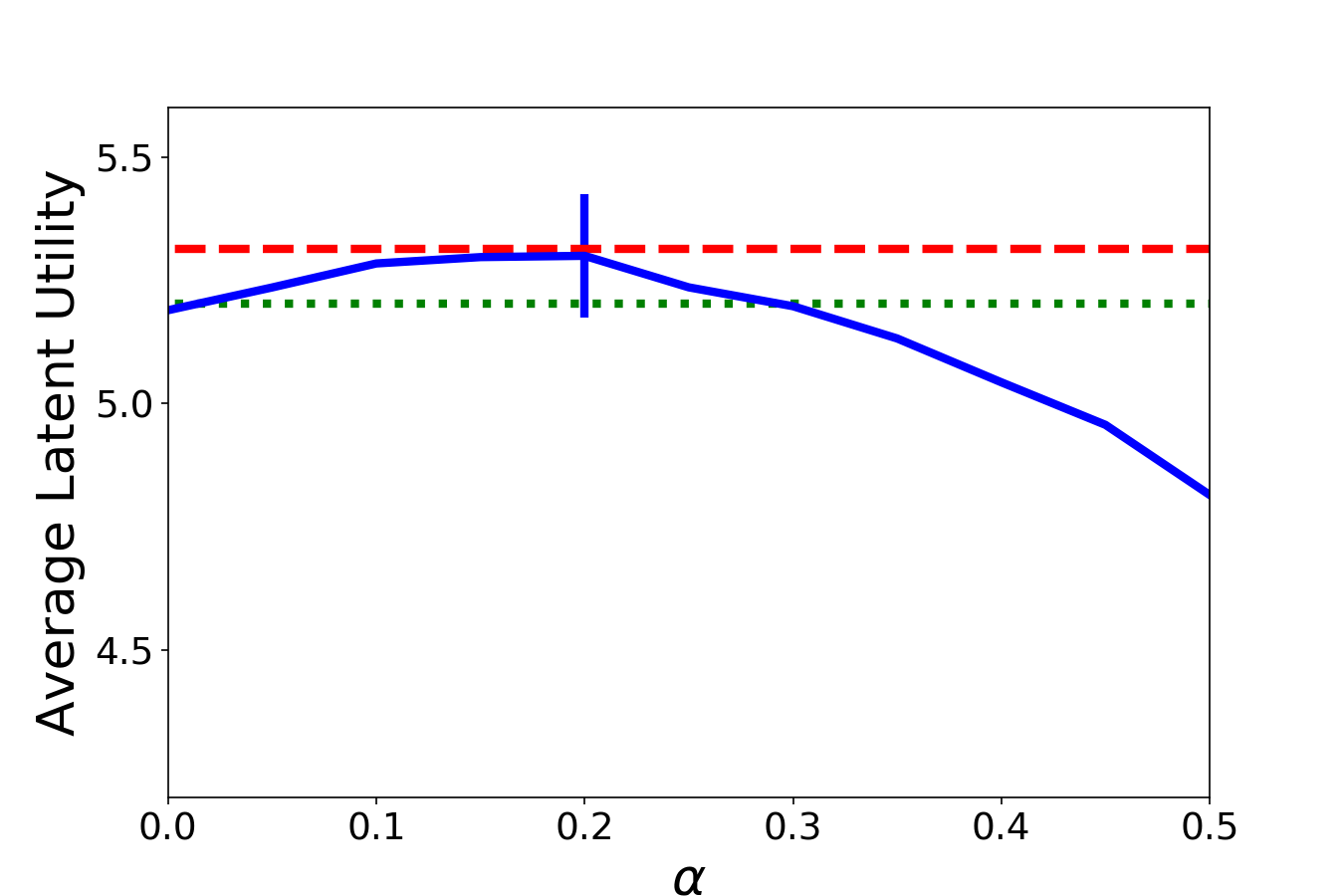}\hspace{-6mm} & %
    \includegraphics[width=0.32\linewidth, trim={-0.1cm 0cm 0.5cm 1.7cm},clip]{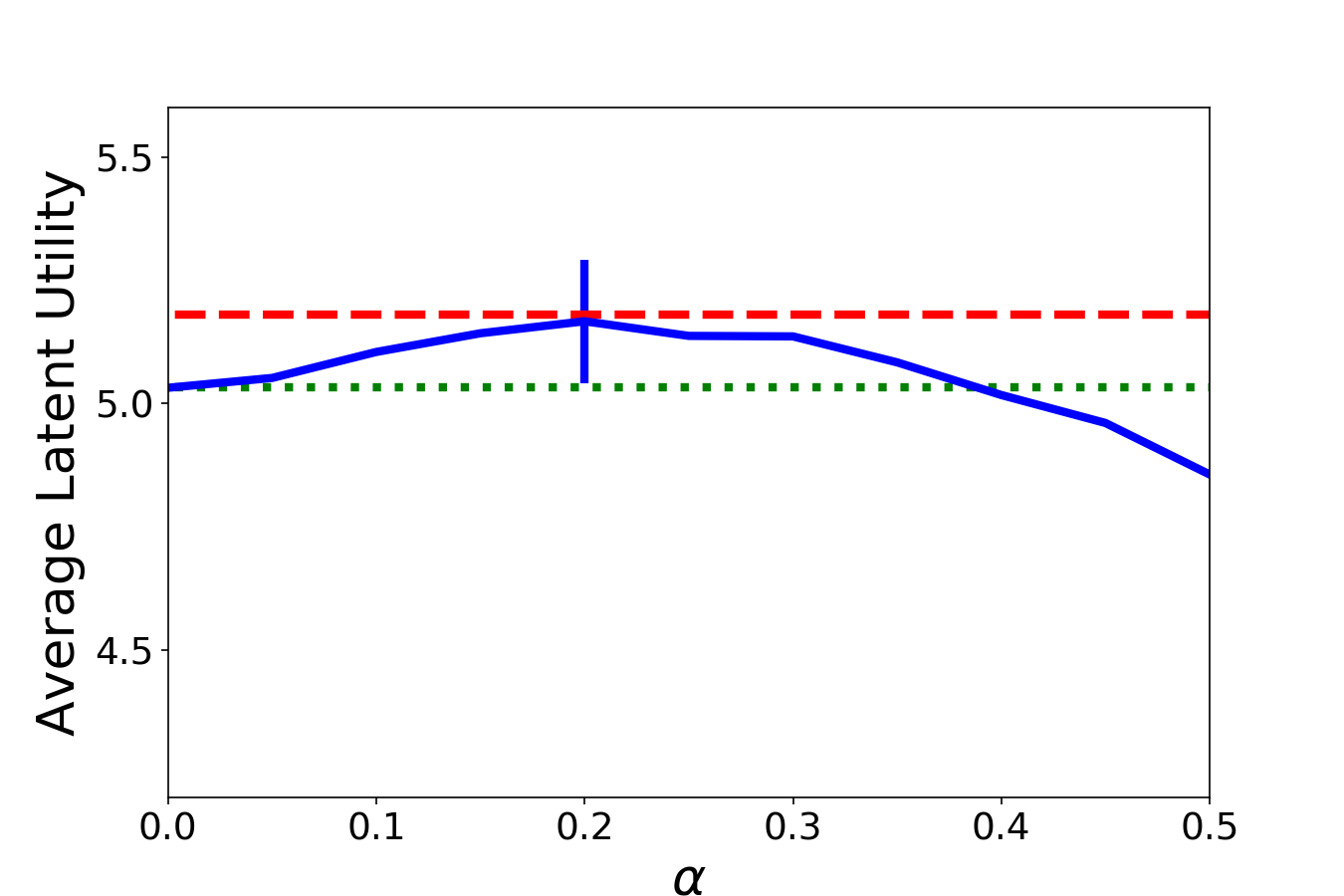}\hspace{-6mm} & %
    \includegraphics[width=0.32\linewidth, trim={-0.1cm 0cm 0.5cm 1.7cm},clip]{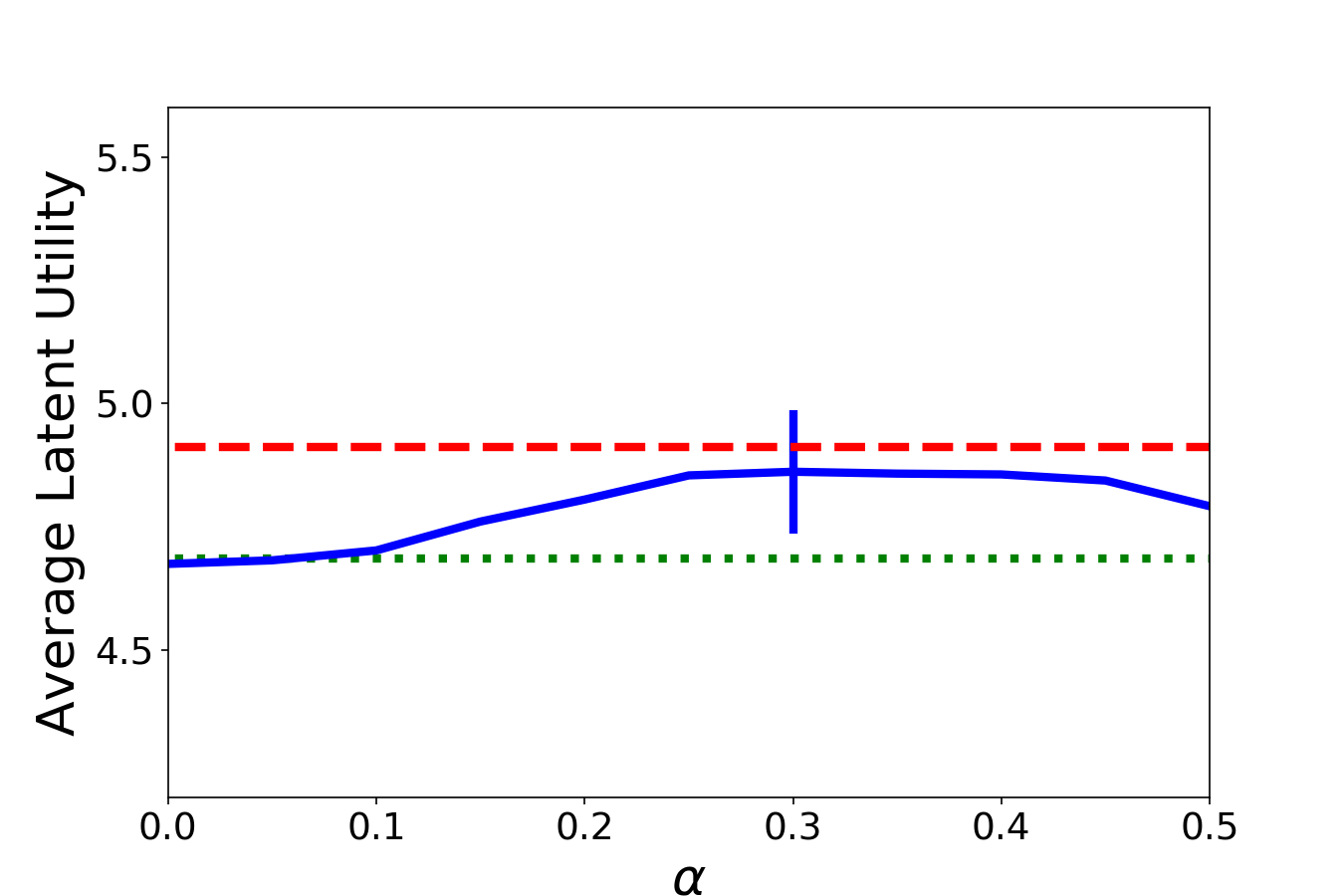}\hspace{-6mm}\\
    \vspace{-19mm}\begin{sideways}$\beta=\nfrac{1}{4}$\white{........}\end{sideways}&&&\\
    &
    \includegraphics[width=0.32\linewidth, trim={-0.1cm 0cm 0.5cm 1.7cm},clip]{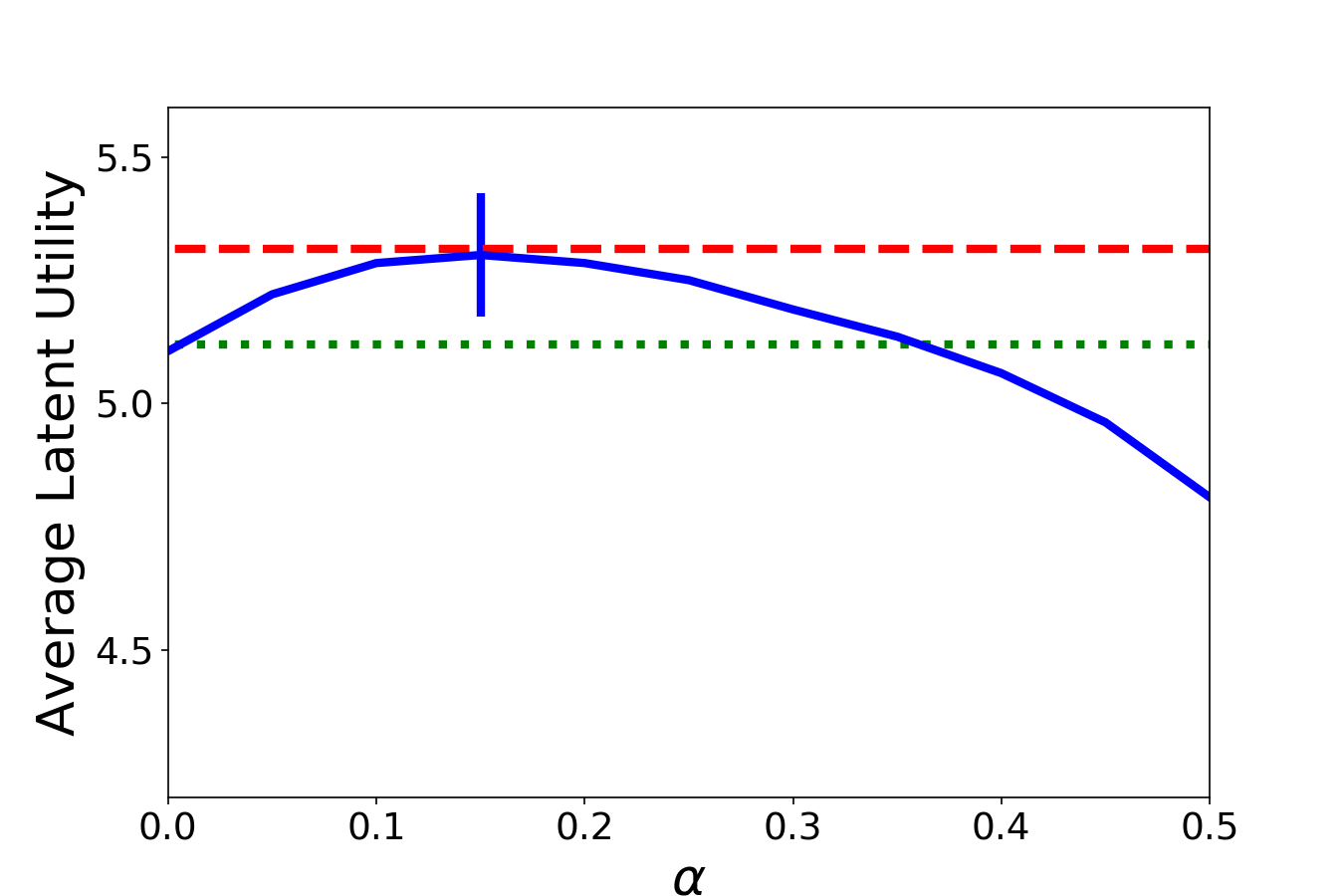}\hspace{-6mm} &
    \includegraphics[width=0.32\linewidth, trim={-0.1cm 0cm 0.5cm 1.7cm},clip]{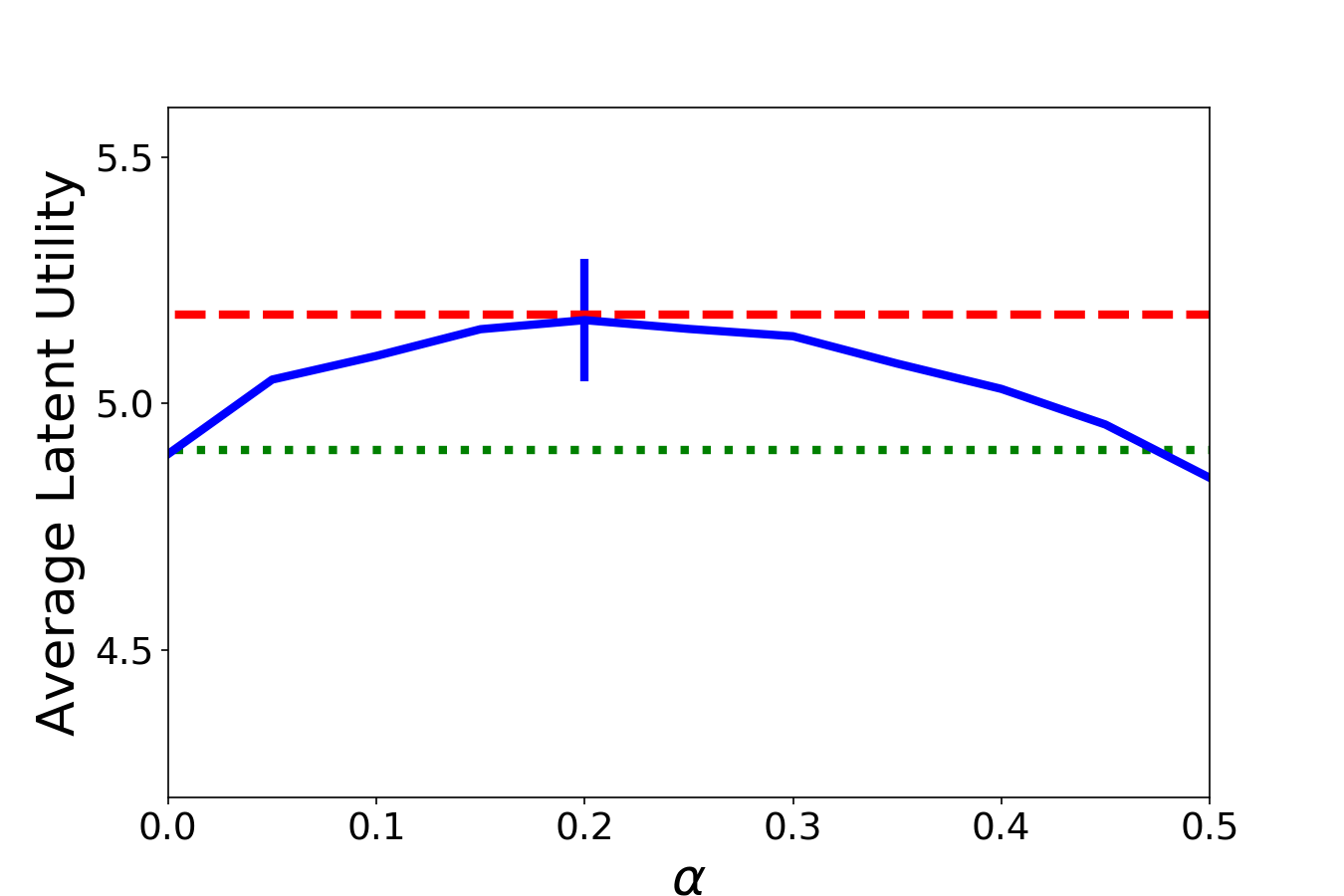}\hspace{-6mm} &
    \includegraphics[width=0.32\linewidth, trim={-0.1cm 0cm 0.5cm 1.7cm},clip]{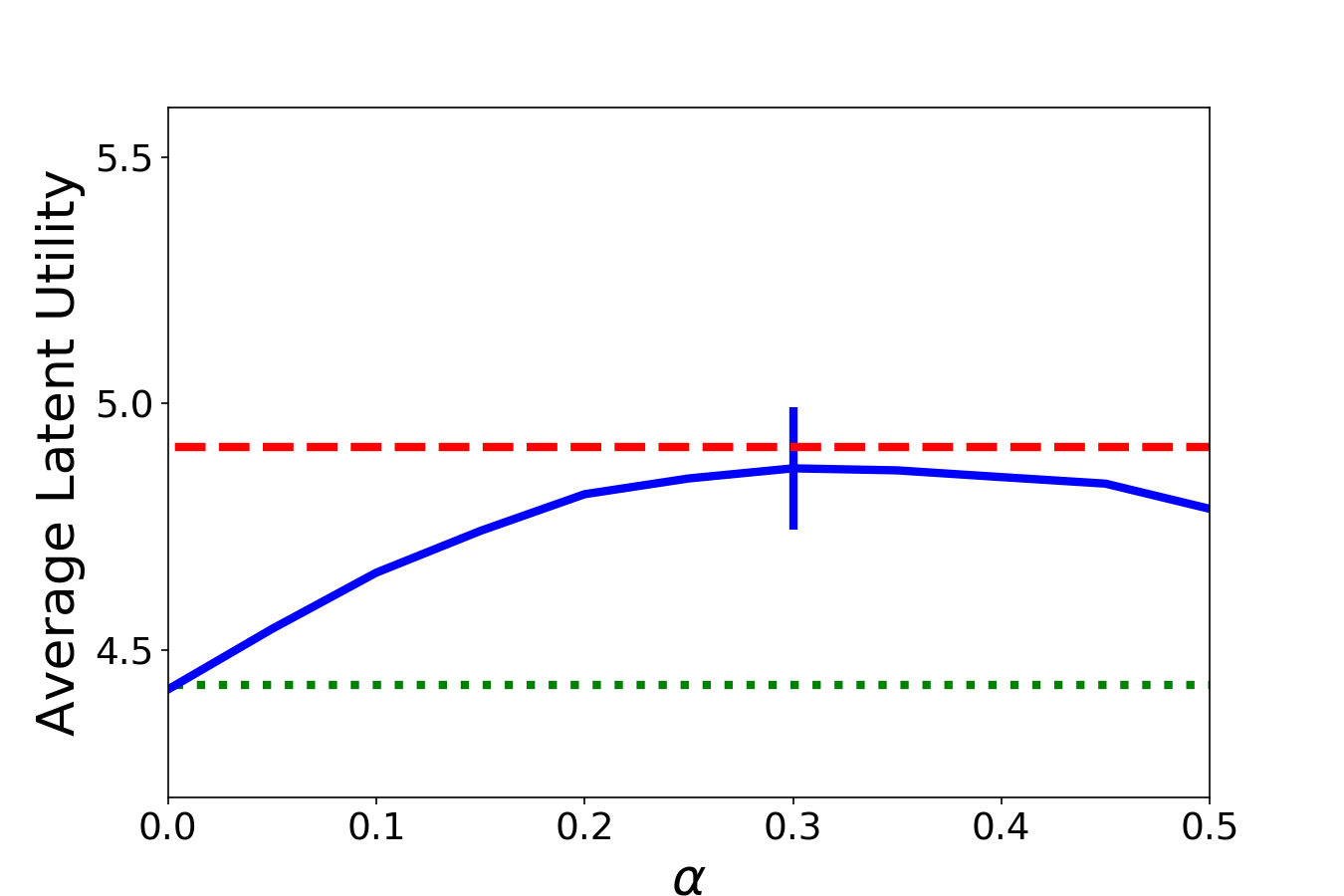}\hspace{-6mm}\\
    \hline
  \end{tabular}\vspace{2mm}
  \caption{
  {\em {Empirical results on IIT-JEE 2009 dataset (with DCG):}}
  We plot the latent utilities, $U_{\mathcal{D},v}(\alpha, \beta)$, $U_{\mathcal{D},v}(0, \beta)$, and $U_{\mathcal{D},v}(0,1)$ obtained by \cons{}, \uncons{} and \opt{} respectively
  (see Equation~\eqref{expected_utility} for the definition of $U_{\mathcal{D},v}(\cdot,\cdot)$);
  we average over values over $5\cdot10^3$ trials.
  Each plot represents an instance of the problem for a given value of implicit bias parameter $\beta$ and the ratio of the size, $m_b$, of the underprivileged group, to the size $m_a$, of the privileged group.
  The bar represents the optimal constraint $\alpha$: where we require the ranking to place at least $k\alpha$ candidates in the top $k$ positions of the ranking for every position $k$.
  We note that our constraints attain close to the optimal latent utility for the right choice of $\alpha$. More notably, they significantly outperform the unconstrained setting for a wide range of $\alpha$, lending robustness to the approach.
  Even when there is no observed bias ($\beta = 1$), adding constraints does not appear to significantly impact the utility of the ranking unless the constraints are very strong.
  }
  \label{fig_2}
\end{table*}

\subsection{Effectiveness of constraints}\label{sec_exp_implicit_bias}

We now evaluate the effectiveness of the constraints as an intervention for implicit bias~(see Section~\ref{sec:implicitbias}).
For this evaluation, we assume that the JEE scores represent the true latent utility of a candidate, and we assume these scores are distributed for male and female students according to the fitted distributions  $\mathcal{D}_a$ and $\mathcal{D}_b$ respectively.
We then consider the case where $k$ students, $m_a$ from group $G_a$ and $m_b$ from group $G_b$ apply to a job where the hiring manager has implicit bias $\beta$ against group $G_b$.
Here, \uncons{} would rank the candidates according to the biased utilities, \opt{} would rank the candidates according to their latent utilities (which is apriori impossible due to the implicit bias), and the proposed solution \cons{} would provide the optimal constrained ranking satisfying constraint parameter $\alpha$ using the fast-greedy algorithm described in~\cite{celis2018ranking}.
More formally, the rankings are
\begin{align*}
  \text{ \uncons{} } &\coloneqq \argmax\nolimits_{x}\mathcal{W}(x,v,\hat{w}),\\
  \text{ \cons{} } &\coloneqq \argmax\nolimits_{x\in \mathcal{K}(L(\alpha))}\mathcal{W}(x,v,\hat{w}),\\
  \text{ \opt{} } &\coloneqq \argmax\nolimits_{x}\mathcal{W}(x,v,w).
\end{align*}

We report the average latent utilities $U_{\mathcal{D},v}(\alpha, \beta)$, $U_{\mathcal{D},v}(0, \beta)$, and $U_{\mathcal{D},v}(0,1)$ obtained by \cons{}, \uncons{}, and \opt{} respectively
(see Equation~\eqref{expected_utility} for the definition of $U_{\mathcal{D},v}(\cdot,\cdot)$).

\subsubsection{Parameters}
We let {$m=m_a+m_b \coloneqq 1000$}, $k\coloneqq100$, $v_k\coloneqq \nfrac{1}{\log(k+1)}$ and vary the implicit bias parameter $\beta\in \{1/4, 1/2, 1\}$ and {$m_b \in \{\nfrac{m}{2}, \nfrac{m}{3}, \nfrac{m}{4}\}.$}
In the JEE dataset, $m_b\approx \nfrac{m}{4}$ is a representative number.
The position discount, $v_k=\nfrac{1}{\log(k+1)}$ corresponds to DCG~\cite{DCG,kanoulas2009empirical,wang2013theoretical} which is popularly used in practice.
We vary $\alpha$ from 0 to 0.5, i.e., from no constraint to the case where half of the candidates in the ranking must be from group $G_b$.
We report the average latent utilities $U_{\mathcal{D},v}(\alpha, \beta)$, $U_{\mathcal{D},v}(0, \beta)$, and $U_{\mathcal{D},v}(0,1)$; we take the average over $5000$ trials in Table~\ref{fig_2}.

\subsubsection{Results}\label{results_of_jee_exp}
We observe that the constraint intervention can significantly increase the latent utility of the ranking, even when,  $\mathcal{D}_a$ and $\mathcal{D}_b$ are non-uniform, and are not exactly the same.
The extent of the improvement depends both on the degree of implicit bias $\beta$, and the fraction of the underrepresented group $m_a/m$ that appears in the dataset.
As expected from Theorem~\ref{thm_const_increase_utility}, there exists an $\alpha$ for which the latent utility from the optimal ranking can be attained.
More generally, we observe that the constraints increase the latent utility when $\alpha$ is such that it ensures no more than proportional representation -- i.e., when the constraint ensures that the ranking reflects the group percentages in the underlying population, even when the bias parameter is small.

We conduct similar simulations without a discounting factor -- i.e., with just selection with no ranking, and observe similar results in Section~\ref{without_dcg_results}.
We also observe similar results on the Semantic Scholar Research Corpus, where the distributions are heavy tailed and hence constitute a very different class of utilities (see Section~\ref{sec_extnd_empirical}).
Overall, these findings suggest that the intervention is robust and across different latent utility distributions, population sizes, and extent of bias, and that the exact parameters need not be known to develop successful interventions.

\subsection{IIT supernumerary seats for women}\label{sec_supernumerary}

If we assume that the scores of the candidates are a true measure of their academic ``potential'', then any scheme which increases the number of underrepresented candidates admitted is bound to decrease the average ``potential'' of candidates at the institute.
However, among candidates of equal ``potential'', those from an underprivileged group are known to perform poorer on standardized tests~\cite{walton2009latent}.
In India, fewer girls than boys attend primary school~\cite{alderman1998gender,censusinfo}, many of whom are forced to drop-out of schools to  help with work at home or get married~\cite{girl_drop_out}.
Therefore, we expect a female student who has the same score as a male student, say in IIT-JEE, to perform better than the boy if admitted.
This is a societal bias, which, while different in nature than implicit bias, is another reason through which utilities can be systematically wrong against a particular group. Hence, the constraint approach presented is equally applicable as we illustrate in this section.
In effect, it means that the scores are in fact biased, and the true latent utility of a candidate from $G_b$ is in fact larger than what is reflected from their score.

To account for this and improve the representation of women in IITs, in 2018, additional ``supernumerary'' seats for women were introduced.
This increases the capacity of all majors in all IITs by creating additional seats which are reserved for women, such that women compose at least 14\% of each major, without decreasing the number of non-female candidates admitted.
More formally, if the original capacity of a major at an IIT was $C$, and the average number of females admitted in  this major in the past few years was $n_f$, then the scheme created $x$ additional seats such that $n_f+x\coloneqq 0.14 (C+x)$.
Where $(n_f+x)$ seats are reserved for females and the other $(C-n_f)$ seats are open to candidates of all genders.
Eligible women are granted admission on reserved seat first, and only when all reserved seats are filled a women admitted on a gender neutral seat~\cite{baswana2015joint}.

\subsubsection{Setup}
We assume the number of slots available in a given year for the IITs is $n\coloneqq 10^4$.\footnote{This number is close to the 9311 and 9576 students admitted into IITs in 2011 and 2012 which had the number of IITs as 2009.}
We assume that the true latent utility of a candidate from group $G_b$ is
given by a shift $\gamma>1$, such that if they attain a score $s_f$, then a true score would be $s_f^\prime \coloneqq (s_f+105)\cdot \gamma - 105$.\footnote{The shifts by 105 are to account fo the score range which can be negative.}
In our simulations, we use $\gamma\approx 1.076$, which results in a shifted distribution $\mathcal{D}^\prime_b$ with the same mean as $\mathcal{D}_a$.

Let the supernumerary scheme, {\sc Sup}$(\alpha)$, admit $n_{\text{\sc Sup}}(\alpha)\geq n$ candidates.
We know that {\sc Sup} always admits more candidates than our constraints approach, \cons{}$(\alpha)$, and the unconstrained approach \uncons{}, both of which choose $n$ candidates.
As such, it would be unfair to compare the average utility of the candidate selected by {\sc Sup} with that of \cons{}$(\alpha)$ or \uncons{}.
We define two more rankings $\overline{\text{\cons{}}}(\alpha)$ and $\overline{\text{\uncons{}}}$, which given an $\alpha$ select $n_{\text{\sc Sup}}(\alpha)$ candidates and are otherwise equivalent to {\cons{}$(\alpha)$} and \uncons{}.

We allocate $n$ seats using the simple constraints scheme, \cons{} and \uncons{} and $n_{\text{\sc Sup}}(\alpha)\geq n$ seats using {\sc Sup}$(\alpha)$, $\overline{\text{\cons{}}}(\alpha)$, and $\overline{\text{\uncons{}}}$.
and compare the average latent utilities of the candidates admitted by the scheme.
Here, we define the average latent utility from scheme ${\sc A}$ which admits $n({\sc A})$ candidates as
\begin{align*}
  &U({\sc A}) \coloneqq \frac{1}{n({\sc A})}\cdot \eE\nolimits_{w\gets \mathcal{D}_a, \mathcal{D}^\prime_b}[ \mathcal{W}(\tilde{x}, v, w)]\\
  &\text{where } \tilde{x}\coloneqq \max_{x\text{ satisfies }{\sc A}}  \mathcal{W}(\tilde{x}, v, \hat{w}).
\end{align*}
We vary $\alpha$ from {0.1 to 0.25}, i.e., from the fraction of women admitted if all candidates are admitted on basis of their scores, to the value which corresponds to proportional the representation based on the number of candidates appearing in IIT-JEE 2009.
We report the average latent utilities in Figure~\ref{analysis_of_jee_scheme}.
\begin{figure}[h!]
  \begin{center}
    \includegraphics[width=0.5\linewidth]{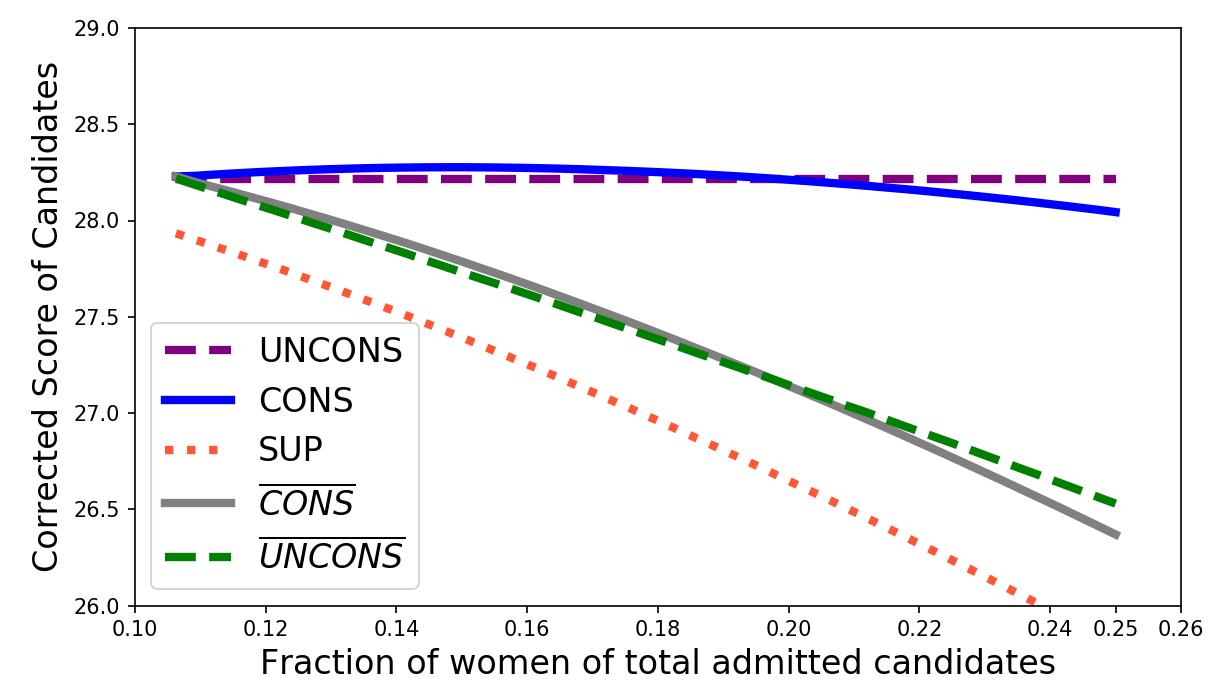}
    \hspace{2mm}
  \end{center}
  \vspace{-4mm}
  \caption{
  {\em Empirical results studying the effect of supernumerary seats for women:}
  We plot the average latent utilities of the {ranking schemes} we consider.
  The $y$-axis represents the utility per candidate, and
  the $x$-axis represents the lower-bound constraint $\alpha$. Our constrained interventions outperform the unconstrained variants up to approximately $20\%$ female seats, and always outperforms the existing supernumerary approach used in practice.
  }
  \label{analysis_of_jee_scheme}
\end{figure}
\subsubsection{Results}

We observe that the {\sc Sup}$(\alpha)$ always has an lower average utility that $\overline{\text\cons{}}(\alpha)$ scheme.
We note that the {\sc Sup}$(\alpha)$ and the $\overline{\text\cons{}}(\alpha)$ always select the same number of candidates.
The difference is that {\sc Sup}$(\alpha)$ could place all the underprivileged candidates at the end of the ranking, while $\overline{\text\cons{}}(\alpha)$ places at least $k\cdot \alpha$ underprivileged  candidates every $k$ positions.
We find that for any $\alpha>0.1$, {\sc Sup}$(\alpha)$ decreases the average {latent utility} of the admitted candidates.

Finally, we observe that for a range of $\alpha$ (from 11\%  to  19.4\%), \cons{} increases the average latent utility of the admitted candidates over \uncons{} and $\overline{\text{\cons{}} }(\alpha)$ increases the latent utility over  $\overline{\text{\uncons{}} }$,  i.e., for all $\alpha\in [0.11,0.194]$,
\begin{align*}
  \nfrac{1}{n}\cdot U(\text{\cons{}}(\alpha))>\nfrac{1}{n}\cdot U(\text{\uncons{}}),\\
  \nfrac{1}{n_{\text{\sc Sup}}(\alpha)}\cdot U(\overline{\text{\cons{}}}(\alpha))>\nfrac{1}{n}\cdot U(\overline{\text{\uncons{}}}).
\end{align*}
The optimal constraint for \cons{}$(\alpha)$ is $\argmax_{\alpha}U(\text{\cons{}}(\alpha))=0.15$.
Intuitively,  \cons{}$(\alpha)$ can increase the average utility by swapping a male candidate $i\in G_a$ and female candidate $j\in G_b$, such that $\hat{w}_j < w_i < w_j$.
Note that in \cons{}$(\alpha)$ additional candidates selected from $G_b$ compete against the lowest scoring candidates from $G_b$ in the ranking, instead of the average utility $U(\text{\cons{}}(\alpha))$ as in {\sc Sup}$(\alpha)$.

\begin{remark}
  Since  {\sc Sup}$(\alpha)$ was not designed to optimize the average utility, it is not surprising that $\overline{\text\cons{}}(\alpha)$ outperforms {\sc Sup}$(\alpha)$ in our experiment.
  However, the goal of this experiment is to study the effect our constraints approach against systematic biases different from implicit bias, by qualitatively comparing them an existing scheme ({\sc Sup}$(\alpha)$).
\end{remark}

\section{Discussion and limitations}\label{sec_limitations}
One could also consider other approaches to mitigate implicit bias.
For instance, in a setting where an interviewer ranks the candidates, we could ask interviewers to self-correct (rescale) their implicit bias.
However, anti-bias training, which aims to correct peoples' implicit bias, has been shown to have limited efficacy~\cite{noon2018pointless}.
Thus, even with training, we do not have a guarantee that the interviewer can self-correct effectively.
Adding interventions like the ones we consider do not require individuals to self-correct their perceptions.

Instead, we could ask interviewers to report their observed utilities and later rescale them.
However, we may not have an accurate estimate of $\beta$. As the interventions we consider are independent of $\beta$ (in Theorem~\ref{thm_const_increase_utility}), they can be applied in such a setting and would still recover the optimal utility (see Remark~\ref{rem_constraint_ind_of_beta}).
Furthermore, if interviewers are {\em explicitly} biased, they can give arbitrarily low scores to one group of candidates;
this would make any rescaling attempt insufficient.
By instead requiring a fixed number of candidates from each group, the interventions we consider are also robust against explicit bias, and perhaps this is why simple versions have been used in practice (e.g., Rooney rule~\cite{facebook_rr,  cavicchia-implicit-bias-Rooney}).

We crucially note that any such intervention will only have a positive end effect if the goal is sincere; a hiring manager who is biased against a group of people can simply not hire a person from that group, regardless of what ranking they are given or what representation is enforced throughout the interview process.
%
%
%
{Furthermore, while interventions such as the ones we describe here are a robust approach to correct for certain kinds of biases, they are only a single step in the process and alone cannot suffice.}
{It would be important to evaluate this approach as part of the larger ecosystem in which it is used in order to adequately measure its efficacy in practice.}

\section{Conclusion and future directions}
We consider a type of constraint-based interventions for re-ordering biased rankings derived from utilities with systematic biases against a particular socially salient group.
The goal of these interventions is to recover, to the best of our ability, the unbiased ranking that would have resulted from the true latent (unbiased) utilities.
We consider a theoretical model of implicit bias, and study the effect of such interventions on rankings in this setting.

We show that this family of constraint-based interventions are sufficient to mitigate implicit bias under this model; in other words, a ranking with the optimal latent utility can be recovered using this kind of intervention.
However, the optimal parameters of the intervention depend on  the specific utilities.
Towards understanding this further, we make a natural assumption that the utilities are drawn from a fixed distribution, and show that simple constraints recover the optimal utility under this bias model.
We focus on specific distributions, but believe that similar theoretical bounds would exist for other bounded distributions as discussed in Section~\ref{sec_utility_bound}.
{Rigorously analyzing necessary properties of distributions for which such bounds hold would be an interesting direction for future work.}

In this work we pose the problem in terms of implicit biases.
However, the source of the bias is not important; our results hold as long as there is systematic bias against one group as captured by our model in Section~\ref{sec:implicitbias}.
We briefly discuss a different type of social bias in Section~\ref{sec_supernumerary}.
Further, we note that this class of interventions is also robust against {\em explicit} bias (see Section~\ref{sec_limitations}).
Their robustness
is further supported by our empirical findings, in which we find that there exist optimal constraints $\alpha^\star$ for which the optimal latent utility is almost attained in expectation.
Importantly, we find that  the interventions are near-optimal even when the distributions are unbounded (e.g., lognormal distributions; see Section~\ref{sec_extnd_empirical}).
Further, the interventions remain near-optimal even when the latent utility distribution of the underprivileged group is similar, but not identical, to that of the privileged group (see Section~\ref{sec_exp_implicit_bias}).
More generally, we observe that the intervention improves the latent utility of the ranking for a {wide range} of $\alpha$, with the highest improvement roughly centered around proportional representation.
This gives an interesting rule of thumb, but more importantly shows the robustness of the method; without knowing the exact optimal $\alpha$, one can still improve the ranking's latent utility significantly.
From these observations, we expect this class of constraint-based intervention to be successful for a wide class of settings; {exploring its limitations and developing clear guidelines on when and how to use the interventions in a particular use case would be an important avenue for further work.}

Lastly, we note that while we phrase the majority of our results in light of rankings, they also have implications for the subset selection problem (where a group must be chosen, but need not be ordered) by taking the discounting factor $v_k=1$ for all positions $k\in [n]$. %
This, in effect, eliminates the importance of the order and the total utility depends only on the set of people or items selected.
In particular, our results answers a question asked in \cite{KleinbergR18} on the efficacy of the $\ell$-th order Rooney Rule when the distribution of utilities is uniform.
We further report empirical results for subset selection, which follow the setup and conclusions discussed in Section~\ref{without_dcg_results}.

In summary, simple constraint-based interventions appear to be highly effective and robust for a wide variety of biases, distributions and settings; {it is our hope that they be adopted and studied further as a simple yet powerful tool to help combat biases in a variety of applications.}

\section{Extended empirical results}
\subsection{Additional results with JEE scores}~\label{without_dcg_results}
We present our empirical results for the IIT-JEE dataset without position-discount as described in Section~\ref{sec_exp_implicit_bias} in Table~\ref{table_2}.

\subsection{Empirical results with the Semantic Scholar Open Research Corpus}\label{sec_extnd_empirical}
\begin{figure*}[h!]
  \begin{center}
    \subfigure[
    Effect of the threshold for predicting names on the total number of male and female authors.
    The linechart shows the number of authors categorized as men or women on setting different probability thresholds for categorizing first names.
    \label{fig_name_thresh}
    ]
    {
    \includegraphics[width=0.45\linewidth]{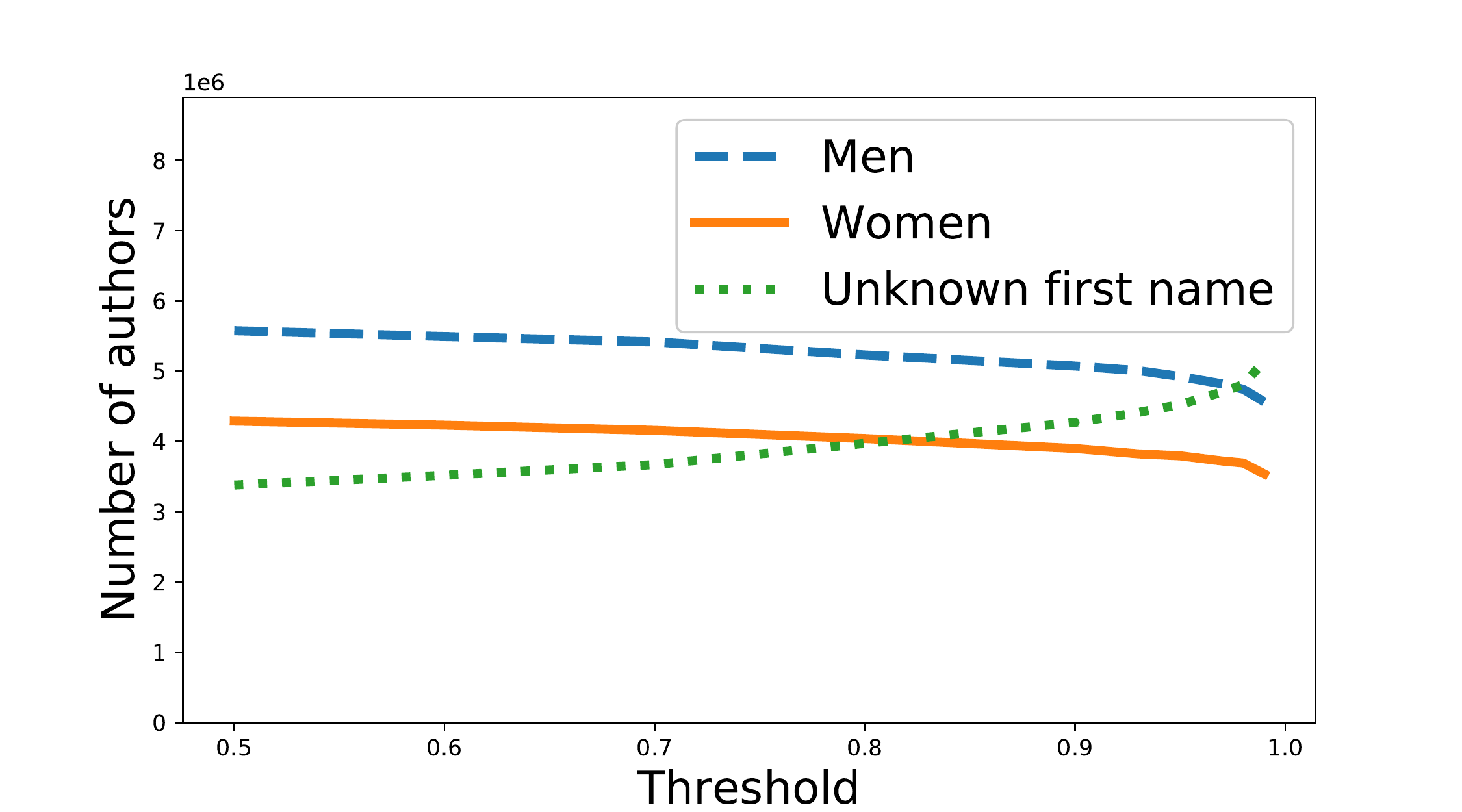}
    }
    \hspace{2mm}
    \subfigure[
    The linechart shows the total number of authors with their first publication before a particular year.
    We note that the dataset largely contains authors who published their first paper after 2000.
    \label{fig_num_auth_year}
    ]
    {
    \includegraphics[width=0.35\linewidth]{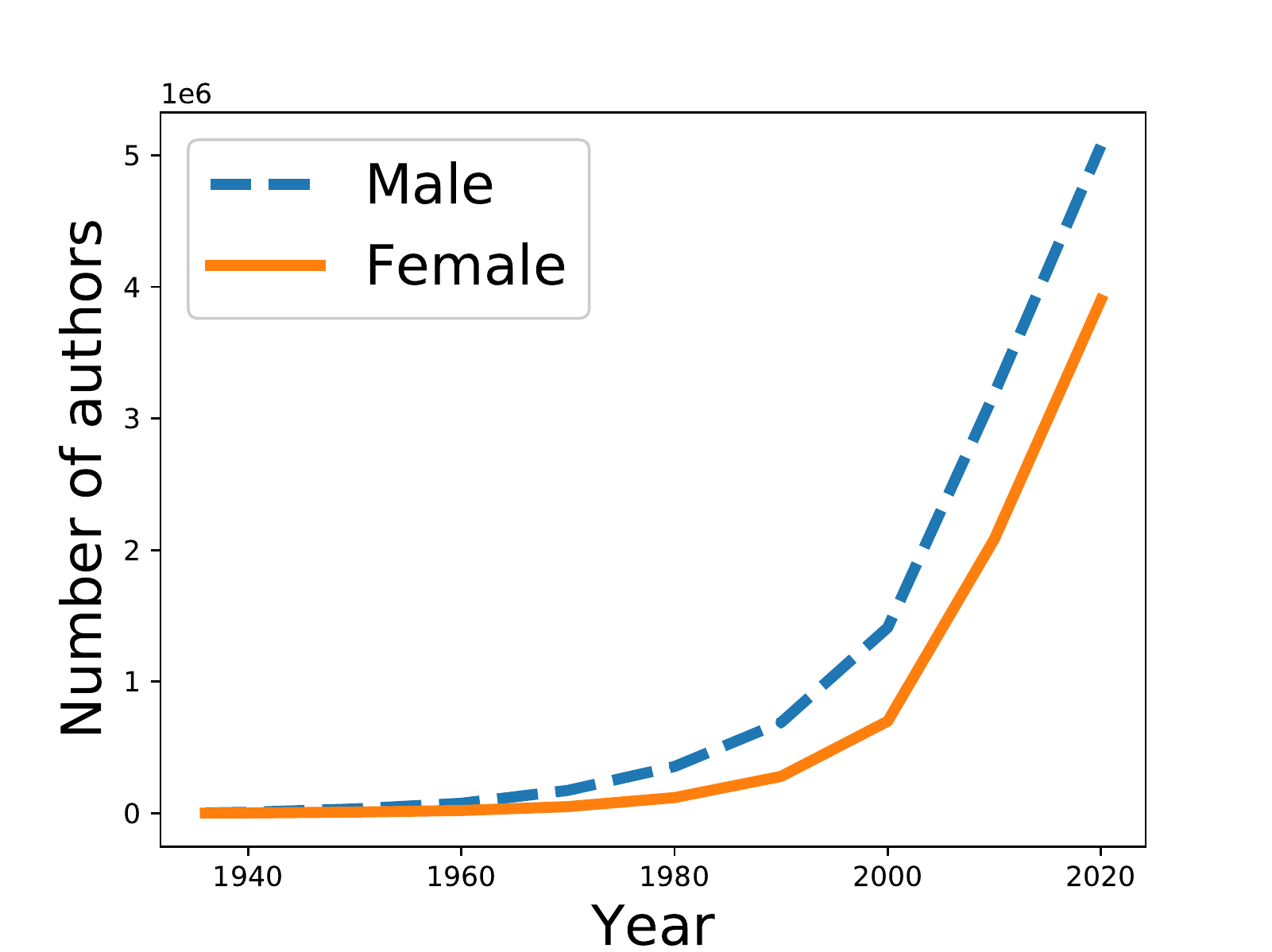}
    }
    \\
    \subfigure[
    Distribution of total citations of men and women
    who published their first paper after 1980.
    \label{fig_4_3}
    ]
    {
    \includegraphics[width=0.5\linewidth]{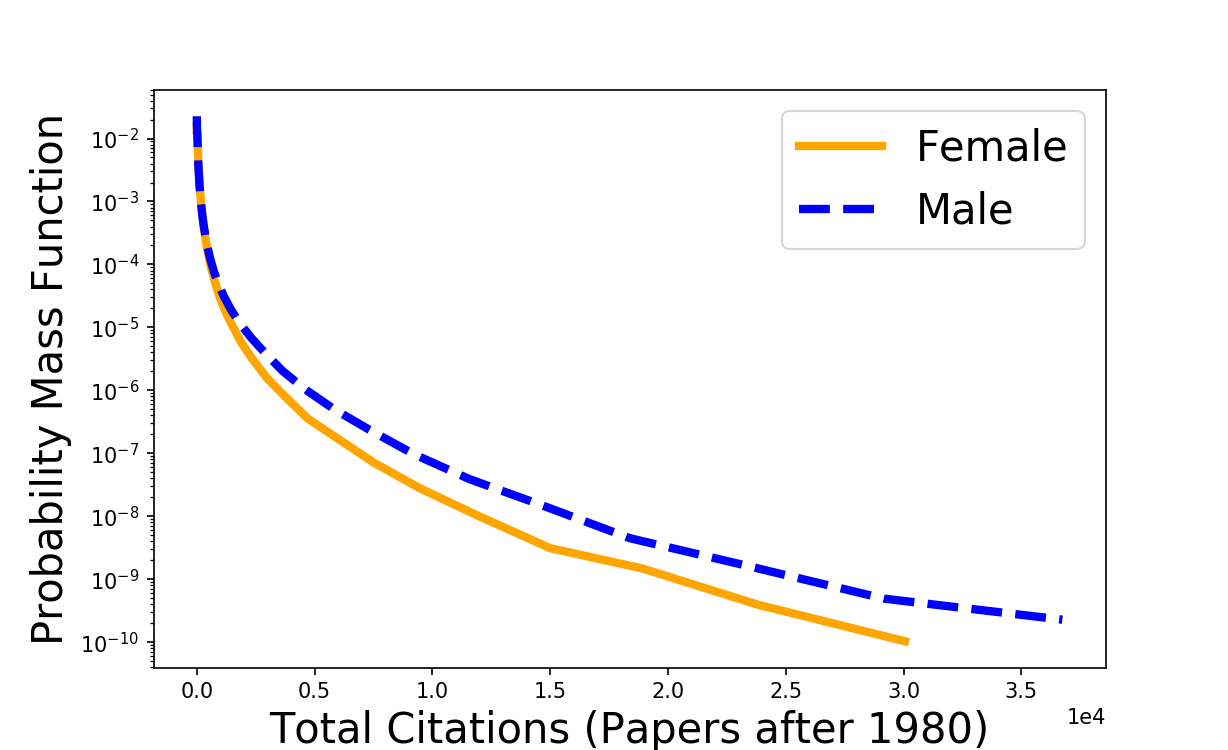}
    }
  \end{center}
  \caption{
  {\em Statistics from the Semantic Scholar Open Research Corpus.}
  }
\end{figure*}
\subsubsection{Dataset}
The {\em Semantic Scholar Open Research Corpus} contains meta-data of 46,947,044 published research papers in Computer Science, Neuroscience, and Bio-medicine from 1936 to 2019 on Semantic Scholar.
The meta-data for each paper includes the list of authors of the paper, the year of publication, list of papers citing it, and journal of publication, along with other details.

\subsubsection{Cleaning dataset and predicting author gender}
While counting the total citations we remove all papers which did not have the year of publishing, this retains 98.14\% of all the papers.
The Semantic Scholar dataset did not contain the author's gender, which we predict from their first name using a publicly available dataset from the the US Social Security Administration~\cite{first_name_dataset}.
We predict an author's gender from their first name using a publicly available dataset of names from the the US Social Security Administration~\cite{first_name_dataset}.
This dataset contains first names and gender of all people born from 1890 to 2018 registered with the social security administration.
We first remove all the authors whose first name is shorter than 2 characters, as these are likely to be abbreviations.
This retains more that 75\% of the 17,805,885 unique authors.
Then, we categorize an author as female if more than 90\% people with the same first name were females, and as a male if more than 90\% people with the same first name were males, otherwise we drop the author.
This retains 67.75\% of the remaining authors, and results in 3900934 women and 5074426 men (43.46\% females).
We plot the trade off between the total authors retained and the threshold used while categorizing in Figure~\ref{fig_name_thresh}.

\subsubsection{Counting the total citations}
We approximate the distribution of total citations of men and women, who published their first paper after 1980.
This removes less than 5\% of the authors.
We present the distribution of Total Citations of men and women in Figure~\ref{fig_4_3}.
We found that the log-normal distribution gave a good fit for the distribution of total citations.

\subsubsection{Parameters}
We let {$m=m_a+m_b \coloneqq 1000$}, $k\coloneqq100$, $v_k=1$ and vary the implicit bias parameter $\beta\in \{1/4, 1/2, 1\}$ and {$m_b \in \{\nfrac{m}{2}, \nfrac{m}{3}, \nfrac{m}{4}\}.$}
We vary $\alpha$ from 0 to 0.5, i.e., from no constraint to the case where half of the candidates in the ranking must be from group $G_b$.
We report the average latent utilities $U_{\mathcal{D},v}(\alpha, \beta)$, $U_{\mathcal{D},v}(0, \beta)$, and $U_{\mathcal{D},v}(0,1)$; we take the average over $5000$ trials in Table~\ref{table_3}.

\subsubsection{Results}
We observe similar results as with the JEE scores dataset.
We refer the reader to Section~\ref{results_of_jee_exp} for a discussion of our observations.

\begin{table*}[t!]
  \begin{tabular}{|c|ccc|}
    \hline
    & $\nfrac{m_b}{m}=\nfrac{1}{4}$  & $\nfrac{m_b}{m}=\nfrac{1}{3}$  & $\nfrac{m_b}{m}=\nfrac{1}{2}$  \\
    \hline
    &&&\\
    &&&\\
    \vspace{-20mm}\begin{sideways}$\beta=1$\white{........}\end{sideways}&&&\\
    &
    \includegraphics[width=0.32\linewidth, trim={-0.1cm 0cm 0.5cm 1.7cm},clip]{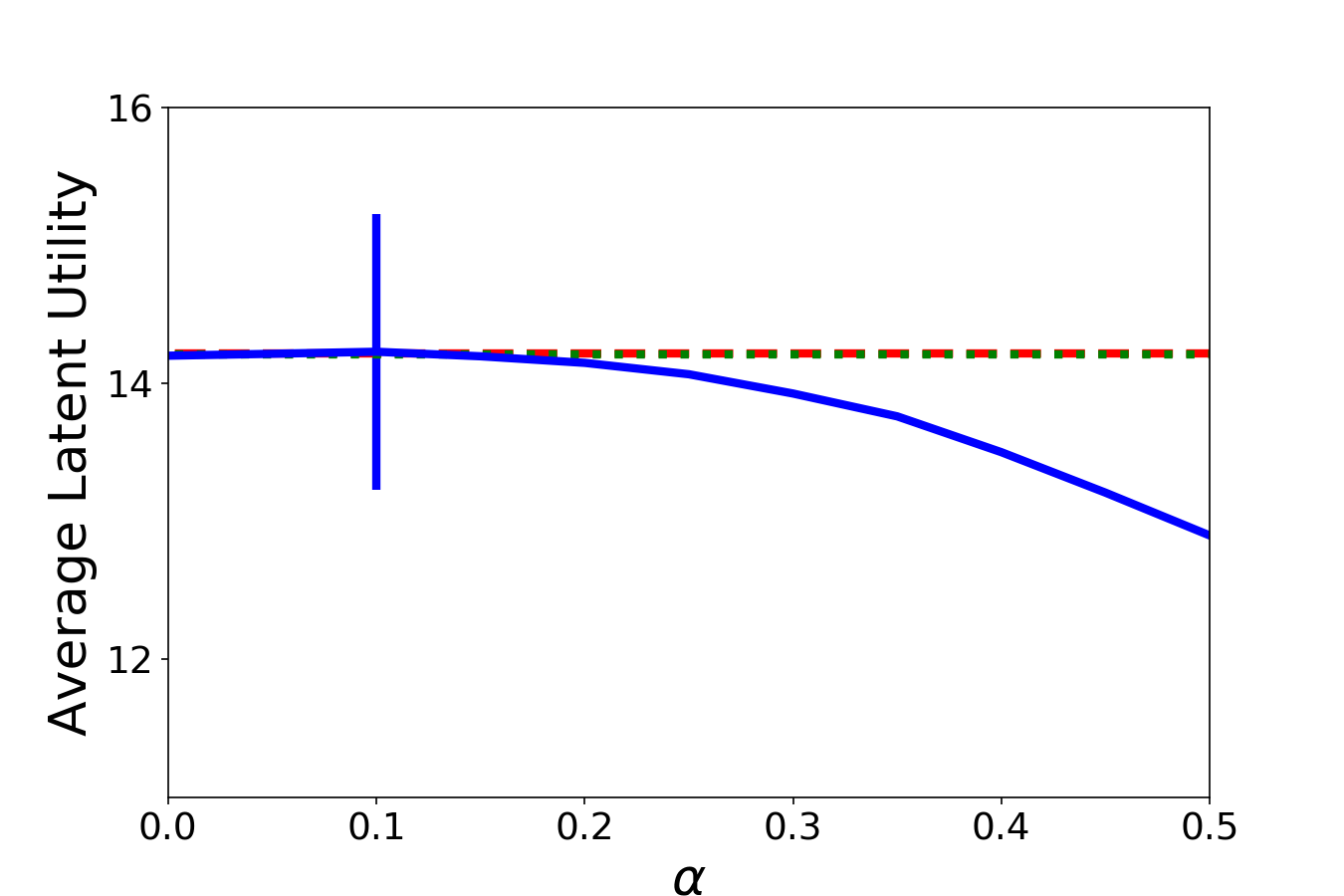}\hspace{-6mm} &
    \includegraphics[width=0.32\linewidth, trim={-0.1cm 0cm 0.5cm 1.7cm},clip]{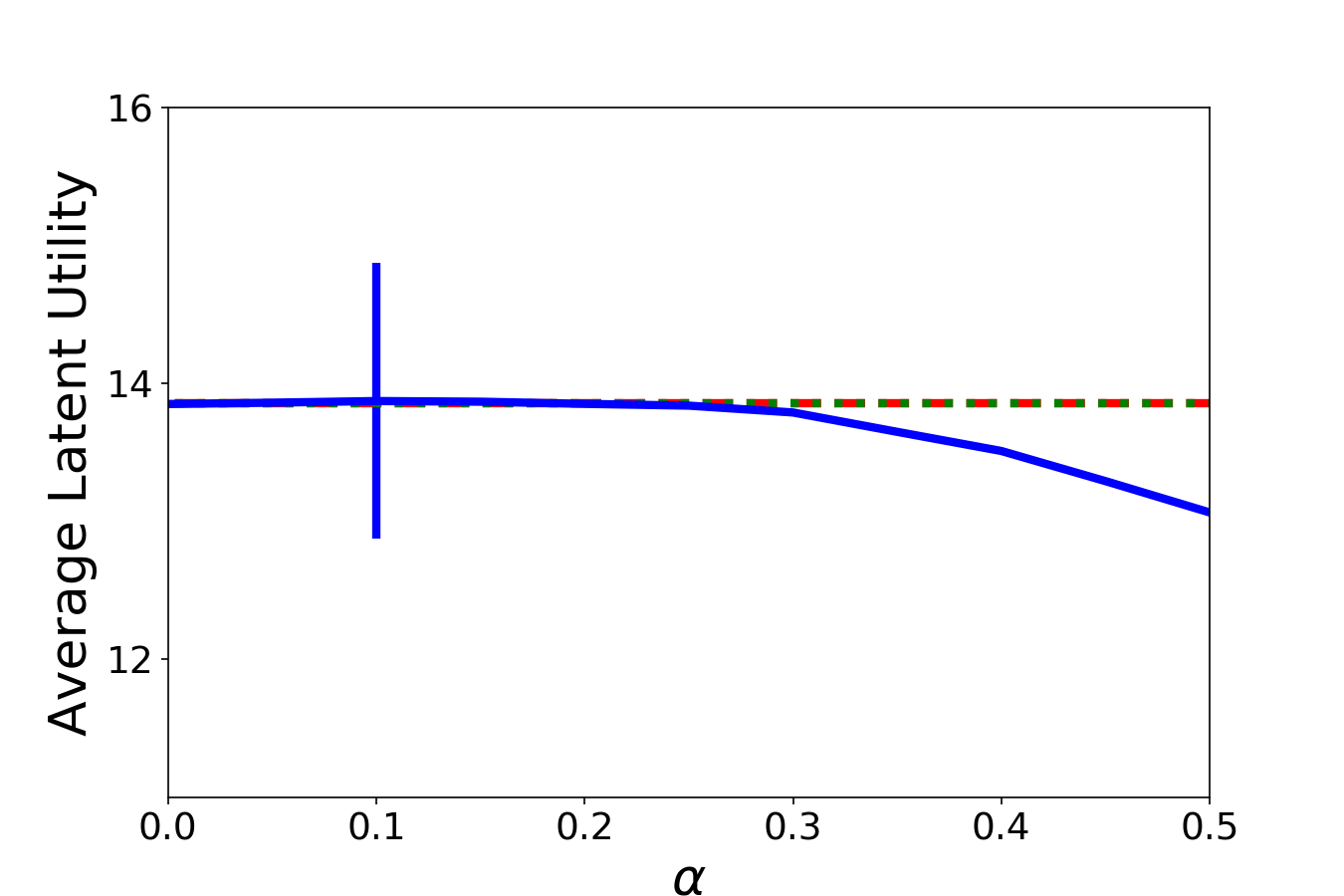}\hspace{-6mm} &
    \includegraphics[width=0.32\linewidth, trim={-0.1cm 0cm 0.5cm 1.7cm},clip]{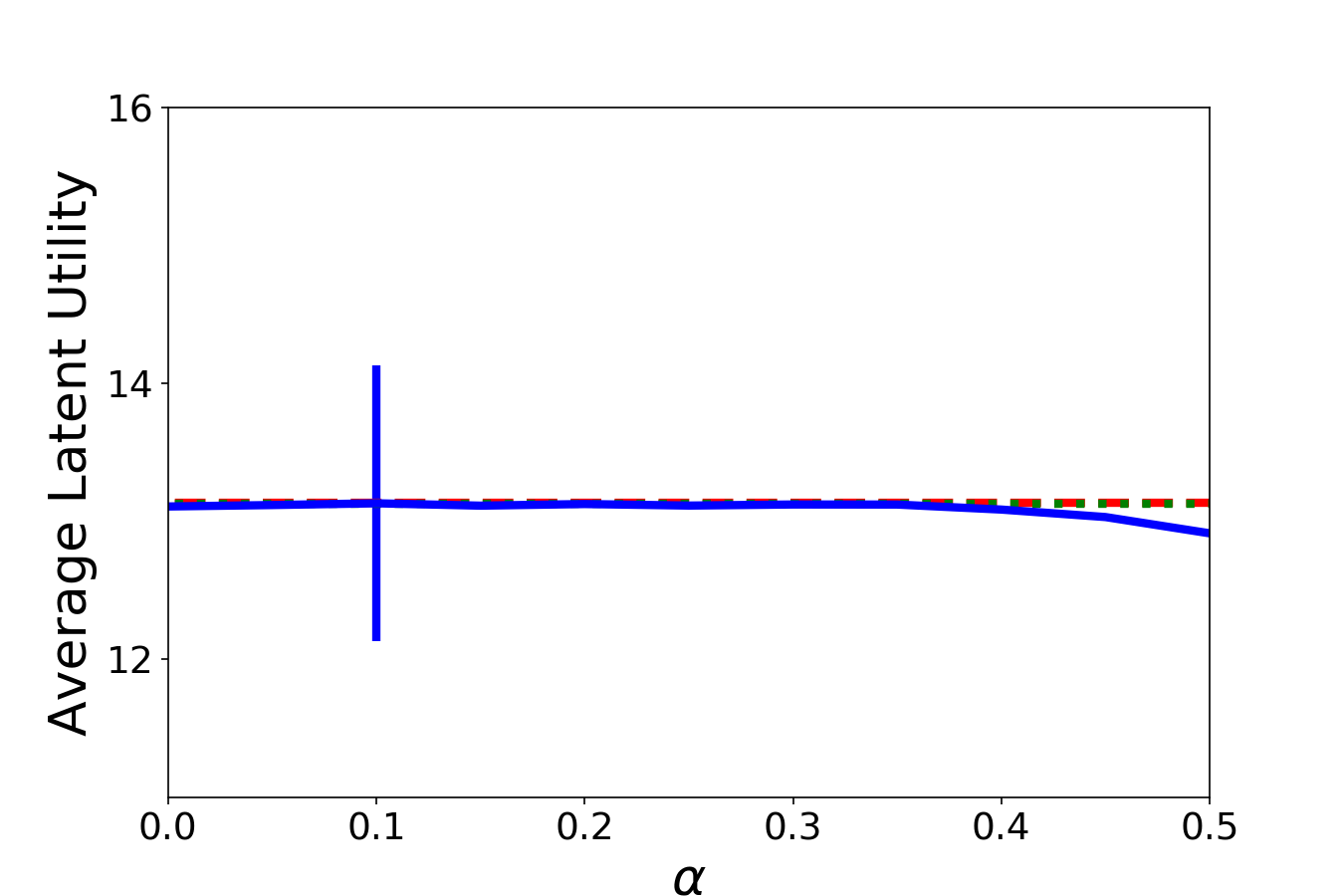}\hspace{-6mm}\\
    \vspace{-19mm}\begin{sideways}$\beta=\nfrac{1}{2}$\white{........}\end{sideways}&&&\\
    &
    \includegraphics[width=0.32\linewidth, trim={-0.1cm 0cm 0.5cm 1.7cm},clip]{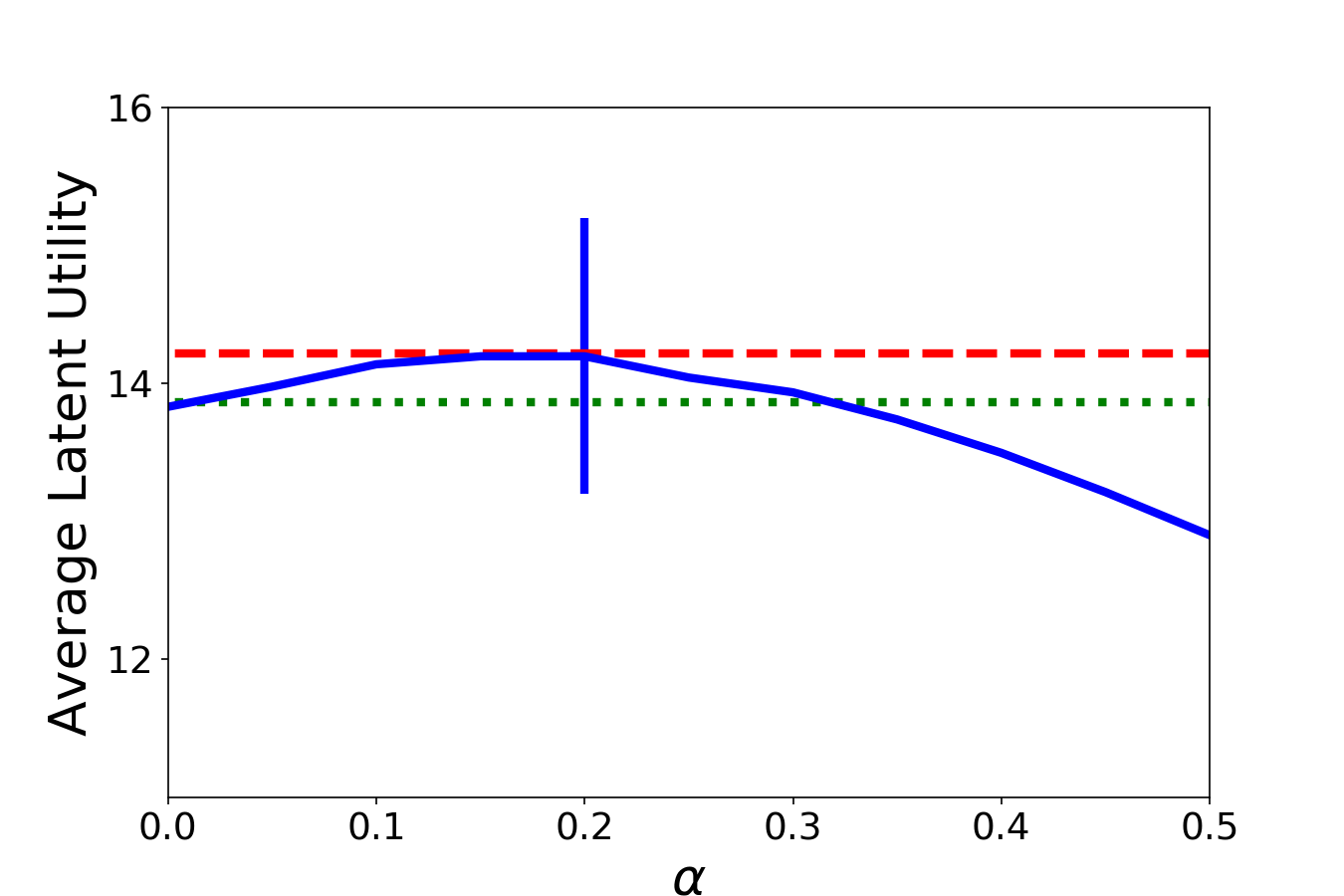}\hspace{-6mm} & %
    \includegraphics[width=0.32\linewidth, trim={-0.1cm 0cm 0.5cm 1.7cm},clip]{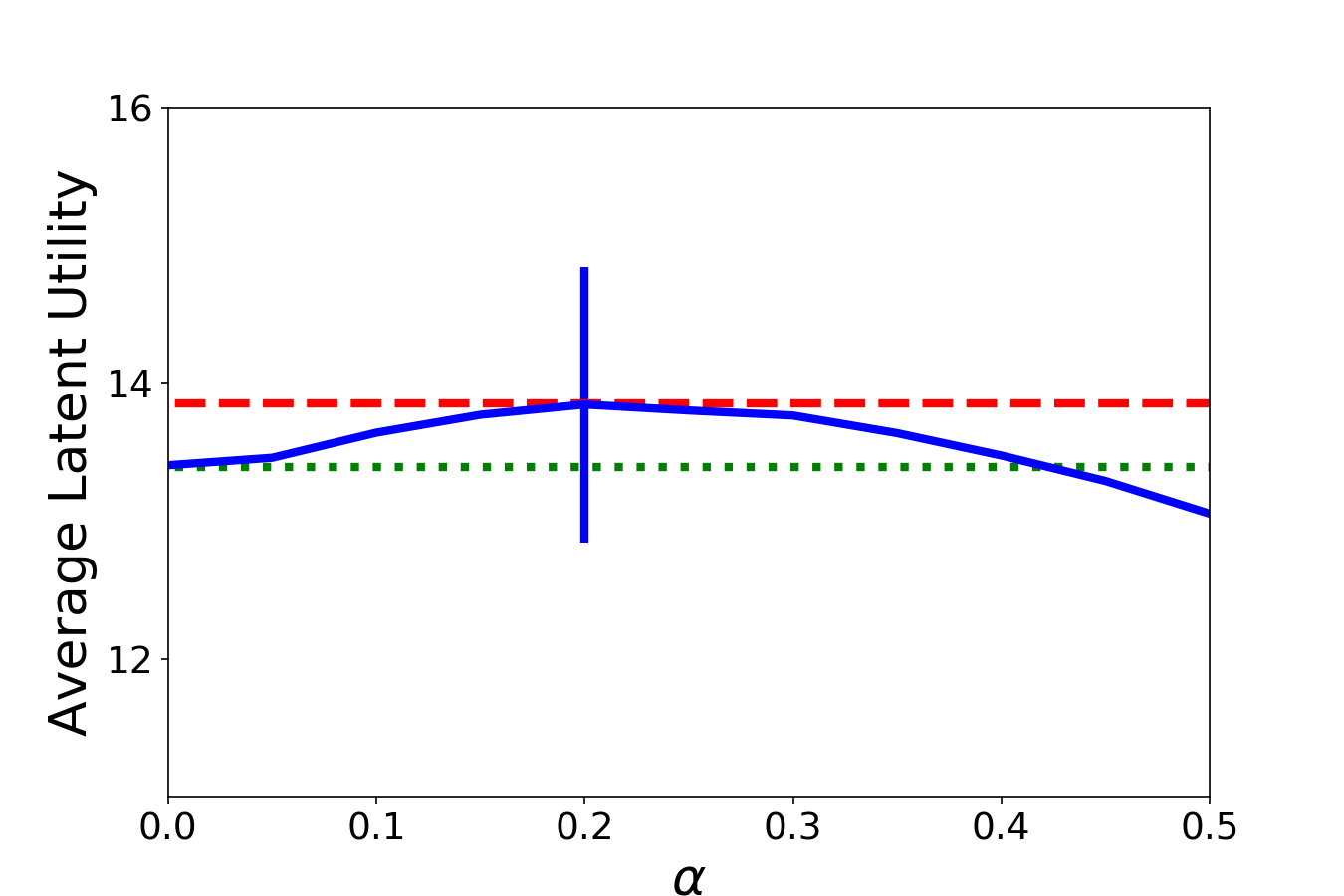}\hspace{-6mm} & %
    \includegraphics[width=0.32\linewidth, trim={-0.1cm 0cm 0.5cm 1.7cm},clip]{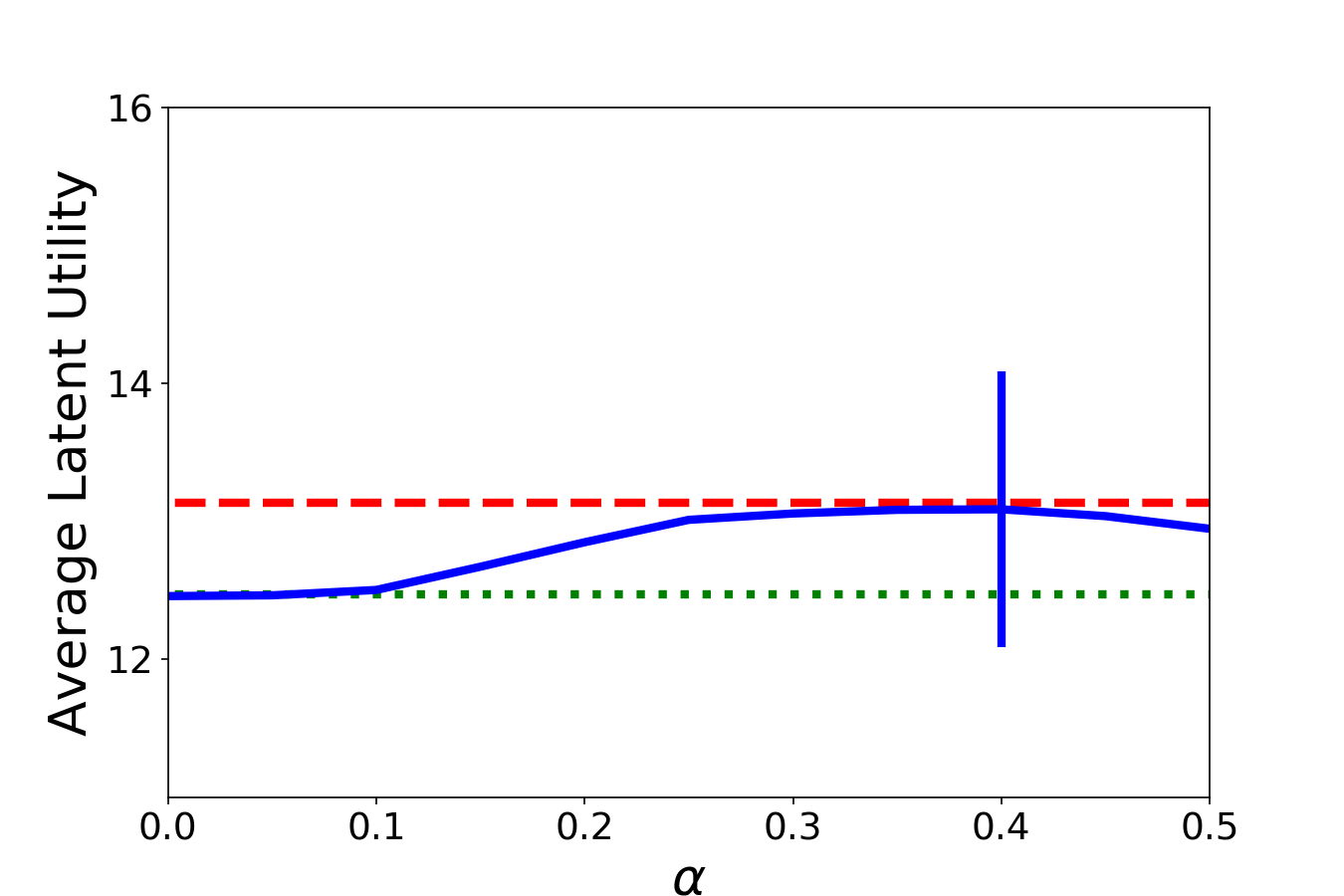}\hspace{-6mm}\\
    \vspace{-19mm}\begin{sideways}$\beta=\nfrac{1}{4}$\white{........}\end{sideways}&&&\\
    &
    \includegraphics[width=0.32\linewidth, trim={-0.1cm 0cm 0.5cm 1.7cm},clip]{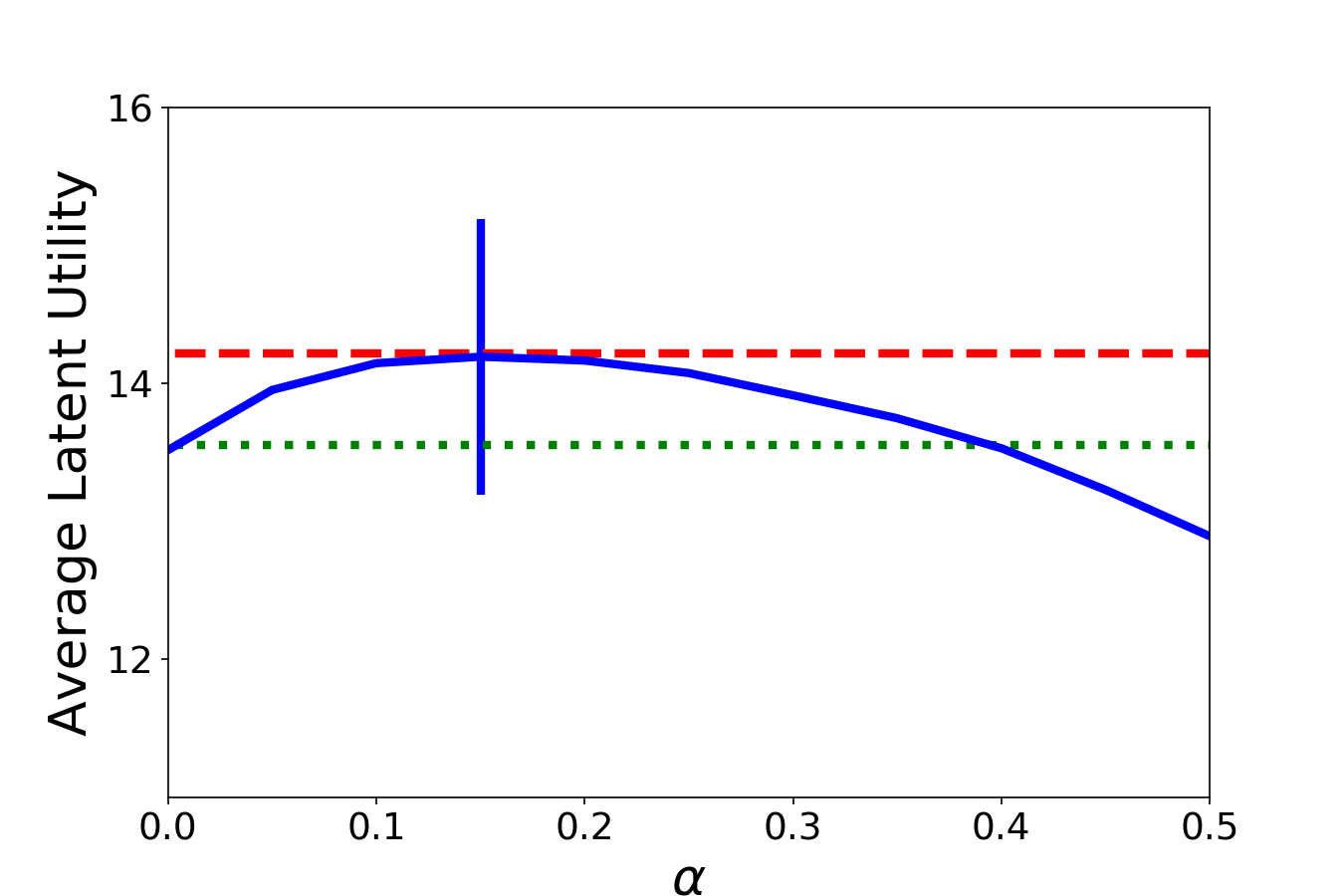}\hspace{-6mm} &
    \includegraphics[width=0.32\linewidth, trim={-0.1cm 0cm 0.5cm 1.7cm},clip]{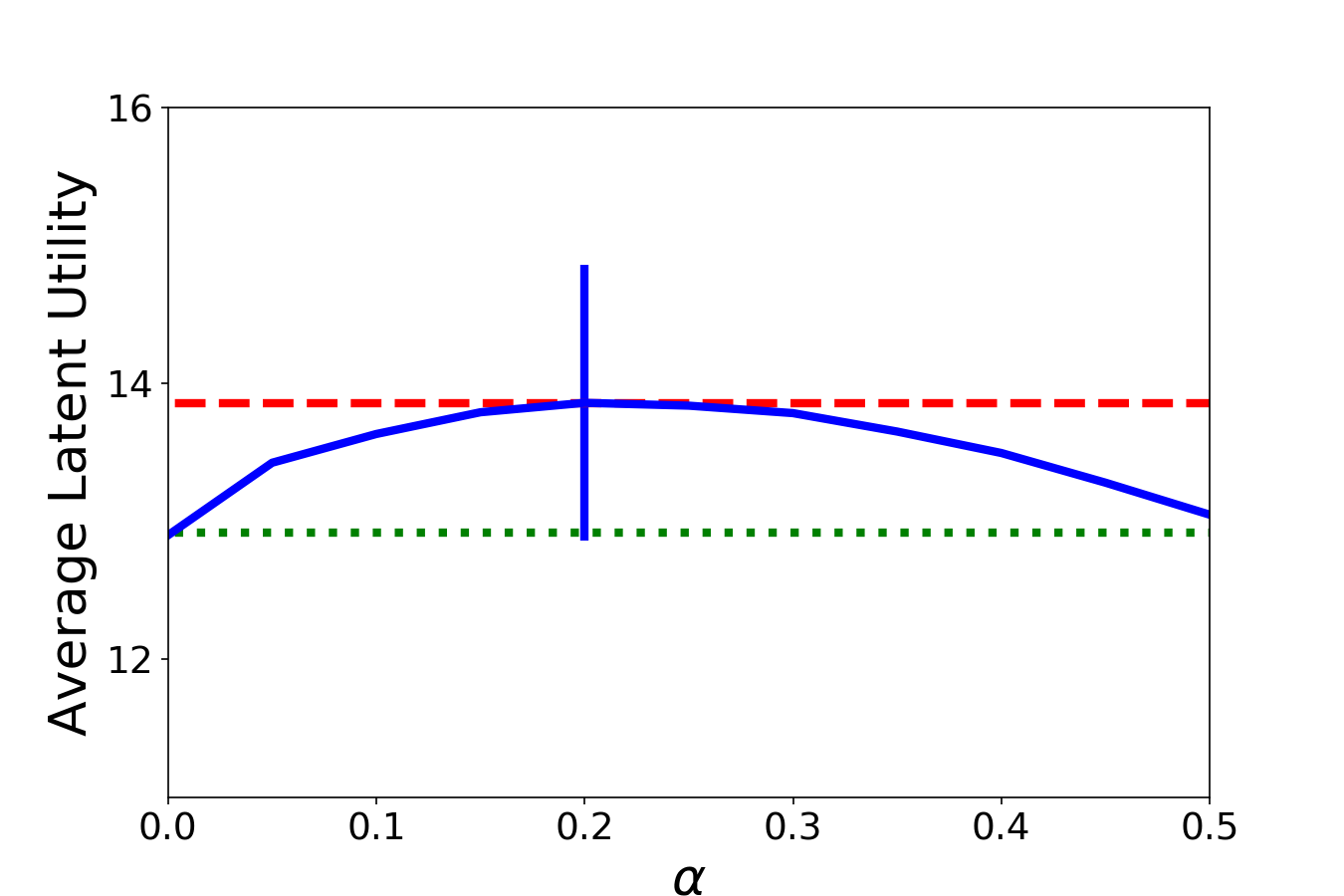}\hspace{-6mm} &
    \includegraphics[width=0.32\linewidth, trim={-0.1cm 0cm 0.5cm 1.7cm},clip]{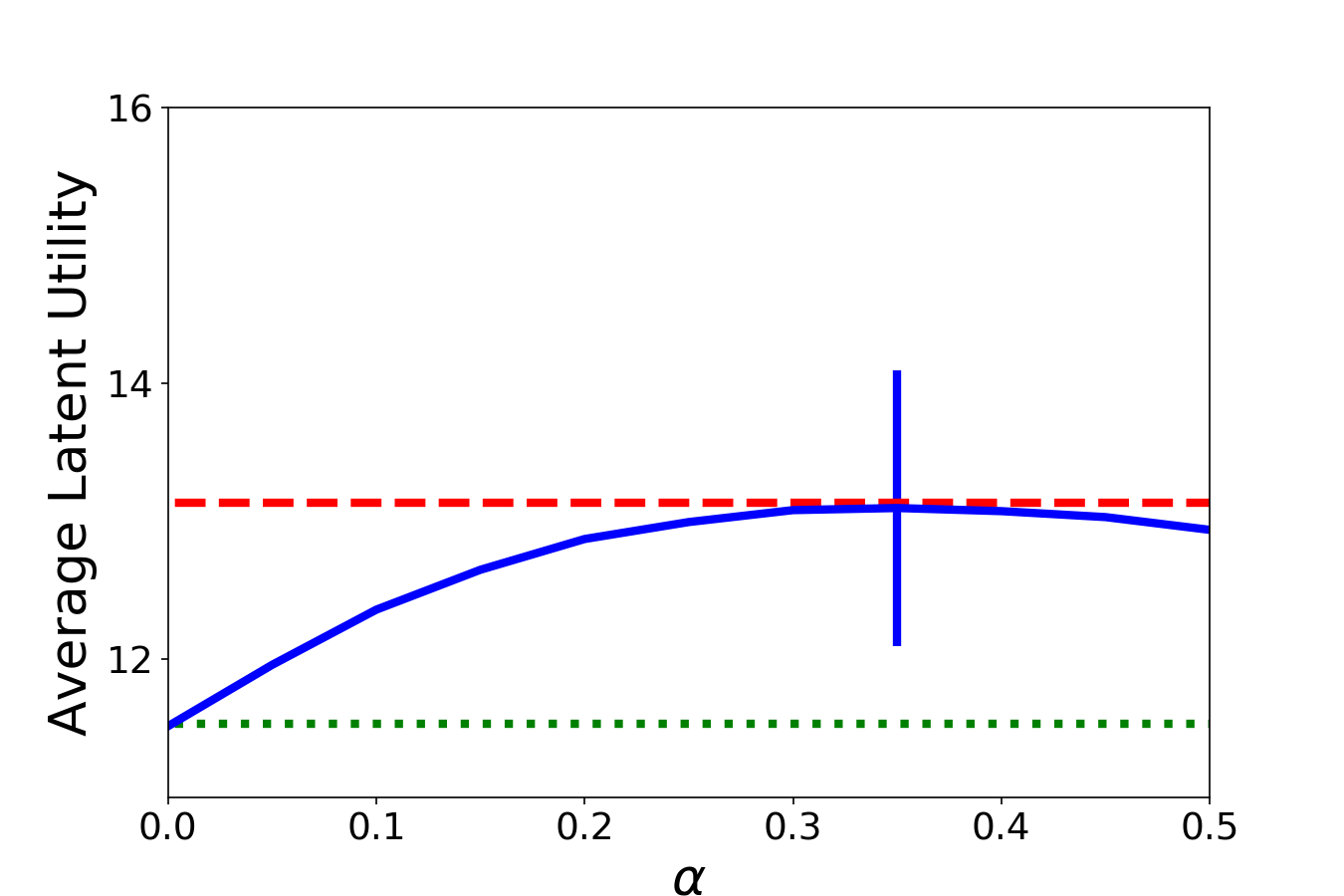}\hspace{-6mm}\\
    \hline
  \end{tabular}
  \caption{
  {\em {Empirical results on IIT-JEE 2009 dataset (without DCG):}}
  We plot the latent utilities, $U_{\mathcal{D},v}(\alpha, \beta)$, $U_{\mathcal{D},v}(0, \beta)$, and $U_{\mathcal{D},v}(0,1)$ obtained by \cons{}, \uncons{} and \opt{} respectively
  (see Equation~\eqref{expected_utility} for the definition of $U_{\mathcal{D},v}(\cdot,\cdot)$);
  we average over values over $5\cdot10^3$ trials.
  Each plot represents an instance of the problem for a given value of implicit bias parameter $\beta$ and the ratio of the size, $m_b$, of the underprivileged group, to the size $m_a$, of the privileged group.
  The bar represents the optimal constraint $\alpha$: where we require the ranking to place at least $k\alpha$ candidates in the top $k$ positions of the ranking for every position $k$.
  }
  \label{table_2}
\end{table*}

\begin{table*}[t!]
  \begin{tabular}{|c|ccc|}
    \hline
    & $\nfrac{m_b}{m}=\nfrac{1}{4}$  & $\nfrac{m_b}{m}=\nfrac{1}{3}$  & $\nfrac{m_b}{m}=\nfrac{1}{2}$  \\
    \hline
    &&&\\
    &&&\\
    \vspace{-15mm}\begin{sideways}$\beta=1$\white{........}\end{sideways}&&&\\
    &
    \includegraphics[width=0.32\linewidth, trim={-0.1cm 0cm 0.5cm 1.4cm},clip]{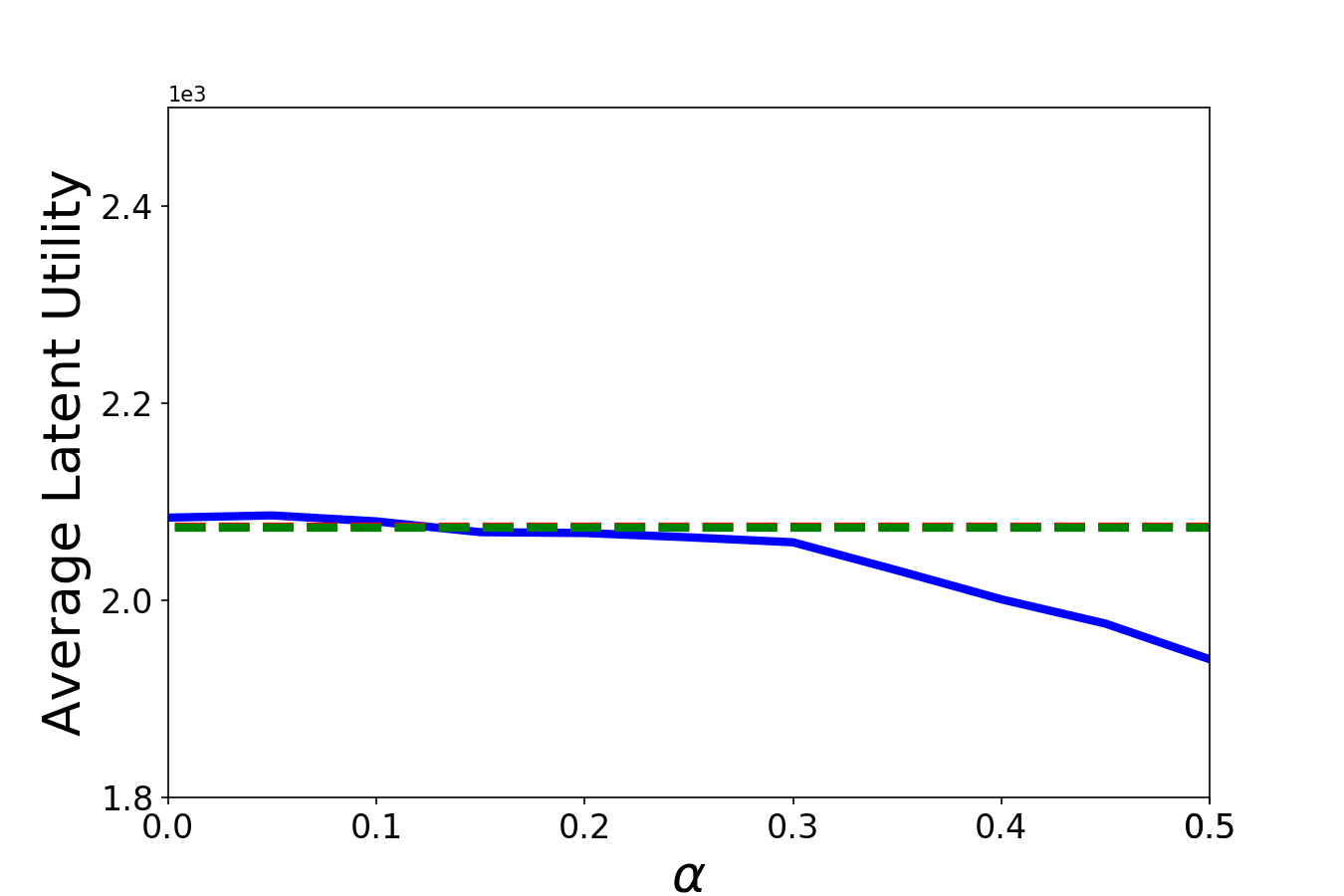}\hspace{-6mm} &
    \includegraphics[width=0.32\linewidth, trim={-0.1cm 0cm 0.5cm 1.4cm},clip]{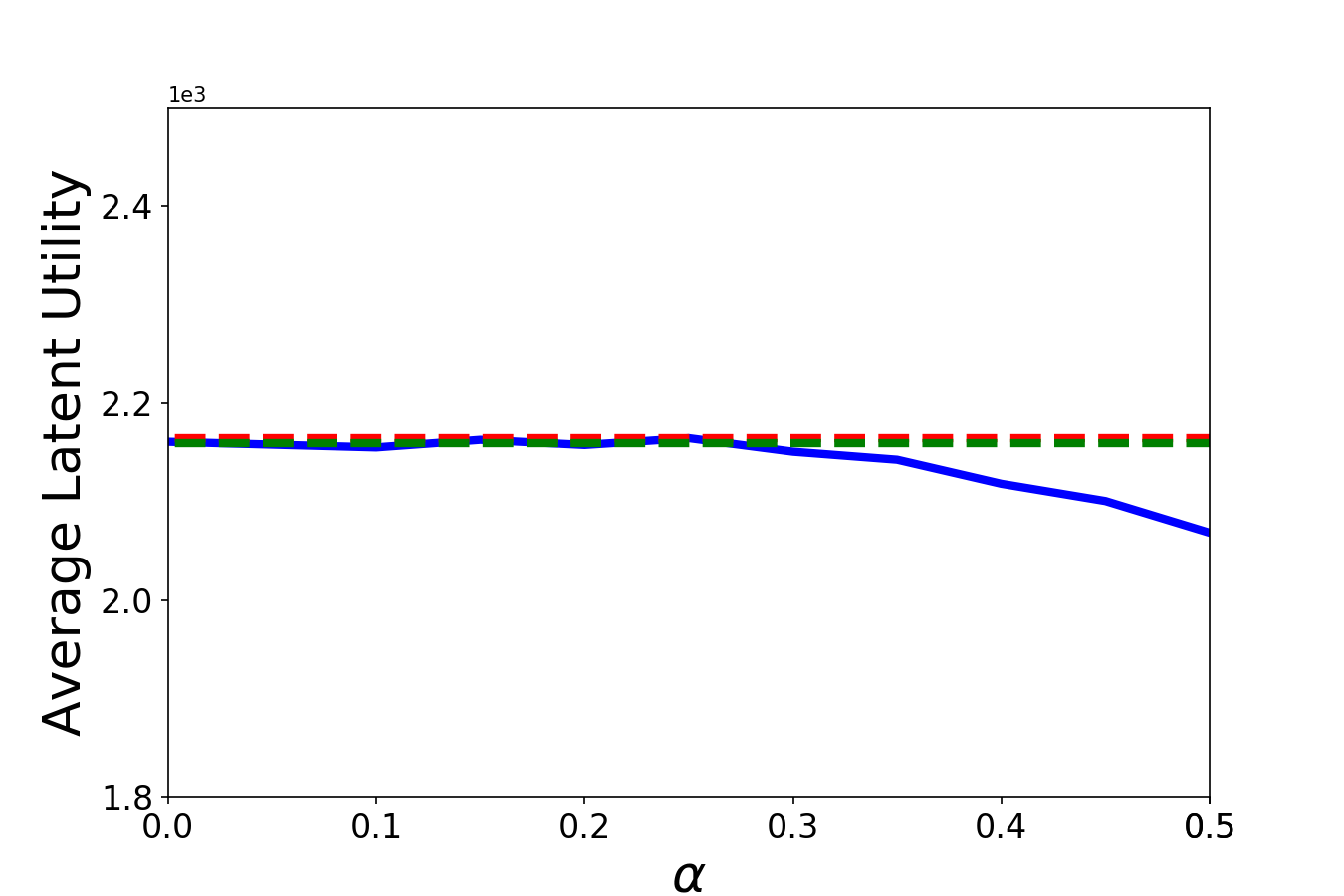}\hspace{-6mm} &
    \includegraphics[width=0.32\linewidth, trim={-0.1cm 0cm 0.5cm 1.4cm},clip]{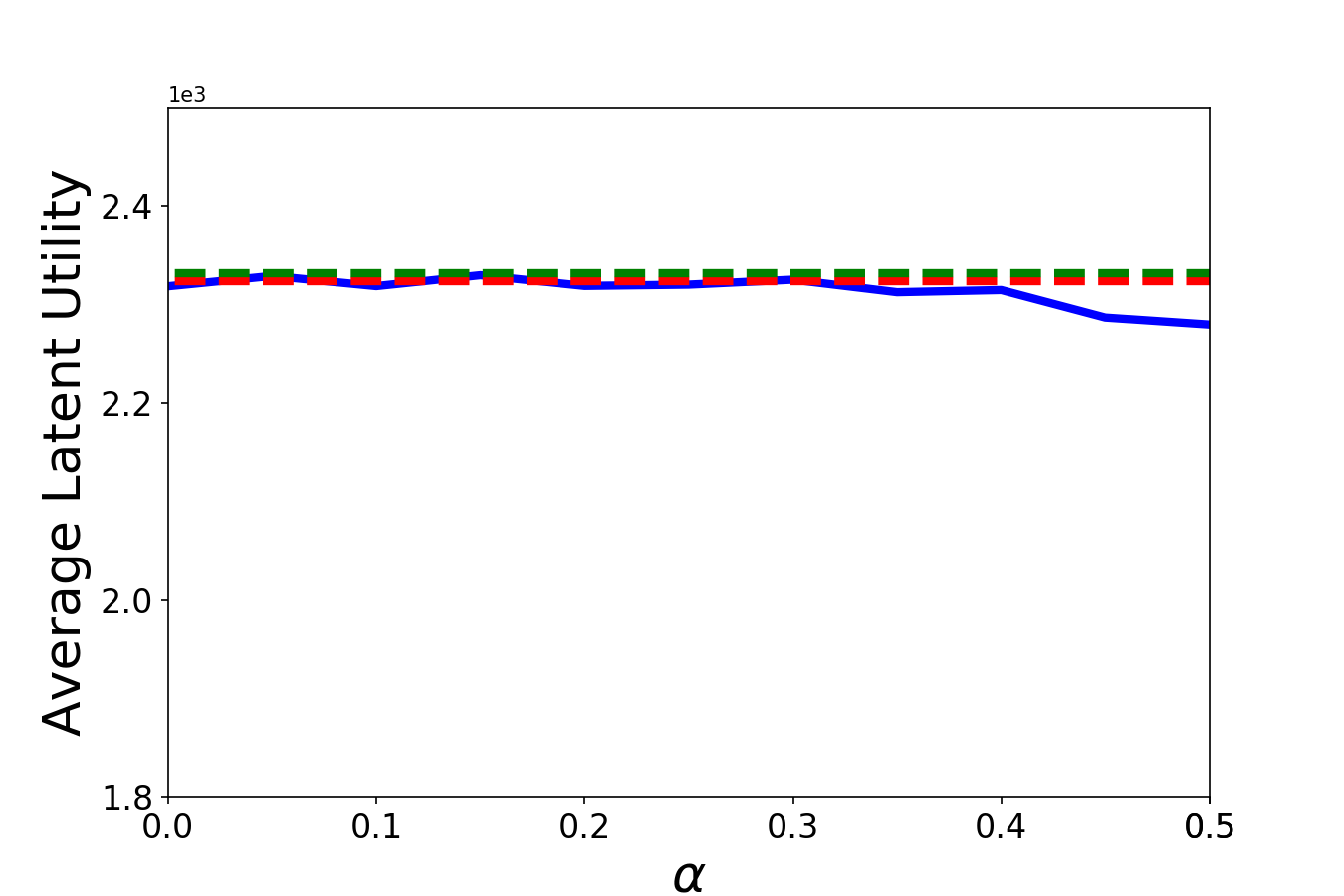}\hspace{-6mm}\\
    \vspace{-15mm}\begin{sideways}$\beta=\nfrac{1}{2}$\white{........}\end{sideways}&&&\\
    &
    \includegraphics[width=0.32\linewidth, trim={-0.1cm 0cm 0.5cm 1.4cm},clip]{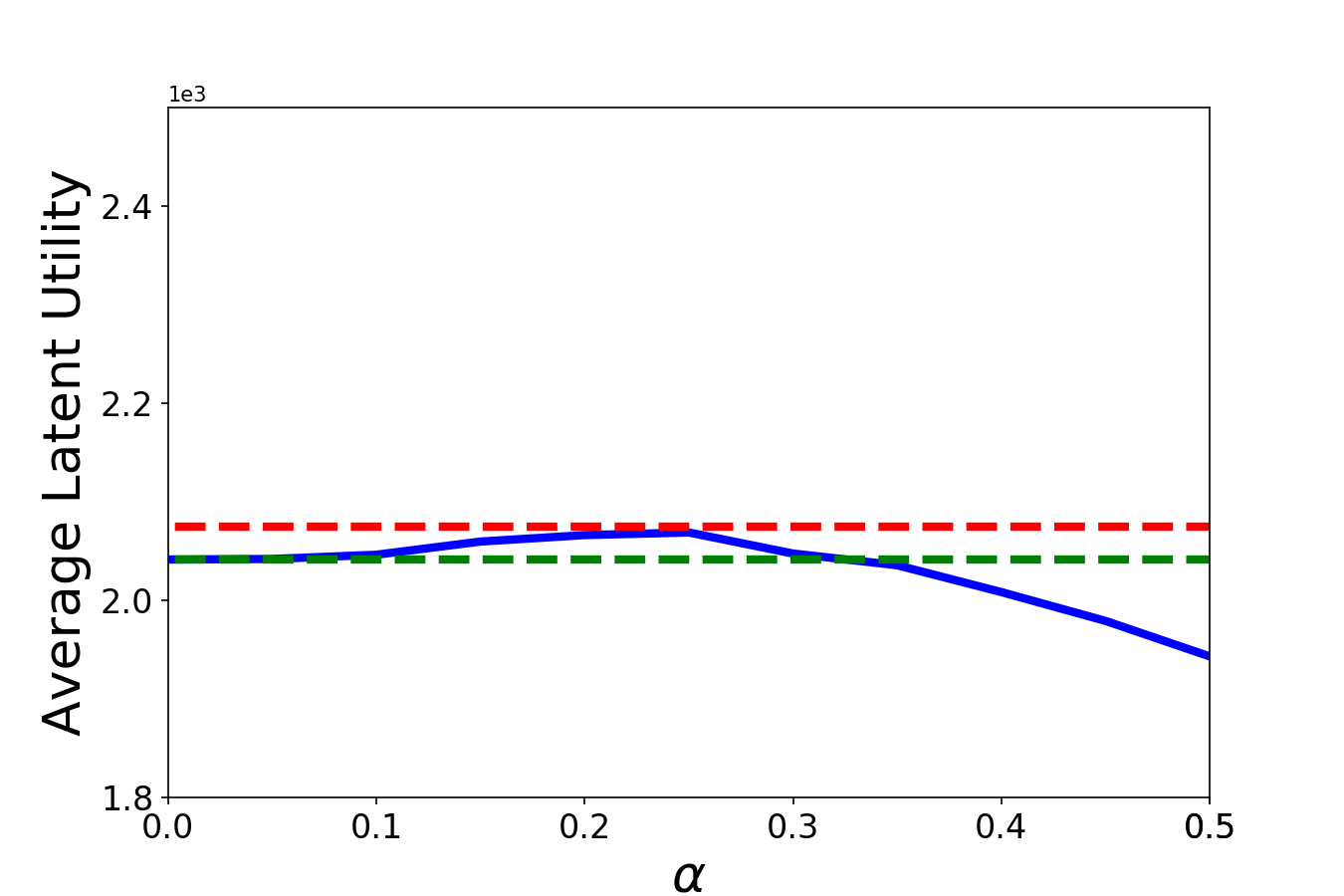}\hspace{-6mm} & %
    \includegraphics[width=0.32\linewidth, trim={-0.1cm 0cm 0.5cm 1.4cm},clip]{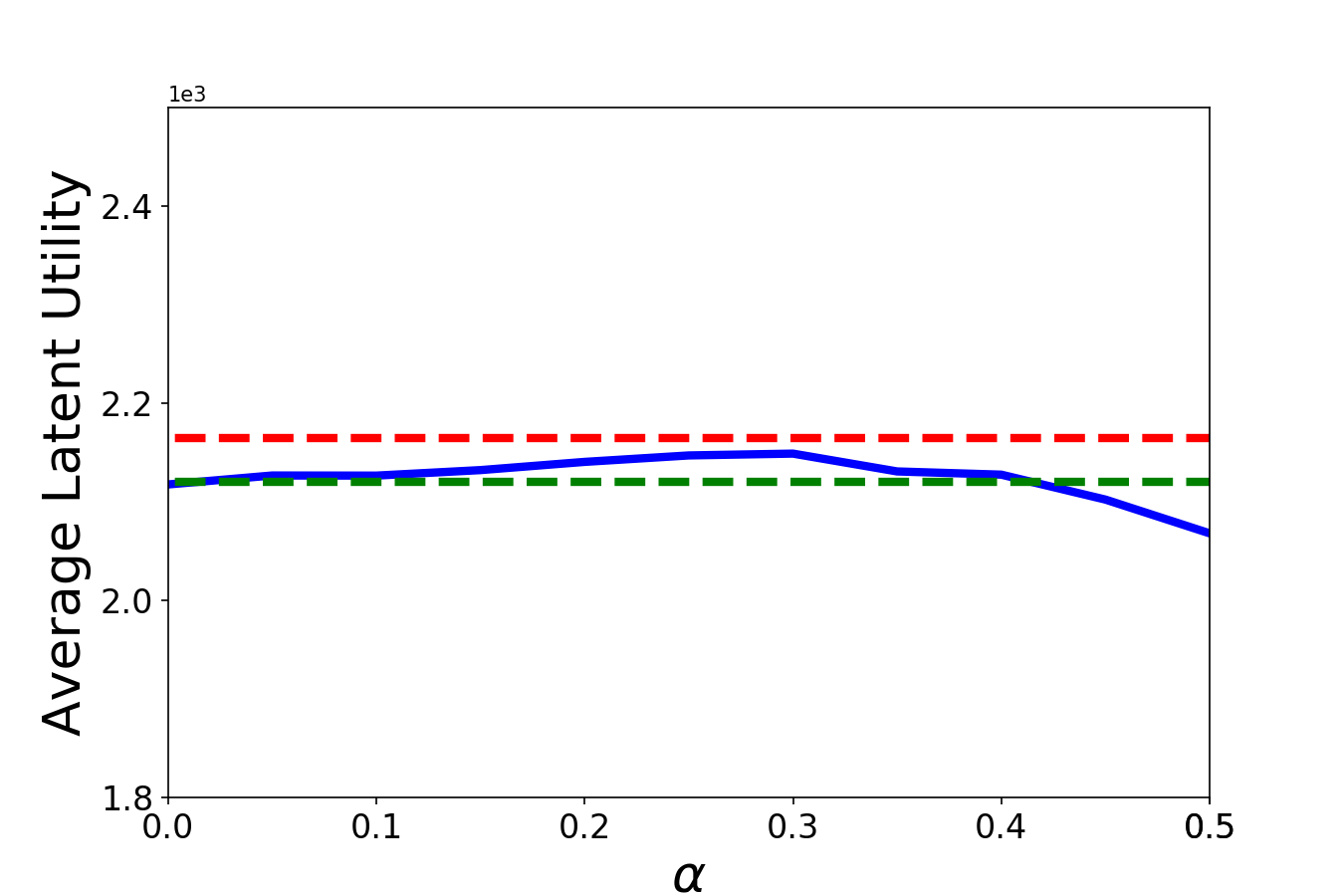}\hspace{-6mm} & %
    \includegraphics[width=0.32\linewidth, trim={-0.1cm 0cm 0.5cm 1.4cm},clip]{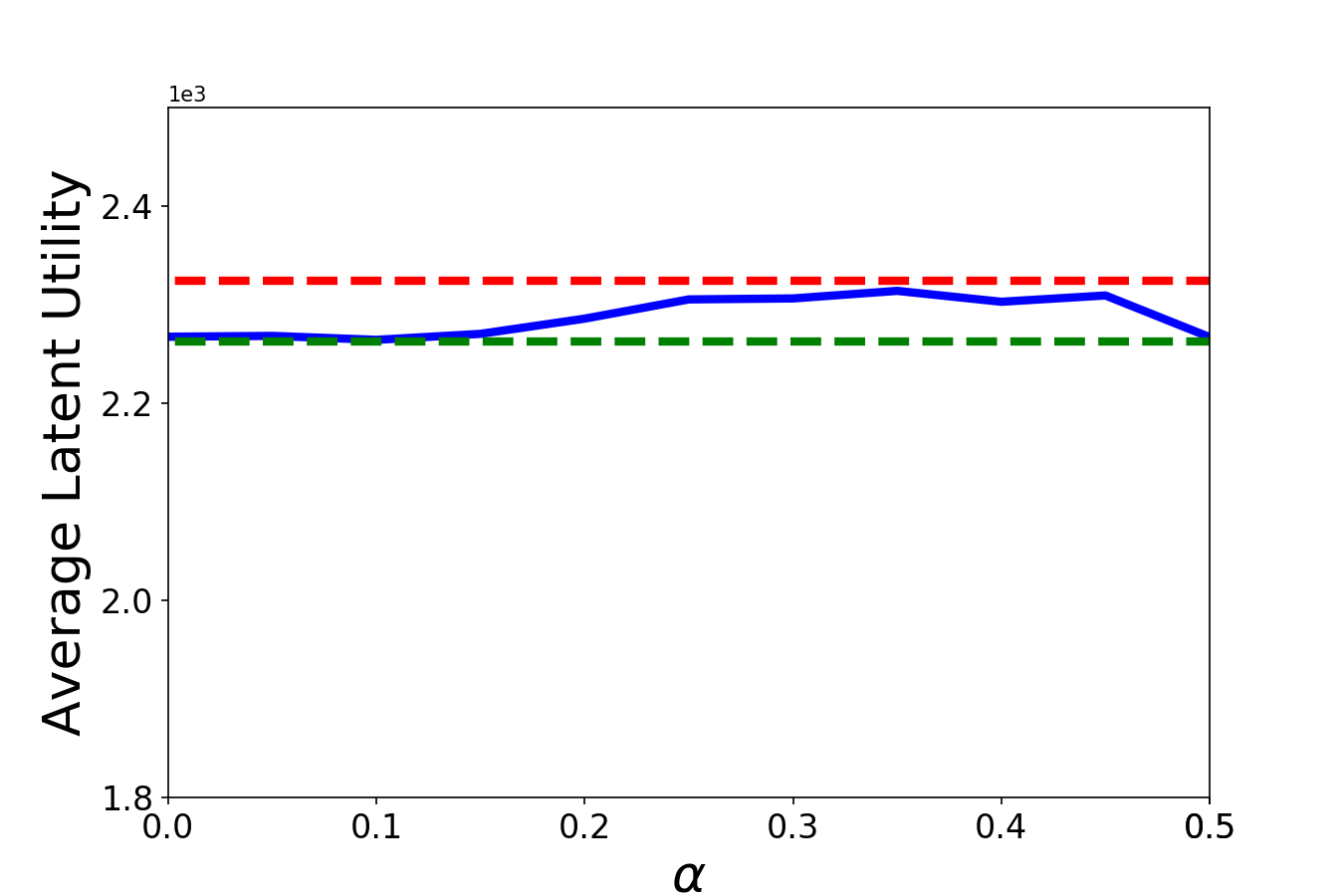}\hspace{-6mm}\\
    \vspace{-15mm}\begin{sideways}$\beta=\nfrac{1}{4}$\white{........}\end{sideways}&&&\\
    &
    \includegraphics[width=0.32\linewidth, trim={-0.1cm 0cm 0.5cm 1.4cm},clip]{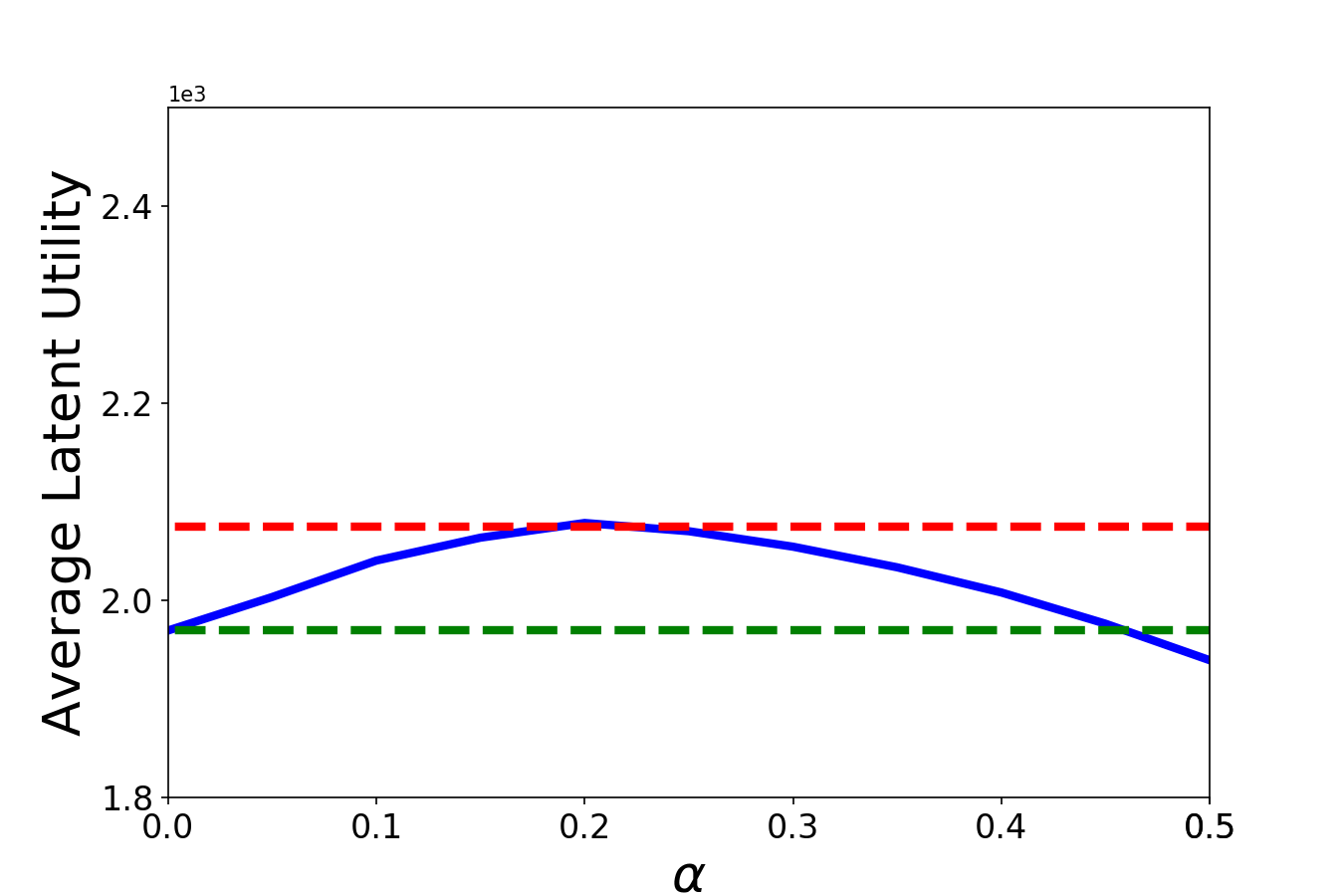}\hspace{-6mm} &
    \includegraphics[width=0.32\linewidth, trim={-0.1cm 0cm 0.5cm 1.4cm},clip]{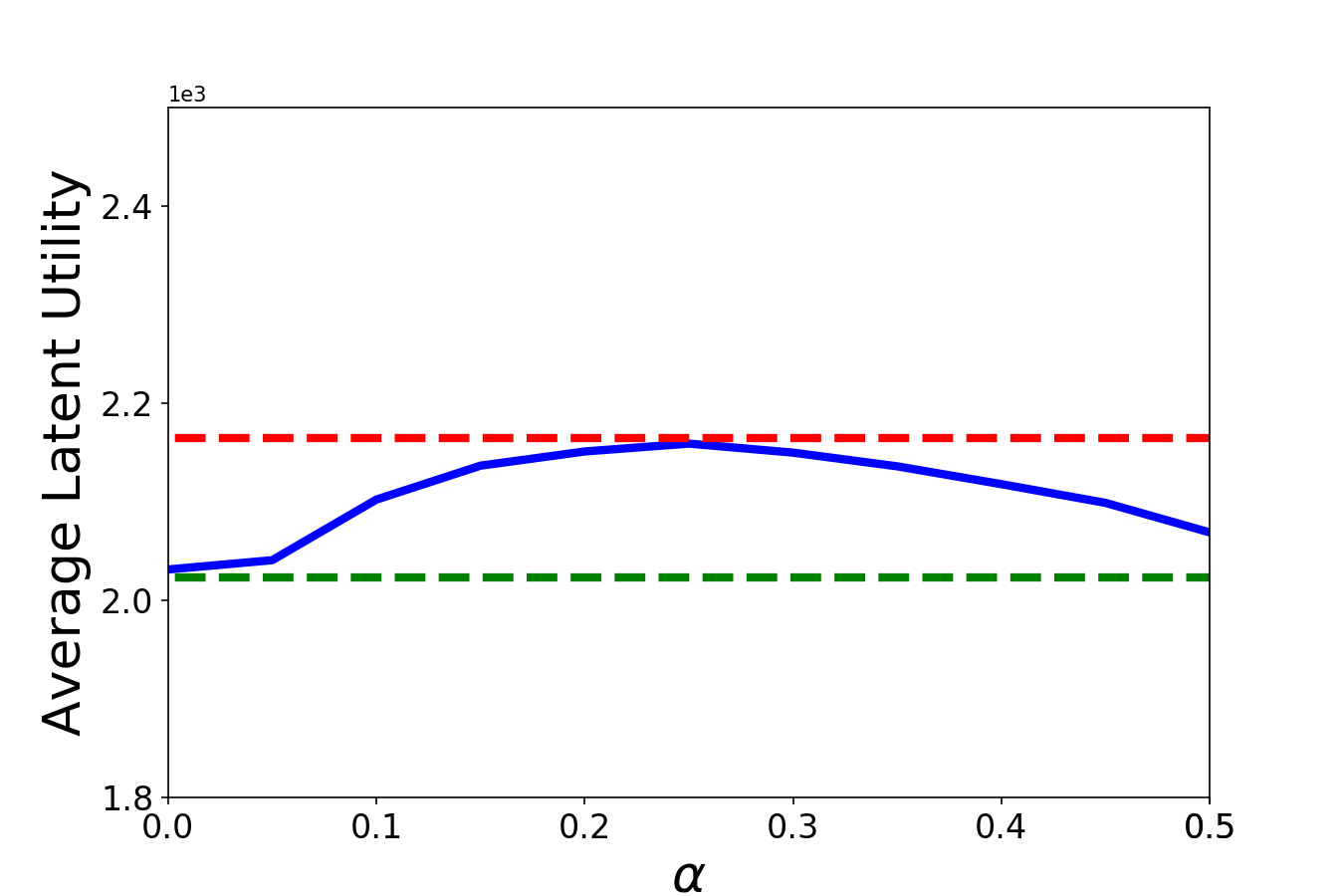}\hspace{-6mm} &
    \includegraphics[width=0.32\linewidth, trim={-0.1cm 0cm 0.5cm 1.4cm},clip]{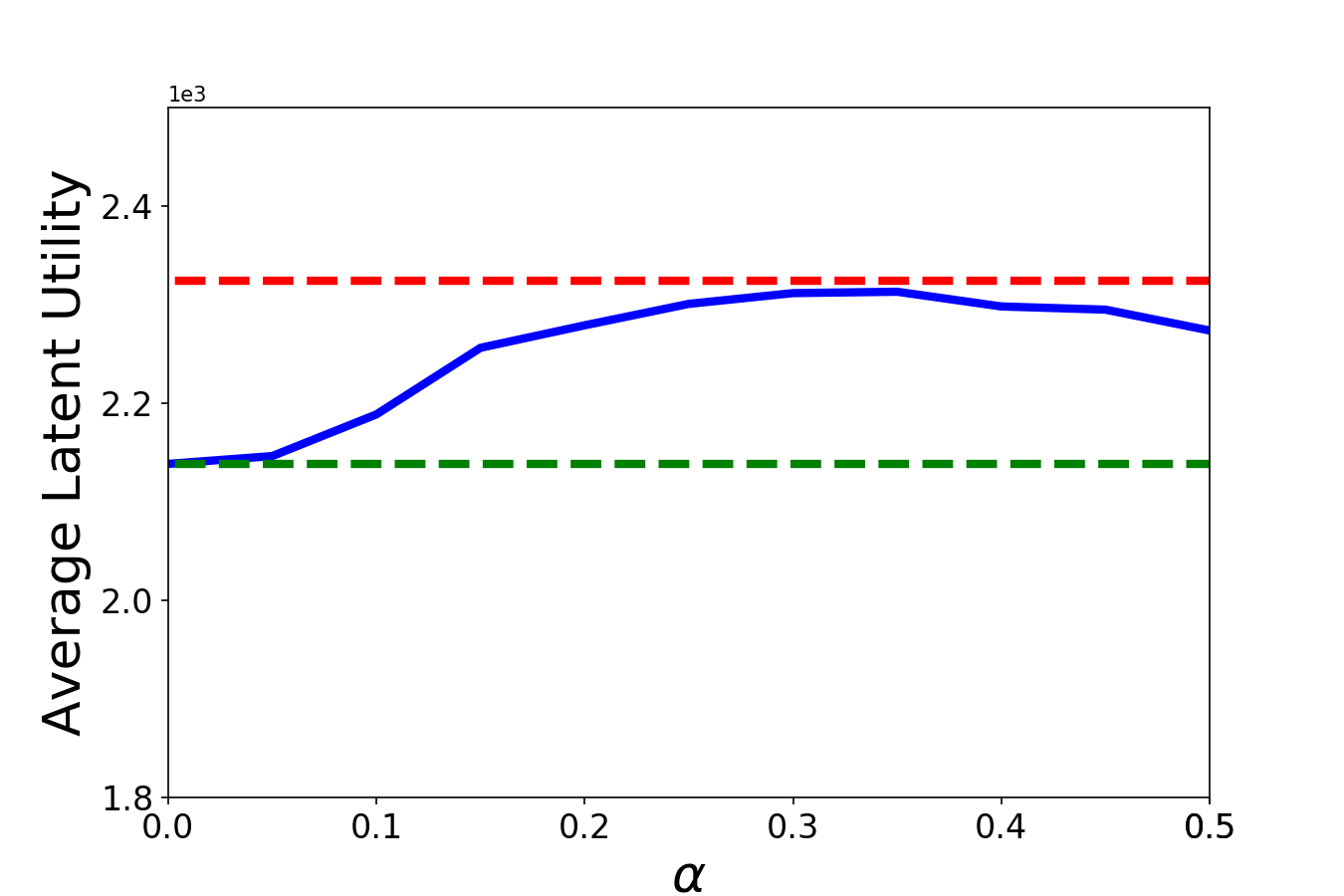}\hspace{-6mm}\\
    \hline
  \end{tabular}
  \caption{
  {\em {Empirical results on distributions of utilities from the Semantic Scholar Open Research Corpus (without DCG):}}
  We plot the latent utilities, $U_{\mathcal{D},v}(\alpha, \beta)$, $U_{\mathcal{D},v}(0, \beta)$, and $U_{\mathcal{D},v}(0,1)$ obtained by \cons{}, \uncons{} and \opt{} respectively
  (see Equation~\eqref{expected_utility} for the definition of $U_{\mathcal{D},v}(\cdot,\cdot)$);
  we average over values over $5\cdot10^3$ trials.
  Each plot represents an instance of the problem for a given value of implicit bias parameter $\beta$ and the ratio of the size, $m_b$, of the underprivileged group, to the size $m_a$, of the privileged group.
  }
  \label{table_3}
\end{table*}

\section{Proofs}
\subsection{Proof of Theorem~\ref{thm_constraints_are_sufficient}}\label{sec_proof_thm_constraints_are_sufficient}
\begin{theorem}
  {\bf (Restatement of Theorem~\ref{thm_constraints_are_sufficient}).}
  Given a set of latent utilities $\{w_i\}_{i=1}^{m}$,
  there exists constraints $L(w)\in \cZ^{n\times p}_{\geq 0}$,
  such that,
  for all implicit bias parameters $\{\beta_s\}_{s=1}^{p}\in(0,1)^{p}$,
  the optimal constrained ranking $\tilde{x}\coloneqq\argmax_{x\in \mathcal{K}(L(w))} \mathcal{W}(x,v,\hat{w})$ satisfies
  \begin{align*}
    \mathcal{W}(\tilde{x}, v, w)=\max\nolimits_{x}\mathcal{W}(x,v,w).\numberthis
  \end{align*}
\end{theorem}

\begin{proof}
  Let $\pi:[m]\to [n+1]$ be the observed utility maximizing ranking and $\pi^\star:[m]\to [n+1]$,  be the optimal {\em latent} utility maximizing ranking, where we define $\pi(i)\coloneqq n+1$ and $\pi^\star(i)\coloneqq n+1$, if item $i$ was not ranked in the first $n$ positions of $\pi$ or $\pi^\star$.
  We claim that the following constraints are suitable for Theorem~\ref{thm_constraints_are_sufficient}:
  For all $s\in [p]$ and $k\in[n]$:
  \begin{align}
    L_{ks} \coloneqq \sum
    \nolimits_{t\in G_s\colon \pi^\star(t)\leq k}1.\label{eq_construction_of_constraints}
  \end{align}
  We will prove that under these constraints, $\pi$ and $\pi^\star$ are the same up to the groups of the items ranked at each position.
  The following lemma shows that they also have the same latent utility.
  \begin{lemma}
    If for all $k\in[n]$, $T_{i_k} = T_{i^\star_{k}}$ where $i_k, i^\star_k$ are such that $i_k=\pi^{-1}(k)$ and $i^\star_k=(\pi^{\star})^{-1}(k)$, i.e., $\pi$ and  $\pi^\star$ are the same up to the groups of the items ranked at each position,
    then $\pi$ and $\pi^\star$ have the same latent utility.
  \end{lemma}
  \begin{proof}
    We show that the relative order of all items with different latent utilities is  the same  between the two rankings, proving that they have the same latent utility.
    Consider two items $i_1$ and $i_2$ in the same set of groups, i.e., such that $T_{i_1}=T_{i_2}$.
    We note that swapping their positions does not violate any new constraints.
    Further since, $i_1$ and $i_2$ have the same implicit bias, we have
    \begin{align*}
      \big( T_{i_1}=T_{i_2} \text{ and }w_{i_1} > w_{i_2} \big) \implies \hat{w}_{i_1} > \hat{w}_{i_2}.\label{eq_implication}\numberthis
    \end{align*}
    Towards a contradiction, assume that $\pi$ and $\pi^\star$ have different relative order of $i_1$ and $i_2$.
    Without loss of generality let $\pi^\star(i_1)  < \pi^\star(i_2)$ and $\pi(i_1) > \pi(i_2)$.
    Since $\pi^\star$ is optimal we have $w_{i_1} \geq w_{i_2}$.
    If $w_{i_1}=w_{i_2}$, then these items do not change the latent utility between $\pi^\star$  and $\pi$.
    Let $w_{i_1} < w_{i_2}$, from Equation~\eqref{eq_implication} we have that $\hat{w}_{i_1} < \hat{w}_{i_2}$.
    Therefore, we can swap the positions of $i_1$ and $i_2$ in $\pi$ to gain the following observed utility
    \begin{align*}
      \hat{w}_{i_1}v_{\pi(i_2)}+\hat{w}_{i_2}v_{\pi(i_1)}&-\hat{w}_{i_1}v_{\pi(i_1)}-\hat{w}_{i_2}v_{\pi(i_2)}
      =
      \big(\hat{w}_{i_2}-\hat{w}_{i_1}\big)\cdot\big(v_{\pi(i_1)}-v_{\pi(i_2)}\big)> 0.
    \end{align*}
    This contradicts the fact that $\pi$ has the optimal observed utility.
  \end{proof}

  It remains to prove that with these constraints, $\pi$ is the same as $\pi^\star$ up to the groups of the items.
  This proof is by induction.

  \noindent {\em Inductive hypothesis: } Let the two rankings agree on the first $(k-1)$ positions (up to the groups of the item ranked at each position).

  \noindent {\em Base case: } The base case for $k=1$ is trivially satisfied.
  \noindent Towards, a contradiction  assume that the items on the $k$-th positions of $\pi$ and $\pi^\star$ are $i$ and $i^\star\in [m]$, such that $T_{i}\neq T_{i^\star}$.
  Since $\pi^\star$ ranks $i$ after $i^\star$ it follows that
  \begin{align*}
    w_{i}\leq w_{i^\star}.\numberthis\label{eq_relative_utility}
  \end{align*}

  \noindent {\bf Case A: $T_{i}\subsetneq T_{i^\star}$}:
  \noindent We claim that this case is not possible.
  \noindent Consider any group $G_s$, for $s\in T_{i^\star}\backslash T_{i}$.
  Since the rankings agree on the first $(k-1)$ positions we have for all $j\in[k-1]$
  \begin{align*}
    \sum\nolimits_{t\in G_s\colon \pi(t)\leq j} 1=\sum\nolimits_{t\in G_s\colon \pi^\star(t)\leq j} 1\label{eq_constraints_in_first_k_pos}.\numberthis
  \end{align*}
  Since $\pi(i)=k$ we have
  \begin{align*}
    \sum\nolimits_{t\in G_s\colon \pi(t)\leq k} 1\ &= \sum\nolimits_{t\in G_s\colon \pi(t)<k} 1 + \mathbb{I}[i\in G_s]\\
    &\hspace{-1mm}\stackrel{\eqref{eq_constraints_in_first_k_pos}}{=}\sum\nolimits_{t\in G_s\colon \pi^\star(t)<k}1+\mathbb{I}[i\in G_s]\\
    &<\hspace{0.7mm}\sum\nolimits_{t\in G_s\colon \pi^\star(t)<k}1+ \mathbb{I}[i\in G_s]+ \mathbb{I}[i^\star \in G_s] \tag{Using $i^\star\in G_s$}\\
    &=\hspace{0.5mm}\sum\nolimits_{t\in G_s\colon \pi^\star(t)\leq k} 1+\mathbb{I}[i\in G_s] \\
    &=\hspace{0.5mm}\sum\nolimits_{t\in G_s\colon \pi^\star(t)\leq k} 1 \stackrel{\eqref{eq_construction_of_constraints}}{=}L_{ks}.\tag{Using $i\not\in G_s$}
  \end{align*}
  Therefore, in this case $\pi$ violates the constraint at position $k$.

  \noindent {\bf Case B: $T_{i}\supsetneq T_{i^\star}$:} In this case we have
  \begin{align*}
    \hat{w}_{i}\coloneqq \bigg(\prod\nolimits_{s\in [p]\colon G_s\ni i}\beta_s \bigg)\cdot w_{i}
    &= \bigg(\prod\nolimits_{s\in [p]}\beta_s^{\mathbb{I}[i\in G_s]}\bigg)w_i\\
    &\hspace{-0mm}= \bigg(\prod_{s\in [p]}\beta_s^{\mathbb{I}[i^\star\in G_s]}\bigg)
    \bigg(\prod_{s\in [p]}\beta_s^{\mathbb{I}[ i\in G_s]\cdot\mathbb{I}[i^\star\not\in G_s]}\bigg)  w_{i}\tag{Using $T_{i}\supsetneq T_{i^\star}$}\\
    &\hspace{-0mm}< \bigg(\prod\nolimits_{s\in [p]}\beta_s^{\mathbb{I}[i^\star\in G_s]}\bigg)\cdot  w_{i}\tag{For all $s\in [p]$, $\beta_s\in (0,1)$}\\
    &\stackrel{\eqref{eq_relative_utility}}{<} \bigg(\prod\nolimits_{s\in [p]}\beta_s^{\mathbb{I}[i^\star\in G_s]}\bigg)\cdot w_{i^\star}\\
    &< \hat{w}_{i^\star}.\label{eq_non_decreasing_utility}\numberthis
  \end{align*}
  \noindent Consider the ranking $\hat{\pi}$ formed by swapping the position of $i$ and $i^\star$ in $\pi$.
  The change in observed utility between $\hat{\pi}$ and $\pi$ is
  \begin{align*}
    \hat{w}_{i}v_{\pi(i^{\star})}+\hat{w}_{i^\star}v_{\pi(i)}&-\hat{w}_{i}v_{\pi(i)}-\hat{w}_{i^\star}v_{\pi(i^\star)}
    =
    \big(\hat{w}_{i^\star}-\hat{w}_{i}\big)\cdot\big(v_{\pi(i)}-v_{\pi(i^{\star})}\big)\stackrel{\eqref{eq_non_decreasing_utility}}{>} 0.
  \end{align*}
  Therefore, $\hat{\pi}$ has higher observed utility than $\pi$.
  Therefore, if $\hat{\pi}$ satisfies the constraints, then this contradicts the optimality of $\pi$, and we are done.
  $\hat{\pi}$ can only violate the constraints in positions $j\in [k, \dots, \pi(i^\star)]$ for groups $T_{i}\backslash T_{i^\star}$.
  Since the number of items selected from the set of groups $T_{i}\backslash T_{i^\star}$ only increases.
  It follows that if $\pi$ satisfies the lower bound constraints on positions $j\in [k, \dots, \pi(i^\star)]$, then so does $\hat{\pi}$.
  Since we assume $\pi$ to be feasible, we are done.
\end{proof}

\subsection{Proof of Theorem~\ref{fact_spread_dist}}\label{sec_proof_spread_distribution}
\begin{theorem}
  {\bf (Restatement of Theorem~\ref{fact_spread_dist}).}
  Let $\mathcal{D}$ be a continuous distribution, $\ell\leq m_b$, and $0<k < \min(m_a, m_b)$ be a position, then
  \begin{align}
    \forall\ \delta\geq 2 \hspace{5mm}&\Pr[\ N_{kb} \leq \eE[N_{kb}] - \delta \ ] \leq e^{-\frac{(2\delta^2-1)}{k}},\label{cons_of_nkb}\\
    &\eE[N_{kb}] = k\cdot\nfrac{m_b}{(m_a+m_b)},\label{expect_nkb}\\
    &\eE[P_\ell] = \ell\cdot\big(1+\nfrac{m_a}{(m_b+1)}\big).\label{eq:expectation_z_ell}
  \end{align}
\end{theorem}

\paragraph{Equivalent simple model.}
Since the items have the same distribution $\mathcal{D}$ of latent utility,
the event: $i\in G_a$ (or equivalently $i\in G_b$), is independent of the event:  $w_i=z$ for some $z\in \text{supp}(\mathcal{D})$.
Therefore, the latent utility maximizing ranking is independent of the groups  assigned to each item.

It follows that the following is an equivalent model of generating the unbiased ranking:
First, we draw $m_a+m_b$ latent utilities and order them in a non-increasing order (randomly breaking ties).
Then, we choose $m_b$ items uniformly without replacement and assign them to $G_b$, and assign the others to $G_a$.
Let a item $i\in G_b$ be a blue ball and $i\in G_a$ be a red ball, then distributions of positions of the balls is equivalent to:
Given $m_a+m_b$ numbered balls, we pick $m_b$ of them uniformly without replacement and color them blue, and color the rest of the balls red.

\begin{proof}
  Using the above model, we first find that $P_{\ell}$ has a negative-hypergeometric distribution and use it to calculate, $\eE[P_{\ell}]$, then we find that $N_{kb}$ is a hypergeometric distribution, use this fact to calculate $\eE[N_{kb}]$ and show $N_{kb}$ is concentrated around its mean.

  \noindent {\bf Expectation of $P_{\ell}$.}\
  The total ways of choosing $m_b$ blue balls is $\binom{m_a+m_b}{m_b}$.
  Given $P_{\ell}=k$, the number of ways of choosing the $\ell-1$ blue balls before it is $\binom{k-1}{\ell-1}$ and the number of ways of choosing the $b-\ell$ blue balls after it is $\binom{m_a+m_b-k}{b-\ell}$.
  Therefore by expressing the probability as the ratio of favorable and total outcomes we have
  \begin{align}
    \Pr[P_\ell = k] = \frac{\binom{k-1}{\ell-1}\binom{m_a+m_b-k}{b-\ell}}{\binom{m_a+m_b}{m_b}}.
  \end{align}
  Rewriting it using the fact that $\binom{m_a}{m_b} = \binom{m_a}{m_a-m_b}$ we get
  \begin{align*}
    \Pr[P_\ell = k] &= \frac{\binom{k-1}{k-\ell}\binom{m_a+m_b-k}{m_a-b+\ell}}{\binom{m_a+m_b}{m_a}}=\frac{\binom{(k-\ell)+(\ell-1)}{k-\ell}\binom{m_a+m_b-(k-\ell)-\ell}{m_a-(x-\ell)}}{\binom{m_a+m_b}{m_a}}.
  \end{align*}
  \noindent Comparing this with the density function of the Negative hypergeometric distribution~\myeqref{eq_neg_hyper}{28} it is easy to observe that $P_\ell-\ell$ is a negative hypergeometric variable.

  \begin{tcolorbox}[colback=yellow!10!white,colframe=black!75!black]
    Given numbers $N, K, r$, the negative hypergeometric random variable, $NG$, has the following distribution, for all $k\in[K]$
    \begin{align}
      \Pr[NG = k] \coloneqq \frac{\binom{k+r-1}{k}\binom{N-r-k}{K-k}}{\binom{N}{K}}.\label{eq_neg_hyper}
    \end{align}
  \end{tcolorbox}

  \hspace{-6mm} From the expectation of a negative-hypergeometric variable we have
  \begin{align*}
    \eE[P_{\ell}-\ell] = \ell\cdot\frac{ m_a}{m_b+1}\implies \eE[P_{\ell}] &= \ell\cdot\frac{ (m_a+m_b+1)}{m_b+1}.\numberthis
  \end{align*}

  \medskip

  \noindent {\bf Expectation and Concentration of $N_{kb}$.}\
  Given $N_{kb}=j$, the number of ways of coloring $j$ out of $k$ balls before it is $\binom{k}{j}$, and the number of ways of coloring $b-j$ balls after it is $\binom{m_a+m_b-k}{b-\ell}$.
  Therefore it follows that:
  \begin{align}
    \Pr[\ N_{kb} = j \ ] = \frac{\binom{k}{j}\binom{m_a+m_b-k}{b-j}}{\binom{m_a+m_b}{m_b}}.\label{eq_prob_mk}
  \end{align}

  \begin{tcolorbox}[colback=yellow!10!white,colframe=black!75!black]
    Given numbers $N, K, n $, for an hypergeometric random variable, $HG$, we have: $\forall \max(0,n+K-N)\leq k \leq \min(K,n)$
    \begin{align}
      \Pr[HG = k] \coloneqq \frac{\binom{K}{k}\binom{N-K}{n-k}}{\binom{N}{n}}.\label{eq_hyper}
    \end{align}
  \end{tcolorbox}

  By comparing with Equation~\eqref{eq_hyper}, we can observe that $N_{kb}$ is a hypergeometric random variable.
  From well known properties of the hypergeometric distribution we have that
  \begin{align}
    \eE[\ N_{kb} \ ] &= \frac{kb}{(m_a+m_b)}\\
    \Pr\left[ N_{kb} \geq \nfrac{kb}{(m_a+m_b)} + \delta\right] &\stackrel{\text{\cite{hush2005concentration}}}{\leq} e^{-2(\delta^2-1)\gamma} \leq e^{-\frac{2(\delta^2-1)}{k+1}}\\
    \Pr\left[ N_{kb} \leq \nfrac{kb}{(m_a+m_b)} - \delta\right] &\stackrel{\text{\cite{hush2005concentration}}}{\leq} e^{-2(\delta^2-1)\gamma} \leq e^{-\frac{2(\delta^2-1)}{k+1}}
  \end{align}
  where $\gamma \coloneqq \max \big(\frac{1}{m_a+1}+\frac{1}{m_b+1}, \frac{1}{k+1}+\frac{1}{m_a+m_b-n+1}\big) \geq \frac{1}{k+1}$.

\end{proof}

\subsection{Proof of Theorem~\ref{thm_const_increase_utility}}\label{sec_proof_thm_const_increase_utility}
We begin by restating Theorem~\ref{thm_const_increase_utility}.
\begin{theorem}
  {\bf (Restatement of Theorem~\ref{thm_const_increase_utility}).}
  Given a $\beta\in(0,1)$,
  if $\mathcal{D}\coloneqq \cU[0,1]$,
  $v$ satisfies Assumptions~\eqref{assumption} and \eqref{assumption_2} with $\epsilon>0$,
  and $n\leq \min(m_a,m_b)$,
  then for $\alpha^\star\coloneqq \frac{m_b}{m_a+m_b}$, then adding $L(\alpha^\star)$-constraints achieve nearly optimal latent utility in expectation:
  \begin{align*}
    U_{\mathcal{D},v}(\alpha^\star, \beta)=U_{\mathcal{D},v}(0,1) \cdot (1-O(n^{-\epsilon/2}+n^{-1}).
  \end{align*}
\end{theorem}
We present the proof of the special case where $m_a=m_b=n$.
The general proof follows from the same outline, by substituting values of $m_a$ and $m_b$ in the appropriate equations.
We use Proposition~\ref{thm_expected_cand} and \ref{fact_negative_exp} in the proof of the special case.
We present their analogues for the general case in Section~\ref{sec_extending_proof}.

\subsubsection{Notation}
Recall that we have two groups of candidates $G_a$ and $G_b$, with $|G_a|\coloneqq n$ and $|G_b|\coloneqq n$.
$G_b$ is the underprivileged group.
The latent utility of candidates in both groups is drawn from the uniform distribution on $[0,1]$.
Let $w_i$ be the random variable representing the latent utility of a candidate $i\in G_a\cup G_b$.
Due to implicit bias the observed utility of a candidate $j\in G_b$, $\hat{w}_j$  is their latent utility, $w_j$, multiplied by $\beta\in [0,1]$, whereas the observed utility of a candidate $i\in G_a$  is equal to their latent utility, $w_i$.
We consider the top $\ell\in [n]$ candidates in decreasing order of their observed utilities.
Let $S_a^\ell\subseteq G_a$ and $S_b^\ell\subseteq G_b$ be the set of candidates selected from $G_a$ and $G_b$ respectively.
Let $S_{a_1}^\ell, S_{a_2}^\ell\subseteq S_a^\ell$, such that $S_{a_1}^\ell \cup S_{a_2}^\ell = S_a^\ell$, be the set of candidates selected from $G_a$ with utilities larger or equal to and strictly smaller than $\beta$ respectively.
We define the following random variables counting the number of candidates selected
\begin{enumerate}[itemsep=0pt]
  \item $N_a^\ell\hspace{1.25mm}\coloneqq |S_a^\ell|$
  \item $N_b^\ell\hspace{1.25mm}\coloneqq |S_b^\ell|$
  \item $N_{a_1}^\ell\coloneqq |S_{a_1}^\ell|$
  \item $N_{a_2}^\ell\coloneqq |S_{a_2}^\ell|$.
\end{enumerate}
When considering the first $n$ positions we drop $\ell$ from the superscripts of the variables.

Recall
\begin{align}
  x^\star(v,w)\coloneqq \argmax_{x}\mathcal{W}(x,v,w)
\end{align}
is the ranking that maximizes the latent utility.
Further, let
\begin{align}
  \picons(\alpha^\star, v,w)\coloneqq\argmax_{x\in \mathcal{K}(L(\alpha^\star))}\mathcal{W}(x,v,\hat{w})
\end{align}
be the optimal constrained ranking optimizing $\hat{w}$, and
\begin{align}
  \piuncons(v,w)\coloneqq \argmax_{x}\mathcal{W}(x,v,\hat{w})
\end{align}
be the optimal unconstrained ranking maximizing $\hat{w}$.
When $v$ and $w$ are clear from the context we write $x^\star$, $\piuncons$, and when $\alpha^\star$ is also clear we write $\picons$.

Let $\delta>0$ be some number greater than $0$.
We will fix $\delta$ later in the proof.
Define a partition $\bigcup_{k=1}^{n^{1-\delta}}[(k-1)n^{\delta}+1,kn^{\delta}]$ of $[n]$, into $n^{1-\delta}$ intervals of size $n^{\delta}$.
Let
$$s_k\coloneqq (k-1)n^{\delta}+1\ \text{and} \ e_k\coloneqq kn^{\delta}$$
be starting and ending index of the $k$-th interval.
\subsubsection{Useful propositions}
\noindent We use the following two propositions in the proof of Theorem~\ref{thm_const_increase_utility}.
\begin{proposition}\label{fact_negative_exp}
  Let $N\sim \text{Binomial}(n,1-\beta)$, then
  \begin{align*}
    \eE\Bracket{  \frac{n}{2n-N} } &= \frac{1}{1+\beta}+ \ohalf,\label{fact_expec_1}\numberthis\\
    \eE\Bracket{  \frac{n^2}{(2n-N)^2} } &= \frac{1}{(1+\beta)^2}+ \ohalf.\label{fact_expec_2}\numberthis
  \end{align*}
\end{proposition}
\begin{proposition}\label{thm_expected_cand}
  {\bf (Expected number of candidates without constraints).}
  Given a number $\beta\in(0,1]$, representing the implicit bias,
  we have
  \begin{align}
    &N_{a_1}\stackrel{d}{=} \text{Binomial}(n,1-\beta),\label{eq_dist_n_a_1}\\
    &\eE[N_b] = n\Paren{\frac{\beta}{1+\beta}+\ohalf},\label{eq_exp_n_b}\\
    &\Pr\bracket{ |N_b - \eE[N_b]| \leq n^{\nfrac{5}{8}} } = 1-\oexp.\label{eq_conc_n_b}
  \end{align}
\end{proposition}

\subsubsection{Proof overview}
Define $k\coloneqq \ceil{\nfrac{2i}{n^{\delta}}}$.
We first show that $\picons$ places the $i$-th order statistics from both $G_a$ and $G_b$ inside the interval $[s_k, e_k]$ with high probability (see Lemma~\ref{HwithHP}).
From this, it follows that the positions of these candidates is within $e_k-s_k=O(n^{\delta})$ positions of $2i$ with high probability.

We compare the product of the utility and position discount of the $i$-th order statistic from $G_a$ and from $G_b$ in $\picons$, with that of the candidates in the $(2i-1)$-th and $2i$-th positions in $\piopt$.
We show that this difference is ``small" using Assumption~\eqref{assumption_2}.
This allows us to bound $U_{\mathcal{D},v}(0,1)-U_{\mathcal{D},v}(\alpha^\star,\beta)$ by
$$O\Paren{(n^{\delta-\epsilon}+n^{-1})\cdot\sum_{k=1}^n v_k}.$$
A lower bound
$$U_{\mathcal{D},v}(0,1)=\Omega\Paren{\sum_{k=1}^n v_k}$$
simply follows from the fact that for all $i\in [n]$
$$\eE[w_{(i\colon m_a+m_b)}]\geq \frac{1}{2}$$
which in turn follows from $2n\leq(m_a+m_b)$.
Finally, choosing $\delta=\nfrac{\epsilon}{2}$ and using Assumption~\eqref{assumption} gives us the required result.

This overview simplifies some details, for instance, we calculate the utility conditioned on the event that the $i$-th order statistic of $G_a$ and $G_b$ are close to $2i$.
This event is not independent of the utilities.
In the proof, we carefully ensure that this conditioning does not alter the utilities by ``a lot''.
Here, we use Proposition~\ref{thm_expected_cand} and Proposition~\ref{fact_negative_exp}.

\begin{remark}
  {\bf (Generalizing the proof).}
  More generally, the proof proceeds by showing that $\picons$ places the $i$-th order statistic from $G_a$ in the interval $[s_k,e_k]$ for $k\coloneqq \ceil{i(m_a+m_b)n^{-\delta}\cdot\frac{1}{m_a}}$, and from $G_b$ in the interval $[s_k,e_k]$ for $k\coloneqq \ceil{i(m_a+m_b)n^{-\delta}\cdot\frac{1}{m_b}}$.
\end{remark}

\subsubsection{Proof}\label{proof_main}
\begin{proof}
  Let $\cH_k$ be the event that $\picons$ places exactly $(\nfrac{n^{\delta}}{2})$ candidates from $G_a$ in the $k$-th interval, $[s_k, e_k]$, i.e.,
  $$\sum\nolimits_{i\in G_a}\sum\nolimits_{j=s_k}^{e_k}(\picons)_{ij}=\frac{n^{\delta}}{2}.$$
  Define the event  $\cH$ as the intersection
  $$\cH\coloneqq\bigcap_{k=1}^{n^{1-\delta}}\cH_k.$$
  Conditioned on the union $\cH$, each interval $[s_k, e_k]$ in $\picons$ contains exactly $\frac{e_k-s_k+1}{2}$ candidates from $G_a$ and from $G_b$.
  Therefore, $\picons$ ranks the $i$-th order statistic both from $G_a$ and from $G_b$ in the interval $[s_k, e_k]$ for $k\coloneqq\ceil{\nfrac{2i}{n^{\delta}}}$.

  \noindent Let $N_{b}^\ell$ be the number of candidates from $G_b$ placed in the first $\ell\in [n]$ positions of $\piuncons$.
  From Proposition~\ref{prop_pick_propotional} we know that if $N_{b}^\ell\leq \nfrac{\ell}{2}$, then the $\picons$ places exactly $\nfrac{\ell}{2}$ candidates from $G_b$ in the first $\ell$ positions.
  Extending this observation to all $1<k\leq n$ it follows that,
  \begin{align*}
    \big(N_{b}^{e_k}\leq \nfrac{e_k}{2}\big)\land\big(N_{b}^{e_{k-1}}\leq \nfrac{e_{k-1}}{2}\big)&\implies \cH_k\\
    \big(N_{b}^{e_1}\leq \nfrac{e_1}{2}\big)&\implies \cH_1.
  \end{align*}
  The first implication follows since LHS implies
  $$\sum_{i\in G_a, j\leq e_k}(\picons)_{ij}=\frac{e_k}{2}\ \ \text{and} \ \sum_{i\in G_a, j\leq e_{k-1}}(\picons)_{ij}=\frac{e_{k-1}}{2},$$
  or equivalently
  $$\sum_{i\in G_a}\sum_{j=s_k}^{e_k}(\picons)_{ij}=\frac{(e_k-e_{k-1})}{2}=\frac{(e_k-s_k+1)}{2}.$$
  Now, taking the intersection of the above events we have
  \begin{align}
    \bigcap\nolimits_{k=1}^{n^{1-\delta}}\Paren{ N_b^{e_k}\leq \frac{e_k}{2}}\implies \bigcap\nolimits_{k=1}^{n^{1-\delta}}\cH_k \iff \cH.  \label{eq_union_of_events}
  \end{align}

  \begin{lemma}\label{HwithHP}
    Given $m_a,m_b\geq n$, $\Pr[\cH]\geq 1-O(e^{-\Theta(n)})$.
  \end{lemma}
  \begin{proof}
    \noindent To show that $\cH$ occurs with high probability we calculate $\Pr[N_b^\ell - \eE[N_b^\ell] \leq k ]$.
    Setting $\ell\coloneqq e_k$, this allows us to bound $\Pr\big[N_b^{e_k}\leq \frac{e_k}{2}\big]$.
    Then Equation~\eqref{eq_union_of_events} this gives us a bound on $\Pr[\cH]$.

    Following the proof of Lemma~\ref{thm_expected_cand}, but only considering the first $\ell$ (instead of $n$) positions (and assuming $\ell> N_{a_1}$),
    we can derive the following equivalents to Equations~\eqref{equivalent_expectation} and \eqref{conc_bound_a}.
    If $\ell > N_{a_1}$, then
    \begin{align*}
      \eE[N_b^\ell | N_{a_1}] &= \frac{n(\ell-N_{a_1})}{2n-N_{a_1}}\numberthis\label{expectation_derive}\\
      \Pr[N_b^\ell - \eE[N_b^\ell] \leq k | N_{a_1}] &\stackrel{}{=} 1- O(e^{-\frac{2(k^2-1)}{\ell-N_{a_1}}}).\numberthis\label{concentration_derive}
    \end{align*}
    These follow by substituting $\ell$ for $n$, whenever $n$ does not refer to the size of the groups.
    The case when $\ell\leq N_{a_1}$ is straightforward to analyze:
    Since there are more than $\ell$ candidates in $G_a$ with high utilities (in $[\beta,1]$),
    no candidate from $G_b$ is selected in the first $\ell$ positions, i.e., $N_b^\ell=0$.
    Therefore we have
    \begin{align*}
      \eE[N_b^\ell | N_{a_1}] &\stackrel{}{=} \begin{cases}
      \frac{n(\ell-N_{a_1})}{2n-N_{a_1}}& \text{if } \ell> N_{a_1}\\
      0 & \text{otherwise.}
      \end{cases}\\
      \eE[N_b^\ell | N_{a_1}]&\leq \frac{n(\ell-N_{a_1})}{2n-N_{a_1}}\\
      &=n+\frac{n(\ell-2n)}{2n-N_{a_1}}.
    \end{align*}
    Taking the expectation of the above expression we have
    \newcommand{\ohalfl}{O\paren{\ell^{-\nfrac{3}{8}}}}
    \begin{align*}
      \hspace{20mm}\hspace{0mm}\eE[N_b^\ell ] &\ \stackrel{\text{Prop.}~\ref{fact_negative_exp}}{\leq}\ n+\frac{n(\ell-2n)}{2n-\eE[N_{a_1}]}\cdot\big(1+\ohalfl\big)\\
      &\hspace{20mm}\hspace{-14mm} \stackrel{}{\leq}\  \frac{\ell-n(1-\beta)}{1+\beta}\cdot
      \big(1+\ohalfl\big).\label{eq_prf_3_3_exp}\labelthis{Using $\eE[N_{a_1}]=n(1-\beta)$ from \text{Prop.}~\ref{thm_expected_cand}.}
    \end{align*}
    To find the distribution of $\big(N_b^\ell|N_{a_1}\big)$ when $\ell>N_{a_1}$, we can substitute $n$ as $\ell$ in all arguments in Section~\ref{proof_of_eq_dist_n_a_1}, specifically from Equation~\eqref{conc_bound_a} to Equation~\eqref{equation_70}.
    This gives us
    \begin{align*}
      \Pr[N_b^\ell - \eE[N_b^\ell] \leq \Theta(n) \ |\ N_{a_1}, \ell > N_{a_1}] &\stackrel{}{=}  1- O(e^{-\frac{2(\Theta(n^2)-1)}{\ell-N_{a_1}}}).
    \end{align*}
    When $\ell\leq N_{a_1}$, there are more than $\ell$ candidates in $G_a$ with high utilities in $[\beta,1]$.
    Therefore, no candidate from $G_b$ is selected in the first $\ell$ positions, i.e., $N_b^\ell=0$.
    In other words, $\Pr[N_b^\ell = \eE[N_b^\ell] \ |\ N_{a_1}, \ell \leq N_{a_1}]=1$.
    Following this,
    we have
    \begin{align*}
      \Pr[N_b^\ell - \eE[N_b^\ell] \leq \Theta(n) \ |\ N_{a_1}] &\stackrel{}{=} \begin{cases}
      1- O(e^{-\frac{2(\Theta(n^2)-1)}{\ell-N_{a_1}}}) & \text{if } \ell>N_{a_1}\\
      1 & \text{otherwise}
      \end{cases}\\
      &=\begin{cases}
      1- O(e^{-\frac{2(\Theta(n^2)-1)}{O(n)}}) & \text{if } \ell>N_{a_1}\\
      1 & \text{otherwise}
      \end{cases}\tag{Using $\ell, N_{a_1}=O(n)$}\\
      &\geq 1- O(e^{-\Theta(n)}).\tag{Using $k=\Theta(n)$.}
    \end{align*}
    Since the RHS holds for all values of $N_{a_1}$, we can drop the conditioning in the LHS to get
    \begin{align*}
      \Pr[N_b^\ell - \eE[N_b^\ell] \leq \Theta(n)] \geq 1- O(e^{-\Theta(n)}). \label{eq_concentration_boundRec_1}\numberthis
    \end{align*}
    Substituting $\ell\coloneqq e_k$, we can bound $\Pr[N_b^{e_k}\leq \nfrac{e_k}{2}]$ as follows
    \begin{align*}
      \Pr\big[N_b^{e_k}\leq \nfrac{e_k}{2}\big]&\stackrel{\myeqref{eq_prf_3_3_exp}{46}}{\geq}\Pr\bigg[N_b^{e_k}-\eE[N_b^{e_k}]\leq \frac{e_k}{2}-\frac{e_k-n(1-\beta)}{1+\beta}\cdot\big(1+\ohalf\big)\bigg]\\
      &\hspace{1mm}=\Pr\bigg[N_b^{e_k}-\eE[N_b^{e_k}]\leq \frac{e_k(\beta-1)}{2(1+\beta)}\cdot(1+\ohalf\big)+\frac{n(1-\beta)}{1+\beta}\cdot\big(1+\ohalf\big)\bigg]\\
      &\hspace{1mm}\geq\Pr\Bracket{N_b^{e_k}-\eE[N_b^{e_k}]\leq \frac{n(1-\beta)}{2(1+\beta)}\cdot\big(1+\ohalf\big)}\tag{$\because e_k\leq n$}\\
      &\hspace{1mm}=\Pr\bigg[N_b^{e_k}-\eE[N_b^{e_k}]\leq \Theta(n)\bigg].
    \end{align*}
    Using Equation~\eqref{eq_concentration_boundRec_1} and that $e_k\leq n$, it follows that
    \begin{align*}
      \Pr\big[N_b^{e_k}\leq \nfrac{e_k}{2}\big]&\geq 1-O(e^{-\Theta(n)})\\
      \Pr\big[N_b^{e_k}> \nfrac{e_k}{2}\big]&\leq O(e^{-\Theta(n)})\numberthis\label{bound_on_nks}.
    \end{align*}
    From Equation~\eqref{eq_union_of_events} we have
    \begin{align*}
      \Pr\big[ \cH\big]&=\Pr\bigg[ \bigcap\nolimits_{k=1}^{n^{1-\delta}}\cH_k \bigg]\\
      &= \Pr\bigg[ \bigcap\nolimits_{k=1}^{n^{1-\delta}}\big[N_b^{e_k}\leq \nfrac{e_k}{2}\big]\ \bigg]\\
      &= 1-\Pr\bigg[ \bigcap\nolimits_{k=1}^{n^{1-\delta}}\big[N_b^{e_k} > \nfrac{e_k}{2}\big]\ \bigg]\\
      &\geq 1-\sum_{k=1}^{n^{1-\delta}}\Pr\big[N_b^{e_k} > \nfrac{e_k}{2}\big]\\
      &\hspace{-1mm}\stackrel{\eqref{bound_on_nks}}{=}1-\sum_{k=1}^{n^{1-\delta}}O(e^{-\Theta(n)})\\
      &\stackrel{}{=}1-n^{1-\delta}O(e^{-\Theta(n)})\\
      &\stackrel{}{=}1-O(e^{-\Theta(n)+(1-\delta)\log n})\\
      &\stackrel{}{=}1-O(e^{-\Theta(n)}).\label{eq_cH_whp}\numberthis
    \end{align*}
    This proves the Lemma~\ref{HwithHP}.
  \end{proof}
  \noindent{\bf Notation.} For all $i\in G_a$, let $a(i)\in[n]$ denote the position of the $i$-th order statistic  in $G_a$ in $\picons$.
  Similarly, for all $j\in G_b$, let $b(i)\in[n]$ denote the position of the $i$-th order statistic in $G_b$ in $\picons$.
  Conditioned on $\cH$, we know that $\picons$ selects exactly $\nfrac{n}{2}$ candidates from $G_a$ and $G_b$ (see Propositions~\ref{prop_pick_propotional}),
  and that for all $i\in [\nfrac{n}{2}]$, $a(i),b(i)\in [s_k,e_k]$ where $k=\ceil{\nfrac{2i}{n^{\delta}}}$.
  Let $W$ be a random utility drawn from $\mathcal{D}$.
  We use $W_{(n-k:n)}$ to represent the $k$-th order statistic from $n$ \iid draws of $W$: the $k$-th largest value from $n$ \iid draws of $W$.

  For all $\ell\in [n]$, define $U_{\ell,a}$ and $U_{\ell,b}$ as follows
  \begin{align*}
    U_{\ell,a}\coloneqq \sum_{i=1}^{\ell}W_{(i:n)} v(a(i)) \text{ and } U_{\ell,b}\coloneqq \sum_{i=1}^{\ell}W_{(i:n)} v(b(i)).
  \end{align*}
  Here $U_{\ell,a}$ represents sum of the latent utilities of the top $\ell$ candidates from $G_a$ weighted by their position discount.
  Likewise, $U_{\ell,b}$ represents sum of the latent utilities of the top $\ell$ candidates from $G_b$ weighted by their position discount.
  From Proposition~\ref{prop_pick_propotional} the expected latent utility of the constrained ranking given $\cH$ is
  \begin{align*}
    \eE\nolimits_{w}[\mathcal{W}(\picons, v, w) | \cH]=\eE[U_{\frac{n}{2},a}+U_{\frac{n}{2},b}\ |\ \cH].
  \end{align*}
  However, $\cH$ is not independent of $U_{\frac{n}{2},a}$ and $U_{\frac{n}{2},b}$.
  Therefore, we cannot calculate $\eE[U_{\text{cons}}\ | \ \cH]$, directly from $\eE[U_{\frac{n}{2},a}]$ and $\eE[U_{\frac{n}{2},b}]$.
  We overcome this by showing that we show that conditioning on $\cH$ changes the expectation by a small amount for large $n$.
  We can approximate $\eE[U_{\frac{n}{2},a}+U_{\frac{n}{2},b}\ |\ \cH]$ using the following equality
  \begin{align*}
    \eE[\ U_{\frac{n}{2},a}+U_{\frac{n}{2},b}]&\proofHspace= \eE[\ U_{\frac{n}{2},a}
    +U_{\frac{n}{2},b}\ |\ \cH] \Pr[\cH] + \eE[\ U_{\frac{n}{2},a}+U_{\frac{n}{2},b}\ |\ \cHC] \Pr[\cHC]\\
    &\proofHspace= \eE[\ U_{\frac{n}{2},a}+U_{\frac{n}{2},b}\ |\ \cH] \Pr[\cH] + O(n) \Pr[\cHC]\\
    &\proofHspace\hspace{-3.5mm}\stackrel{{\rm Lem.}~\ref{HwithHP}}{=} \eE[\ U_{\frac{n}{2},a}+U_{\frac{n}{2},b}\ |\ \cH] \big(1- O(e^{-\Theta(n)})\big) + O(ne^{-\Theta(n)})\label{eq_prf_3_3_expect_util}\numberthis\\
    \eE\nolimits_{w}[\mathcal{W}(\picons, v, w) | \cH]&\proofHspace =\eE[\ U_{\frac{n}{2},a}+U_{\frac{n}{2},b}\ |\ \cH] \\
    &\proofHspace\hspace{-1mm}\stackrel{\eqref{eq_prf_3_3_expect_util}}{=} \bigg(\eE[\ U_{\frac{n}{2},a}
    +U_{\frac{n}{2},b}]-  O\big(ne^{-n^{2\delta}}\big)\bigg) \cdot \big(1- O(e^{-\Theta(n)})\big)^{-1}\\
    &\proofHspace=\bigg(\eE[\ U_{\frac{n}{2},a}+U_{\frac{n}{2},b}]-  O\big(ne^{-n^{2\delta}}\big)\bigg) \cdot \big(1+ O(e^{-\Theta(n)})\big)\\
    &\proofHspace=\eE[\ U_{\frac{n}{2},a}+U_{\frac{n}{2},b}] \cdot \big(1+ O(e^{-\Theta(n)})\big) -  O(ne^{-\Theta(n)}).\label{approximation_of_utility_aa}\numberthis
  \end{align*}
  We would later find a $O(e^{\delta-\epsilon})$ factor (see Equation~\eqref{eq_bound_conditioned}) in our analysis.
  Note that we can ignore the $O(ne^{-\Theta(n)})$ factors in the above equation in front of $O(e^{\delta-\epsilon})$.
  For ease of notation, we omit these factors in the next few equations, and approximate
  $\eE\nolimits_{w}[\mathcal{W}(\picons, v, w) | \cH]$ by $\eE[\ U_{\frac{n}{2},a}+U_{\frac{n}{2},b}]$.
  We explicitly account for them in Equation~\eqref{adding_errors_back}.
  Now, we can calculate $\eE\nolimits_{w}[\mathcal{W}(\picons, v, w) | \cH]$ as follows
  \newcommand{\tm}{\paren{\frac{j+s_k+1}{2}:n} }
  \newcommand{\ta}{\paren{\nfrac{1}{2}\cdot(j+s_k+1)} }
  \begin{align*}
    \eE\nolimits_{w}[\mathcal{W}(\picons, v, w) | \cH]&=\eE[\ U_{\frac{n}{2},a}+U_{\frac{n}{2},b}]\\
    &= \eE\big[\ \sum_{i=1}^{\nfrac{n}{2}}W_{(i:n)} v(a(i)) \ \big]+\eE\big[\ \sum_{i=1}^{\nfrac{n}{2}}W_{(i:n)} v(b(i)) \ \big]\\
    &=\sum_{i=1}^{\nfrac{n}{2}}\Paren{\eE\bracket{W_{(i:n)} v(a(i))}+\eE\bracket{W_{(i:n)} v(b(i))}}.
  \end{align*}
  Before rewriting the sum, we note that for ease of notation we denote $\floor{\ta}$ by $\ta$ in the following equations.
  \begin{align*}
    \eE\nolimits_{w}[\mathcal{W}(\picons, v, w) | \cH]&\proofHspace=\frac{1}{2}\sum_{k=1}^{n^{1-\delta}}\sum_{j=0}^{n^{\delta}}\bigg(\hspace{3.5mm}\eE\bracket{ W_{\tm} v\paren{a\ta} }\\
    &\hspace{25.6mm}+\eE\bracket{ W_{\tm}  v\paren{b\ta}} \bigg)\\
    &\proofHspace\geq \frac{1}{2}\sum_{k=1}^{n^{1-\delta}}\sum_{j=0}^{n^{\delta}}\bigg(\eE\bracket{W_{\tm}\cdot  v(e_k)}+\eE\bracket{W_{\tm} \cdot v(e_k)} \  \bigg).
  \end{align*}
  Here, the last inequality follow since given $\cH$, for all $j<n^{\delta}$ and $k\in [n^{1-\delta}]$, $a\ta\in [s_k, e_k]$ and $b\ta\in [s_k, e_k]$.
  Next we have
  \begin{align*}
    \eE\nolimits_{w}[\mathcal{W}(\picons, v, w) | \cH]&\geq\sum_{k=1}^{n^{1-\delta}}\sum_{j=0}^{n^{\delta}}\Paren{\eE\bracket{\ W_{\tm} \ }\cdot  v(e_k)\ }.\label{const_utility_expr}\numberthis
  \end{align*}
  Consider $\piopt\coloneqq \argmax_x \mathcal{W}(x,v,w)$.
  By our assumption that $v_k\geq v_{k+1}$, it follows that $\piopt$ places candidates in decreasing order of their latent utility ($w$), i.e., it places the $i$-th order statistic from $G_a\cup G_b$ at the $i$-th position.
  Let us denote the $i$-th order statistic from the $2n$ values by $V_{(i:2n)}$.
  Then $U_{\mathcal{D},v}(0,1)\coloneqq\eE\nolimits_{w}\big[\mathcal{W}(\piopt, v, w)\big] $, of $\piopt$ is
  \begin{align*}
    U_{\mathcal{D},v}(0,1)&= \eE\bigg[\ \sum_{i=1}^{n}V_{(i:2n)}\cdot v(i)\ \bigg] \\
    &= \sum_{i=1}^{n}\eE[V_{(i:2n)}]\cdot v(i)\\
    &= \sum_{k=1}^{n^{1-\delta}}\sum_{j=0}^{n^{\delta}} \eE[V_{(j+s_k:2n)}]\cdot v(j+s_k)\numberthis\label{optimal_utility_expr}.
  \end{align*}
  Using Equations~\eqref{const_utility_expr} and \eqref{optimal_utility_expr} we can calculate the difference in the expected utilities as
  \begin{align*}
    &U_{\mathcal{D},v}(0,1)-\eE\nolimits_{w}\big[\mathcal{W}(\picons, v, w) | \cH\big]\\
    &\hspace{-1mm}\stackrel{\eqref{const_utility_expr}, \eqref{optimal_utility_expr}}{\leq}  \sum_{k=1}^{n^{1-\delta}}\sum_{j=0}^{n^{\delta}} \eE[V_{(j+s_k:2n)}]\cdot v(j+s_k)
    -\sum_{k=1}^{n^{1-\delta}}\sum_{j=0}^{n^{\delta}}
    \eE\bracket{ W_{\tm} }\cdot  v(e_k)\\
    &\proofHspace=\sum_{k=1}^{n^{1-\delta}}\sum_{j=0}^{n^{\delta}}\Paren{
    \eE[V_{(j+s_k:2n)}]\cdot v(j+s_k)
    -\eE\bracket{\ W_{\tm} \ }\cdot  v(e_k)
    }\\
    &\proofHspace= \sum_{k=1}^{n^{1-\delta}}\sum_{j=0}^{n^{\delta}} \bigg(
    \bigg(1-\frac{j+s_k}{2n+1}\bigg)\cdot v(j+s_k)
    - \bigg(1-\frac{j+s_k+1}{2(n+1)}\bigg)\cdot  v(e_k) \bigg)
    \tag{$\forall\ x,y\in \cZ,\ $ $\eE[W_{(x:x+y)}]=\eE[V_{(x:x+y)}]=1-\frac{x}{x+y+1}$}\\
    &\proofHspace= \sum_{k=1}^{n^{1-\delta}}\sum_{j=0}^{n^{\delta}}
    \Paren{1-\frac{j+s_k+1}{2(n+1)} + \frac{2n+1-j-s_k}{2(n+1)(2n+1)}}\cdot v(j+s_k)
    - \sum_{k=1}^{n^{1-\delta}}\sum_{j=0}^{n^{\delta}}\Paren{1-\frac{j+s_k+1}{2(n+1)}}\cdot  v(e_k)\\
  \end{align*}
  \begin{align*}
    &\proofHspace= \sum_{k=1}^{n^{1-\delta}}\sum_{j=0}^{n^{\delta}}
    \Paren{1-\frac{j+s_k+1}{2(n+1)} + O(n^{-1})}\cdot v(j+s_k)
    - \sum_{k=1}^{n^{1-\delta}}\sum_{j=0}^{n^{\delta}}\Paren{1-\frac{j+s_k+1}{2(n+1)}}\cdot  v(e_k)\\
    &\proofHspace\leq \Paren{1-\frac{j+s_k+1}{2(n+1)}}\cdot \sum_{k=1}^{n^{1-\delta}}\sum_{j=0}^{n^{\delta}} \paren{
    v(j+s_k) -   v(e_k)}
    + O(n^{-1})\cdot\sum_{k=1}^{n^{1-\delta}}\sum_{j=0}^{n^{\delta}}v(j+s_k)\\
    &\leq O(1)\cdot\sum_{k=1}^{n^{1-\delta}}\sum_{j=0}^{n^{\delta}} \paren{v(j+s_k) - v(e_k)} + O(n^{-1})\cdot\sum_{i=1}^{n}v(i)\\
    &\proofHspace\hspace{-2mm}\stackrel{\eqref{assumption_2}}{\leq} O(1)\sum_{k=1}^{n^{1-\delta}}\sum_{j=0}^{n^{\delta}} O(n^{\delta})\paren{v(j+s_k) - v(j+s_k+1)}
    + O(n^{-1})\sum_{i=1}^{n}v(i) \tag{Using $|(j+s_k)-e_{k}|\leq n^{\delta}$}\\
    &\proofHspace\leq O(n^{\delta})\cdot\sum_{i=1}^{n}\paren{v(i) - v(i+1)} + O(n^{-1})\cdot\sum_{i=1}^{n}v(i)\\
    &\proofHspace\stackrel{\eqref{assumption}}{\leq} O(n^{\delta-\epsilon}+n^{-1})\cdot\sum_{i=1}^{n}v(i).\label{eq_difference_utility}\numberthis
  \end{align*}
  \noindent  To show this difference in utilities is much smaller than the unconstrained utility, consider $U_{\mathcal{D},v}(0,1)\coloneqq\eE\nolimits_{w}\big[\mathcal{W}(\piopt, v, w)\big]$,
  \begin{align*}
    U_{\mathcal{D},v}(0,1) &= \eE\bigg[\ \sum_{i=1}^{n}V_{(i:2n)}\cdot v(i)\ \bigg]\\
    &= \sum_{i=1}^{n}\eE\big[ V_{(i:2n)}\big]\cdot v(i)\\
    &= \sum_{i=1}^{n}\bigg(1-\frac{i}{2n+1}\bigg)\cdot v(i)\\
    &\geq \frac{1}{2}\sum_{i=1}^{n}v(i).\label{eq_sum_utility}\numberthis
  \end{align*}
  Using Equations~\eqref{eq_difference_utility} and \eqref{eq_sum_utility} we have
  \begin{align*}
    \eE\nolimits_{w}\big[\mathcal{W}(\picons, v, w) | \cH\big]&\geq U_{\mathcal{D},v}(0,1) - O(n^{\delta-\epsilon})\sum_{i=1}^{n}v(i)\\
    &\stackrel{\eqref{eq_sum_utility}}{=} U_{\mathcal{D},v}(0,1)\cdot \paren{1- O(n^{\delta-\epsilon}+n^{-1})}.
  \end{align*}
  Adding the additional errors terms from Equation~\eqref{approximation_of_utility_aa} we have
  \begin{align*}
    \eE\nolimits_{w}\big[\mathcal{W}(\picons, v, w) | \cH\big]&\proofHspace = U_{\mathcal{D},v}(0,1) \cdot \paren{1- O(n^{\delta-\epsilon}+n^{-1})}\cdot \paren{1 + O(e^{-\Theta(n)})} -  O(ne^{-\Theta(n)})\numberthis\label{adding_errors_back}\\
    &\proofHspace = U_{\mathcal{D},v}(0,1)\cdot \paren{1- O(n^{\delta-\epsilon}+n^{-1})+ O(e^{-\Theta(n)})} -  O(ne^{-\Theta(n)})\\
    \eE\nolimits_{w}\big[\mathcal{W}(\picons, v, w) | \cH\big]
    &\proofHspace=U_{\mathcal{D},v}(0,1)\cdot \paren{1- O(n^{\delta-\epsilon}+n^{-1})} -  O(ne^{-\Theta(n)}).\label{eq_bound_conditioned}\numberthis
  \end{align*}
  Finally, using the fact that $\cH$ occurs with high probability from Equation~\eqref{eq_cH_whp} we have
  \begin{align*}
    U_{\mathcal{D},v}(\alpha^\star,\beta)=\eE\nolimits_{w}\big[\mathcal{W}(\picons, v, w)]&\proofHspace= \eE\nolimits_{w}\big[\mathcal{W}(\picons, v, w) | \cH]\Pr[\cH]+\eE\nolimits_{w}\big[\mathcal{W}(\picons, v, w)  | \cHC]\Pr[\cHC]\\
    &\proofHspace= \eE\nolimits_{w}\big[\mathcal{W}(\picons, v, w)  | \cH]\Pr[\cH]+O(n)\Pr[\cHC]\\
    &\proofHspace\hspace{-5mm}\ \stackrel{\text{Prop.}~\ref{HwithHP}}{=}\eE\nolimits_{w}\big[\mathcal{W}(\picons, v, w)  | \cH]\cdot (1-O(e^{-\Theta(n)}))+O(n)\cdot O(e^{-\Theta(n)})\\
    &\proofHspace= \eE\nolimits_{w}\big[\mathcal{W}(\picons, v, w) | \cH]+O(ne^{-\Theta(n)})\\
    &\proofHspace\stackrel{\eqref{eq_sum_utility}}{=} U_{\mathcal{D},v}(0,1)\cdot \big(1- O(n^{\delta-\epsilon}+n^{-1})\big)-O(ne^{-\Theta(n)})+O(ne^{-\Theta(n)}).
  \end{align*}
  Fixing $\delta=\nfrac{\epsilon}{2}$, we have
  \begin{align*}
    U_{\mathcal{D},v}(\alpha^\star,\beta) &= U_{\mathcal{D},v}(0,1)\cdot \paren{1- O(n^{-\epsilon/2}+n^{-1})}+O(ne^{-\Theta(n)})\\
    &=U_{\mathcal{D},v}(0,1)\cdot \paren{1- O(n^{-\epsilon/2}+n^{-1})}.
  \end{align*}
  This completes the proof of Theorem~\ref{thm_const_increase_utility} when $m_a=m_b=n$.
\end{proof}

\subsubsection{Proof of Proposition~\ref{fact_negative_exp}}\label{sec_proof_fact_negative_exp}
\begin{proof}
  Let $\mu\coloneqq \eE[N]=n(1-\beta)$.
  Let $\cE$ be the event that $|N-\mu| \leq n^{\nfrac{1}{2}+\delta}$.
  Using Hoeffding's inequality~\cite{boucheron2013concentration} and the fact that $N_{a_1}\stackrel{d}{=}\text{Binomial}(n,1-\beta)$ from Proposition~\ref{thm_expected_cand} we have
  \begin{align*}
    \Pr[\cE]=\Pr[\ |N_{a_1}-\mu| \leq n^{\nfrac{1}{2}+\delta}\ ] &\stackrel{    }{\geq} 1-2e^{-2n^{2\delta}}\label{eq_conc_of_binomial_1}\numberthis\\
    \Pr[\cEC]=\Pr[\ |N_{a_1}-\mu| > n^{\nfrac{1}{2}+\delta}\ ] &\stackrel{    }{\leq} 2e^{-2n^{2\delta}}.\label{eq_conc_of_binomial_2}\numberthis
  \end{align*}
  From these we can prove Equation~\eqref{fact_expec_1} as follows
  \begin{align*}
    \eE\bigg[ \frac{n}{2n-N} \bigg]&\proofHspace= \eE\bigg[ \frac{n}{2n-N} \bigg| \ \cE \bigg]\cdot\Pr[\cE] + \eE\bigg[ \frac{n}{2n-N} \bigg| \ \cEC  \bigg]\cdot \Pr[\cEC]\\
    &\proofHspace= \eE\bigg[ \frac{n}{2n-N} \bigg| \ \cE \bigg]\cdot\Pr[\cE] +  \Pr[\cEC]\tag{$\forall \ 0\leq N\leq n,\ \frac{n}{2n-N}\leq 1$\hspace{-5mm}}\\
    &\hspace{-1mm}\proofHspace\stackrel{\eqref{eq_conc_of_binomial_2}}{=}\hspace{2mm} \eE\bigg[  \frac{n}{2n-N} \bigg| \ \cE \bigg]\cdot\Pr[\cE] + \oexpd\\
    &\hspace{-1mm}\proofHspace\stackrel{\eqref{eq_conc_of_binomial_1}}{=} \eE\bigg[  \frac{n}{2n-N} \bigg| \ \cE \bigg]\cdot\big(1-\oexpd\big) + \oexpd\\
    &\hspace{-7.75mm}\proofHspace\hspace{2mm}\stackrel{\mu\coloneqq n(1-\beta)}{=}\ \ \eE\bigg[ \frac{n}{n(1+\beta)-(N-\mu)} \bigg| \ \cE \bigg]\cdot\big(1-\oexpd\big)+ \oexpd \\
    &\proofHspace= \eE\bigg[ \frac{1}{1+\beta}\cdot\frac{1}{1-\nfrac{(N-\mu)}{(n+n\beta)}} \bigg| \ \cE \bigg]\cdot\big(1-\oexpd\big)+ \oexpd\\
    &\proofHspace= \eE\bigg[ \frac{1}{1+\beta}c\cdot\bigg(1+\frac{(N-\mu)}{n(1+\beta)}+O\bigg(\frac{(N-\mu)^2}{n^2}\bigg) \bigg) \bigg| \ \cE \bigg]\cdot\big(1-\oexpd\big)+ \oexpd	\tag{Using $\forall\ x\in [0,1), \frac{1}{1+x}=1-x+O(x^2)$}\\
    &\proofHspace= \eE\bigg[ \frac{1}{1+\beta}\cdot\bigg(1+ \ohalfd \bigg) \bigg| \ \cE \bigg] + \oexpd\tag{Given $\cE$, $N-\mu\leq n^{\nfrac{1}{2}+\delta}$}\\
    &\proofHspace= \frac{1}{1+\beta}+ \ohalfd+ \oexpd\\
    &\proofHspace\hspace{-1.5mm}\stackrel{\delta>0}{=} \frac{1}{1+\beta}+ \ohalfd.\numberthis\label{approximation}
  \end{align*}
  Setting $\delta\coloneqq \frac{1}{8}$ we have
  \begin{align*}
    \eE\bigg[ \frac{n}{2n-N} \bigg]\stackrel{}{=} \frac{1}{1+\beta}+ \ohalf.
  \end{align*}
  Further, we can prove Equation~\eqref{fact_expec_2} as follows
  \begin{align*}
    \eE\bigg[\frac{n^2}{(2n-N)^2}\bigg]&\proofHspace=\eE\bigg[\frac{n^2}{(2n-N)^2} \ \bigg|\  \cE\ \bigg]\cdot\Pr[\cE]+\eE\bigg[\frac{n^2}{(2n-N)^2} \ |\  \cEC\ \bigg]\cdot\Pr[\cEC]\\
    &\proofHspace=\eE\bigg[\frac{n^2}{(2n-N)^2} \ \bigg|\  \cE\ \bigg]\cdot\Pr[\cE]+\Pr[\cEC]\tag{$\forall \ 0\leq N\leq n,\ $ $\frac{n}{2n-N}\leq 1$}\\
    &\proofHspace\stackrel{\eqref{eq_conc_of_binomial_2}}{=}\eE\bigg[\frac{n^2}{(2n-N)^2} \ \bigg|\  \cE\ \bigg]\cdot\Pr[\cE]+\oexpd\\
    &\proofHspace\stackrel{\eqref{eq_conc_of_binomial_1}}{=}\eE\bigg[\frac{n^2}{(2n-N)^2} \ \bigg|\  \cE\ \bigg]\cdot\big(1-\oexpd\big)+\oexpd\\
    &\hspace{-7.75mm}\proofHspace\hspace{2mm}\stackrel{\mu\coloneqq n(1-\beta)}{=}\ \ \eE\bigg[\frac{n^2}{(n(1+\beta)-(N-\mu))^2} \ \bigg|\  \cE\ \bigg]\cdot(1-\oexpd)+\oexpd\\
    &\proofHspace= \eE\bigg[\frac{1}{(1+\beta)^2}  \frac{1}{\big(1-\nfrac{(N-\mu)}{n(1+\beta)}\big)^2 }\ \bigg|\  \cE \bigg]\cdot(1-\oexpd)+\oexpd
  \end{align*}

  \begin{align*}
    \eE\bigg[\frac{n^2}{(2n-N)^2}\bigg]&\proofHspace= \eE\bigg[\frac{1}{(1+\beta)^2}  \bigg(1+\frac{(N-\mu)}{n(1+\beta)}+O\big(n^{-1+2\delta}\big)\bigg)^2 \ \bigg|\  \cE\ \bigg]\cdot(1-\oexpd)+\oexpd	\tag{Using $\forall x\in [0,1), \frac{1}{1+x}=1-x+O(x^2)$}\\
    &\proofHspace= \eE\bigg[\frac{1}{(1+\beta)^2} \cdot \paren{1+\ohalfd}^2 \ \bigg|\  \cE\ \bigg]\big(1-\oexpd\big)+\oexpd\tag{Given $\cE$, $N-\mu\leq n^{\nfrac{1}{2}+\delta}$}\\
    &\proofHspace= \eE\bigg[\frac{1}{(1+\beta)^2} \cdot \paren{1+\ohalfd} \ \bigg|\  \cE\ \bigg]\big(1-\oexpd\big)+\oexpd\\
    &\proofHspace= \frac{1}{(1+\beta)^2} \cdot \big(1+\ohalfd\big) \cdot \big(1-\oexpd\big)+\oexpd\\
    &\proofHspace\stackrel{\delta>0}{=} \frac{1}{(1+\beta)^2} \cdot \big(1+\ohalfd\big)+\oexpd\\
    &\proofHspace\stackrel{\delta>0}{=} \frac{1}{(1+\beta)^2} \cdot \big(1+\ohalfd\big).
  \end{align*}
  Setting $\delta\coloneqq \nfrac{1}{8}$ we have
  \begin{align*}
    \eE\bigg[\frac{n^2}{(2n-N)^2}\bigg]\stackrel{}{=} \frac{1}{(1+\beta)^2} \cdot \big(1+\ohalf\big).
  \end{align*}
\end{proof}

\subsubsection{Proof of Proposition~\ref{thm_expected_cand}}\label{sec_proof_thm_expected_cand}

\begin{proof}
  We prove each of the Equations~\eqref{eq_dist_n_a_1}, \eqref{eq_exp_n_b} and \eqref{eq_conc_n_b}  in the following sections.

  \paragraph{Proof of Equation~\eqref{eq_dist_n_a_1}.}\label{proof_of_eq_dist_n_a_1}

  We can express $N_{a_1}$ as a sum of the following $n$ random variables
  \begin{align*}
    N_{a_1} = \sum\nolimits_{i\in G_a } \mathbb{I}[W_i\in (\beta,1]].
  \end{align*}
  Since $W_i$  are drawn \iid from $\cU[0,1]$, it follows that for $W\stackrel{d}{=} \cU[0,1]$
  \begin{align*}
    N_{a_1} = n\cdot \mathbb{I}[W\in (\beta,1]].
  \end{align*}
  Since $\mathbb{I}[W\in (\beta,1]]$ a bernoulli random variable, which is 1 with probability $(1-\beta)$ and 0 with probability $\beta$,
  it follows that
  \begin{align*}
    N_{a_1}\stackrel{d}{=}\text{Binomial}(n,1-\beta).
  \end{align*}

  \paragraph{Proof of Equations~\eqref{eq_exp_n_b} and \eqref{eq_conc_n_b}.}

  \noindent We would require the following proposition
  \begin{proposition}\label{fact_conditioning_uniform}
    {\bf (Truncated uniform distribution). } Let $W\stackrel{d}{=} \cU[0,1]$, then for all $0<u<v<1$, $\big(W \ | \ W \in [u,v]\big)\stackrel{d}{=}\cU[u,v]$.
  \end{proposition}
  \begin{proof}
    For any $k\in [u,v]$ we have,
    \begin{align*}
      \Pr[\ W \leq k\ | \ W \in [u,v]\ ] &= \frac{ \Pr[\ W\leq k  \land W \in [u,v]\ ]}{\Pr[\ W\in [u,v]\ ]}
      = \frac{ k-u }{ v-u }.
    \end{align*}
    Comparing the above with the cdf of $\cU[u,v]$ we have $(W \ |\ W\in[u,v])\stackrel{d}{=}\cU[u,v]$.
  \end{proof}
  \noindent
  Conditioning on $N_{a_1}$ fixes the number of candidates (all from $G_a$) with utility in $[\beta,1]$.
  From Proposition~\ref{fact_conditioning_uniform} we have that for all $i\in G_a$ and $j\in G_b$
  \begin{align*}
    (w_i \ |\ w_i\in[0,\beta])\stackrel{d}{=}\cU[0,\beta]\stackrel{d}{=}w_j\beta\stackrel{d}{=}\hat{w}_j.
  \end{align*}
  In other words, all the candidates with utilities in $[0,\beta]$ ($n$ candidates from $G_b$ and $n-N_{a_1}$ candidates in $G_a$) have the same distribution of observed utilities conditioned on $N_{a_1}$.
  Consider the ordered list of these $2n-N_{a_1}$ candidates in decreasing order of their observed utilities (breaking ties with unbiased coin tosses).
  Since, all observed utilities are drawn \iid each permutation of these candidates is equally likely.

  We choose the top $n-N_{a_1}$ candidates from this list to complete our ranking (which has $N_{a_1}$ candidates as of now).
  The number of ways of choosing $k$ candidates from the $n$ in $G_b$ is $\binom{n}{k}$, and  the number of ways of choosing the other $n-N_{a_1}-k$ candidates from the $n-N_{a-1}$ in $G_a$ is $\binom{n-N_{a_1}}{n-N_{a_1}-k}$.
  Since each of these colorings is equally likely it follows that
  \begin{align*}
    \Pr[N_b=k \ | \  N_{a_1}] &= \frac{ \binom{n-N_{a_1}}{n-N_{a_1}-k} \binom{n}{k} }{ \binom{2n-N_{a_1}}{n-N_{a_1}} } = \frac{ \binom{n-N_{a_1}}{k} \binom{n}{n-k} }{ \binom{2n-N_{a_1}}{n} }.
  \end{align*}

  \begin{tcolorbox}[colback=yellow!10!white,colframe=black!75!black]
    Given numbers $N, K, n $, for an hypergeometric random variable, $HG$, we have: $\forall \max(0,n+K-N)\leq k \leq \min(K,n)$
    \begin{align}
      \Pr[HG = k] \coloneqq \frac{\binom{K}{k}\binom{N-K}{n-k}}{\binom{N}{n}}.\label{eq_hyper_2}
    \end{align}
  \end{tcolorbox}

  By comparing with Equation~\eqref{eq_hyper_2}, we can observe that conditioned on $N_{a_1}$, $N_b$ is a hypergeometric random variable.
  We would soon use this fact to show  that $N_b$ is concentrated around its mean.
  But, first we would calculate its mean\footnote{Note that using the mean of the hypergeometric distribution  we get $\eE[N_{b}|N_{a_1}]$ and not $\eE[N_{b}]$.}.

  Let $X_i(k)$ be a random variable which is 1 if candidate $i$ is in the first $k$ positions out of $2n-N_{a_1}$, and 0 otherwise.
  Since, each permutation of these candidates is equally likely it follows that
  \begin{align*}
    \text{For all }i\in (G_a\cup G_b)\backslash S_{a_1},\  \Pr[X_i(k)=1\ | \ N_{a_1}]=\frac{k}{2n-N_{a_1}}.
  \end{align*}
  Expressing $N_b$ as $\sum_{i\in G_b} X_i(n-N_{a_1})$ we have
  \begin{align*}
    \eE[N_b|N_{a_1}] &=\eE\bracket{\sum_{i\in G_b}X_i(n-N_{a_1}) | N_{a_1}}=\sum_{i\in G_b}\eE[X_i(n-N_{a_1}) | N_{a_1}]\\
    &=n\cdot \frac{n-N_{a_1}}{2n-N_{a_1}}.\numberthis\label{equivalent_expectation}
  \end{align*}
  Recall that $N_{a_1}\stackrel{d}{=}\text{Binomial}(n,1-\beta)$ from Proposition~\ref{thm_expected_cand}.
  Taking the expectation of these we get
  \begin{align*}
    \eE[N_{b} ] &= \eE[\ \eE[N_{b}|N_{a_1}] \ ]= n\eE\bigg[1-\frac{n}{2n-N_{a_1}}\bigg] .\numberthis\label{double_expectation_n_b}\\
    &\hspace{-3mm}\stackrel{\text{Fact}~\ref{fact_negative_exp}}{=}\ n \cdot\big(\nfrac{\beta}{1+\beta}+ \ohalf\big).
  \end{align*}

  This completes the proof of Equation~\eqref{eq_exp_n_b} and gives us the mean of $\eE[N_b]$.
  We proceed to show that $N_b$ is concentrated around this mean.

  From our earlier observation that $N_b$ follows a hypergeometric distribution and using known tail bounds for the hypergeometric distribution (see Theorem 1 in \cite{hush2005concentration}) we have
  \begin{align*}
    \Pr[\ |N_b - \eE[N_b]| \leq k \ | \  N_{a_1}\ ] &\stackrel{\text{Thm. 1 \cite{hush2005concentration}}   }{\geq}  1 - 2e^{-\frac{2(k^2-1)}{n-N_{a_1}} }.\label{conc_bound_a}\numberthis
  \end{align*}
  However, the above concentration bound is not very useful unless we remove the conditioning on $N_{a_1}$.
  To do this, we generalize the conditioning in the above to a lowerbound on $N_{a_1}$
  \begin{align*}
    \Pr\bigg[\ |N_b - \eE[N_b]| \leq k \ \bigg| \  N_{a_1} \geq \eE[N_{a_1}]-k \ \bigg].
  \end{align*}
  We lowerbound the above probability by choosing the smallest value of $N_{a_1}$, $N_{a_1}=\eE[N_{a_1}]-k$ in Equation~\eqref{conc_bound_a} to get
  \begin{align*}
    \Pr\bigg[\ |N_b - \eE[N_b]| \leq k \ \bigg| \  N_{a_1} \geq \eE[N_{a_1}]-k \ \bigg] \geq  1 - 2e^{-\frac{2(k^2-1)}{n-\eE[N_{a_1}]+k} }.
  \end{align*}
  Finally, we remove the conditioning in the above using the concentration of $N_b$ and a union bound.

  \noindent Let $\cE$ be the event that $N_{a_1} \geq \eE[N_{a_1}]-k$.
  Using Hoeffding's inequality~\cite{boucheron2013concentration} and the fact that $N_{a_1}\stackrel{d}{=}\text{Binomial}(n,1-\beta)$ we have
  \begin{align*}
    \Pr[\cE]=\Pr[\ N_{a_1}\geq \eE[N_{a_1}] -  k\ ] &\stackrel{    }{\geq} 1-e^{-\nfrac{2k^2}{n}}.\numberthis\label{conc_of_na1_b}
  \end{align*}
  Using this we can calculate $\Pr[\ |N_b - \eE[N_b]| \leq k \ ]$ as follows
  \begin{align*}
    \Pr\big[\ |N_b - \eE[N_b]| \leq k \ \big] &=  \Pr\big[\ |N_b - \eE[N_b]| \leq k \ \big| \  \cE \ \big] \cdot \Pr\big[ \cE \big] + \Pr\big[\ |N_b - \eE[N_b]| \leq k \ \big| \  \cEC \ \big] \cdot \Pr\big[ \cEC \big]\\
    & \geq \Pr\big[\ |N_b - \eE[N_b]| \leq k \ \big| \  \cE \ \big] \cdot \Pr\big[ \cE \big]\\
    & \hspace{-4mm}\stackrel{ \eqref{conc_bound_a},\eqref{conc_of_na1_b}}{\geq}  \big(1 - 2e^{-\frac{2(k^2-1)}{n-\eE[N_{a_1}]+k} } \big)\cdot\big( 1-e^{-\nfrac{2k^2}{n}} \big)\\
    & \geq  1 - 2e^{-\frac{2(k^2-1)}{n-\eE[N_{a_1}]+k} } - e^{-\nfrac{2k^2}{n}}\\
    &\stackrel{}{\geq}  1 - 2e^{-\frac{2(k^2-1)}{n\beta+k} } - e^{-\nfrac{2k^2}{n}}.\numberthis\label{equation_70}
  \end{align*}
  Substituting $k = n^{\nfrac{1}{2}+\delta}$ we get
  \begin{align*}
    \Pr\big[\ |N_b - \eE[N_b]| \leq n^{\nfrac{1}{2}+\delta} \ \big] &=  1 - 2e^{-\frac{2(k^2-1)}{n\beta+k} } - e^{-\nfrac{2k^2}{n}}\\
    &\geq 1 - 2e^{-\frac{k^2-1}{n\beta} } - e^{-\nfrac{2k^2}{n}}\tag{$k<n\beta$}\\
    &= 1 - 2e^{\frac{1}{n\beta} }e^{-\frac{k^2}{n\beta} } - e^{-\nfrac{2k^2}{n}}\\
    &= 1 - O(e^{-\nfrac{k^2}{n}})\tag{$\beta\in (0,1]$}\\
    &= 1 - \oexpd.\tag{$k\coloneqq n^{\nfrac{1}{2}+\delta}$}
  \end{align*}
  Setting $\delta=\nfrac{1}{8}$ have
  \begin{align*}
    \Pr\big[\ |N_b - \eE[N_b]| \leq n^{\nfrac{5}{8}} \ \big] = 1 - \oexp.
  \end{align*}
  This completes the proof of Equation~\eqref{eq_conc_n_b} and also of Proposition~\ref{thm_expected_cand}.
\end{proof}

\subsection{Proof of Equations~\myeqref{eq_utility_with_const}{16} and \myeqref{eq_utility_without_const}{17}}\label{proof_of_utilit_when_no_disc}
  We begin by restating the two equations
  \begin{align*}
    U_{\mathcal{D},v}(\alpha^\star, \beta)&=n\big(1-\frac{n}{2(m_a+m_b)}+\ohalf\big),\labelthis{Utility with constraints}\\
    U_{\mathcal{D},v}(0, \beta)&=\begin{cases}
    \frac{m_a(1-\beta^2)}{2} + \frac{m_a\beta^2+m_b}{2}\bigg[1-\frac{(m_a+m_b-n)^2}{(m_a\beta+m_b)^2}+\ohalf\bigg], &\hspace{-2mm}c=n-\omega(n^{\nfrac{5}{8}})\\
    n\big(1-\frac{n}{2m_a}+\ohalf\big),  &\hspace{-2mm}c\geq n+\Theta(n^{\nfrac{5}{8}}).\labelthis{Utility without constraints}
  \end{cases}
  \end{align*}
  Recall that these equations hold when $v_k=1 \ \forall \ k\in [n]$ and we define $c\coloneqq m_a(1-\beta)$.
  We first give the proof for the special case where $m_a=m_b=n$ and then show how to extend this proof to the general case with ($m_a\neq m_b\neq n$) in Section~\ref{sec_extending_proof}.
  We split this into the following lemmas, one for each equation.
  \begin{lemma}\label{lem_utility_unconst_app}
    {\bf (Expected latent utility without constraints).}
    Given $\beta\in(0,1]$, representing the implicit bias,
    we have $U_{\mathcal{D},1}(0,\beta)$ is
    \begin{align}
      U_{\mathcal{D},1}(0,\beta) = \frac{n}{2}\left(1+\frac{2\beta}{(1+\beta)^2}+\ohalf \right).
    \end{align}
  \end{lemma}
  \begin{lemma}\label{lem_utility_const_app}
    {\bf (Expected latent utility with constraints).}
    Given $\beta\in(0,1)$, representing the implicit bias,
    we have $U_{\mathcal{D},1}(\alpha^\star,\beta)$ is
    \begin{align}
      U_{\mathcal{D},1}(\alpha^\star,\beta) = \frac{3n}{4}\bigg(1-\frac{1}{n+1}\bigg)\bigg(1+\oexp\bigg).
    \end{align}
  \end{lemma}

\subsubsection{Notation}
We recall that we have two groups of candidates $G_a$ and $G_b$, with $|G_a|\coloneqq n$ and $|G_b|\coloneqq n$, where $G_b$ is the underprivileged group.
The latent utility of candidates in both groups $G_a$ and $G_b$ is drawn from the uniform distribution on $[0,1]$.
Let $w_i$ be the random variable representing the latent utility of a candidate $i\in G_a\cup G_b$.
Due to implicit bias the observed utility of a candidate $j\in G_b$, $\hat{w}_j$  is their latent utility, $w_j$, multiplied by $\beta\in [0,1]$, whereas the observed utility of a candidate $i\in G_a$  is equal to their latent utility, $w_i$.
We consider the top $n$ candidates in decreasing order of their observed utilities.
Let $S_a\subseteq G_a$ and $S_b\subseteq G_b$ be the set of candidates selected from $G_a$ and $G_b$ respectively.
Let $S_{a_1}\subseteq S_a$ and $S_{a_2}\subseteq S_a$ be the set of candidates selected from $G_a$ with utilities larger or equal to and strictly smaller than $\beta$ respectively.
Observe that $S_{a_1}\cup S_{a_2} = S_a$.
We define the following random variables counting the number of candidates selected
\begin{enumerate}[itemsep=0pt]
  \item $N_a\hspace{1.75mm} \coloneqq |S_a|$,
  \item $N_b\hspace{1.75mm} \coloneqq |S_b|$,
  \item $N_{a_1}\coloneqq |S_{a_1}|$,
  \item $N_{a_2}\coloneqq |S_{a_2}|$.
\end{enumerate}
Further define random variables $U_{a_1}$, $U_{a_2}$, and $U_{b}$ representing the utilities of selected candidates without position-discount as
\begin{enumerate}[itemsep=0pt]
  \item $U_a\hspace{1.75mm} \coloneqq \sum_{i\in S_a} w_{i}$,
  \item $U_b\hspace{1.75mm} \coloneqq \sum_{j\in S_b} w_{j}$,
  \item $U_{a_1}\coloneqq \sum_{i\in S_{a_1}} w_{i}$,
  \item $U_{a_2}\coloneqq \sum_{i\in S_{a_2}} w_{i}$.
\end{enumerate}
Here, $U_{a_1}$, $U_{a_2}$ and $U_{b}$ represent the latent utilities of the candidates selected from $G_a$ with observed utilities higher than $\beta$, those selected from the other $n-N_{a_1}$ candidates of $G_a$, and those selected from all $n$ candidates from $G_b$ respectively.

Let
$$x^\star(v,w)\coloneqq \argmax_{x}\mathcal{W}(x,v,w)$$
be the ranking that maximizes the latent utility.
Further, let
$$\picons(\alpha^\star, v,w)\coloneqq \argmax_{x\in \mathcal{K}(L(\alpha^\star))}\mathcal{W}(x,v,\hat{w})$$
be the optimal constrained ranking optimizing $\hat{w}$,
and
$$\piuncons(v,w)\coloneqq \argmax_{x}\mathcal{W}(x,v,\hat{w})$$
be the optimal unconstrained ranking maximizing $\hat{w}$.
When $v$ and $w$ are clear from the context we write $x^\star$, $\piuncons$, and when $\alpha^\star$ is also clear we write $\picons$.

Given a draw of utility, $w$, we have  $\mathcal{W}(\piuncons,v,w) = U_{a_1} + U_{a_2} + U_{b}$.
Then the total expected latent utilities, $U_{\mathcal{D},1}(0,\beta)$ and $U_{\mathcal{D},1}(\alpha^\star,\beta)$, are
\begin{enumerate}[itemsep=0pt]
  \item $U_{\mathcal{D},1}(0,\beta)\hspace{2.1mm} \coloneqq \eE\nolimits_{w}[\mathcal{W}(\piuncons,v,w)]= \eE[U_{a_1}+U_{a_2}+U_b]$
  \item $U_{\mathcal{D},1}(\alpha^\star,\beta)\coloneqq\eE_{w}[\mathcal{W}(\picons,v,w)]$.
\end{enumerate}

\subsubsection{Proof of Lemma~\ref{lem_utility_unconst_app}}\label{sec_proof_lem_utility_unconst_app}
\begin{proof}[Proof of Lemma~\ref{lem_utility_unconst_app}]
  We can calculate $\eE[U]$ assuming the following equations
  \begin{align}
    \eE[U_{a_1}] &= \frac{n}{2}(1-\beta^2)\label{eq_exp_ua1}\\
    \eE[U_{a_2}] &= \frac{n}{2}\bigg[\beta^2-\frac{\beta^2}{(1+\beta)^2}+\ohalf\bigg]\label{eq_exp_ua2}\\
    \eE[U_b] &= \frac{n}{2}\bigg[1-\frac{1}{(1+\beta)^2}+\ohalf\bigg].\label{eq_exp_ub}
  \end{align}
  Assuming the above we get
  \begin{align*}
    \eE[U]&=\eE[U_{a_1}+U_{a_2}+U_b]\\
    &=\eE[U_{a_1}]+\eE[U_{a_2}]+\eE[U_b]\\
    &\hspace{-7.2mm}\stackrel{\eqref{eq_exp_ua1}, \eqref{eq_exp_ua2}, \eqref{eq_exp_ub}}{=}\ \frac{n(1-\beta^2)}{2}+\frac{n}{2}\bigg[\beta^2-\frac{\beta^2}{(1+\beta)^2}+\ohalf\bigg]
    +\frac{n}{2}\bigg[1-\frac{1}{(1+\beta)^2}+\ohalf\bigg]\\
    &=\frac{3n}{4}\left(1+\frac{(1-\beta)^2}{3(1+\beta)^2}+\ohalf\right).
  \end{align*}
  We devote the rest of the proof to deriving Equations~\eqref{eq_exp_ua1}, \eqref{eq_exp_ua2}, and \eqref{eq_exp_ub}.
  \paragraph{Proof of Equation~\eqref{eq_exp_ua1}.}\label{proof_exp_ua1}
  Fix a value of $N_{a_1}$, by definition this also fixes the size of $S_{a_1}$.
  Then, from the definition of $U_{a_1}$ we have
  \begin{align*}
    \eE[U_{a_1} \ | \ N_{a_1}] &=  \eE\bracket{\textstyle{\sum\nolimits_{i\in S_{a_1}}} w_i \ | \ N_{a_1}}\\
    &= \eE\bracket{\textstyle{\sum\nolimits_{i\in S_{a_1}}} \eE[w_i \ | \ i\in S_{a_1}] \ | \ N_{a_1}}\\
    &= \eE\bracket{\textstyle{\sum\nolimits_{i\in S_{a_1}}} \eE[w_i \ | \ w_i \in [\beta, 1]] \ | \ N_{a_1}}.
  \end{align*}
  Define a random variable $W\stackrel{d}{=}\cU[0,1]$ then we have
  \begin{align*}
    \eE[U_{a_1} \ | \ N_{a_1}] &= \eE\bracket{\textstyle{\sum\nolimits_{i\in S_{a_1}}} \eE[W \ | \ W \in [\beta, 1]] \ | \ N_{a_1}}\\
    &= \eE\bracket{\textstyle{\sum\nolimits_{i\in S_{a_1}}1} \ | \ N_{a_1}} \cdot \eE[W \ | \ W \in [\beta, 1]]\\
    &= N_{a_1}\cdot \eE[W \ | \ W\in [\beta, 1]] \tag{$|S_{a_1}|\coloneqq N_{a_1}$.}\\
    &\hspace{-3.8mm}\stackrel{\text{Prop}~\ref{fact_conditioning_uniform}}{=} N_{a_1} \int_{\beta}^{1} \frac{x}{1-\beta} dx\\
    &= \frac{1}{2}\cdot N_{a_1}(1+\beta).
  \end{align*}

  \noindent Taking the expectation over $N_{a_1}$ we get
  \begin{align*}
    \eE[U_{a_1}] = \eE[\ \eE[U_{a_1} \ | \ N_{a_1}]\ ] = \frac{1}{2}\eE[N_{a_1}](1+\beta) \stackrel{\eqref{eq_dist_n_a_1}}{=} \frac{n}{2}(1-\beta^2).
  \end{align*}
  This proves Equation~\eqref{eq_exp_ua1}.

  \paragraph{Proof of Equations~\eqref{eq_exp_ua2} and \eqref{eq_exp_ub}.}
  We recall that all candidates in $G_a$ have the same observed and latent utility, whereas all candidates in $G_b$ have an observed utility which is $\beta$ times their latent utility.
  Therefore, $U_{a_2}$ and $\beta U_{b}$ are the total observed utility of candidates in $S_{a_2}$ and $S_b$.

  \noindent From Proposition~\ref{fact_conditioning_uniform} we know that for all $i\in G_a$ and $j\in G_b$
  \begin{align*}
    (w_i \ |\ w_i\in[0,\beta])\stackrel{d}{=}\cU[0,\beta]\stackrel{d}{=}w_j\beta\stackrel{d}{=}\hat{w}_j.
  \end{align*}
  In other words, that conditioned on $N_{a_1}$, the observed utilities of all candidates not in $S_{a_1}$ follows the uniform distribution on $[0,\beta]$.
  It follows that the expected observed utility, $\eE[U_{a_2}+\beta\cdot U_{b}]$, is the expected value of the $(n-N_{a_1})$ largest order statistics from a total of $(2n-N_{a_1})$ draws of $\cU[0,\beta]$.
  Define a random variable $Z_k$ to be the $k$-th largest observed utility from the union of the $(n-N_{a_1})$ candidates from $G_a\backslash S_{a_1}$ and $n$ candidates from $G_b$.
  Then we have
  \begin{align*}
    \eE[\ U_{a_2} \  | \ N_{a_1}\ ]+ \beta \eE[\ U_{b}\  | \ N_{a_1}\ ]&=\eE[\ U_{a_2}+\beta\cdot U_{b}\  | \ N_{a_1}\ ]\\
    &= \sum\nolimits_{k=1}^{n-N_{a_1}}\eE[Z_{k}]\\
    &= \sum\nolimits_{k=1}^{n-N_{a_1}}\int_{0}^{\beta}x\cdot x^{2n-N_{a_1}-k} (\beta-x)^{k-1}dx\\
    &= \beta\sum\nolimits_{k=1}^{n-N_{a_1}}\frac{2n-N_{a_1}+1-k}{2n-N_{a_1}+1}\\
    &= \beta\bigg(n-N_{a_1} - \frac{\binom{n-N_{a_1}+1}{2}}{2n-N_{a_1}+1}\bigg).\label{eq_74}\numberthis
  \end{align*}
  Further, from symmetry we expect the ratio of $\eE[U_{a_2} | N_{a_1}]$ and $\eE[ \beta\cdot U_{b}| N_{a_1}]$ will be the same as the ratio of the number of candidates, i.e., $\nfrac{(n-N_{a_1})}{n}$.
  We can rigorously verify this as follows
  \begin{align*}
    \eE[\ U_{a_2} \  | \ N_{a_1}\ ]&=\eE\bracket{\ \textstyle{\sum\nolimits_{i\in G_a\backslash S_{a_1}}} \eE\big[ \ W_i\cdot\mathbb{I}[W_i \geq Z_{n-N_{a_1}}]\ \big]\ }\\
    &= \eE\bracket{\ \textstyle{\sum\nolimits_{i\in G_a\backslash S_{a_1}}\int_{0}^{\beta} x\cdot x^{n-{N_{a_1}-1}}dx} \ }\\
    &= \eE\bracket{\ \textstyle{\sum\nolimits_{i\in G_a\backslash S_{a_1}}1 \ }}\cdot \int_{0}^{\beta} x\cdot x^{n-{N_{a_1}-1}}dx\\
    &= |G_a\backslash S_{a_1}| \cdot \int_{0}^{\beta} x\cdot x^{n-{N_{a_1}-1}}dx\\
    &= (n-N_{a_1}) \cdot \int_{0}^{\beta} x\cdot x^{n-{N_{a_1}-1}}dx\\
    \eE[\ \beta U_{b} \  | \ N_{a_1}\ ]&=  \eE\bracket{\ \textstyle{\sum\nolimits_{j\in G_b}  \eE\big[ \ \hat{V}_j\cdot\mathbb{I}[  \hat{V}_j \geq Z_{n-N_{a_1}}]\ \big]\ }}\\
    &= \eE\bracket{\ \textstyle{\sum\nolimits_{i\in G_b}\int_{0}^{\beta} x\cdot x^{n-{N_{a_1}-1}}dx} \ }\\
    &=\eE\bracket{\ \textstyle{\sum\nolimits_{i\in G_a\backslash S_{a_1}} 1} \ } \cdot \int_{0}^{\beta} x\cdot x^{n-{N_{a_1}-1}}dx\\
    &= |G_b| \cdot \int_{0}^{\beta} x\cdot x^{n-{N_{a_1}-1}}dx\\
    &= n\cdot \int_{0}^{\beta} x\cdot x^{n-{N_{a_1}-1}}dx.
  \end{align*}
  Taking the ratio of the above two equations we get
  \begin{align*}
    \frac{\eE[U_{a_2}|N_{a_1}]}{ n-N_{a_1}} = \frac{\eE[ \beta U_{b} | N_{a_1}]}{n}.\label{eq_ratio_of_utility}\numberthis
  \end{align*}
  From Equations~\eqref{eq_74} and \eqref{eq_ratio_of_utility} we have
  \begin{align*}
    \eE[U_{a_2}  \ | \ N_{a_1}] &= \beta\frac{n-N_{a_1}}{2n-N_{a_1}} \bigg( n-N_{a_1} - \frac{\binom{n-N_{a_1}+1}{2}}{2n-N_{a_1}+1}\bigg)\\
    \eE[U_{b}  \ | \ N_{a_1}] &=  \frac{n}{2n-N_{a_1}} \bigg( n-N_{a_1} - \frac{\binom{n-N_{a_1}+1}{2}}{2n-N_{a_1}+1}\bigg).
  \end{align*}
  We can rewrite the above as follows
  \begin{align*}
    \eE[U_{b}  \ | \ N_{a_1}] &= \frac{n}{2}\bigg[1-\frac{n^2}{(2n-N_{a_1})^2}+O(n^{-1})\bigg]\numberthis\label{eq_exp_ub_conditioned}\\
    \frac{\eE[U_{a_2}\ | \ N_{a_1}]}{\beta}   &= \frac{n}{2}\bigg[2-\frac{n}{2n-N_{a_1}} + O(n^{-1}) \bigg]-\frac{N_{a_1}}{2}-\eE[U_{b} |  N_{a_1}].\numberthis\label{eq_exp_ua1_conditioned}
  \end{align*}
  We give the details of the rearrangement in Section~\ref{sec_omitted_eqs}.
  Now, calculating the expectation of Equations~\eqref{eq_exp_ub_conditioned} and \eqref{eq_exp_ua1_conditioned} we get
  \begin{align*}
    \eE[U_b] &= \eE[\ \eE[U_{b}  \ | \ N_{a_1}]\ ]\\
    &\frac{n}{2}-\frac{n}{2}\eE\bigg[\frac{n^2}{(2n-N_{a_1})^2}+O(n^{-1})\bigg]\\
    &\stackrel{\eqref{fact_expec_2}}{=}\frac{n}{2}\bigg[1-\frac{1}{(1+\beta)^2}+\ohalf\bigg]\label{eq_exp_ub_2}\numberthis\\
    \eE[U_{a_2}] &= \eE[\eE[U_{a_2}  \ | \ N_{a_1}]]\\
    &=n\beta\bigg[1-\frac{1}{2}\eE\bigg[\frac{n}{2n-N_{a_1}}\bigg] + O(n^{-1})\bigg]-\frac{\beta\eE[N_{a_1}]}{2}-\beta\eE[\eE[U_{b} |  N_{a_1}]]\\
    &=n\beta\bigg[1-\frac{1}{2}\eE\bigg[\frac{n}{2n-N_{a_1}}\bigg] + O(n^{-1})\bigg]-\frac{\beta\eE[N_{a_1}]}{2}-\beta\eE[U_{b}]\\
    &\hspace{-7.2mm}\stackrel{\eqref{fact_expec_1},\eqref{eq_dist_n_a_1},\eqref{eq_exp_ub_2}}{=}\hspace{-1mm}n\beta\bigg[1-\frac{1}{2(1+\beta)} + O\bigg(\frac{1}{n}\bigg)\bigg]-\frac{n\beta(1-\beta)}{2}-\frac{n\beta}{2}\bigg[1-\frac{1}{(1+\beta)^2}+O(n^{-\frac{3}{8}})\bigg]\\
    &=n\beta\bigg[1-\frac{1}{2(1+\beta)} -\frac{(1-\beta)}{2}-\frac{1}{2}+\frac{1}{2(1+\beta)^2}+O(n^{-\nfrac{3}{8}})\bigg]\\
    &=n\beta\bigg[\frac{1}{2}-\frac{1}{2(1+\beta)} -\frac{(1-\beta)}{2}+\frac{1}{2(1+\beta)^2}+O(n^{-\nfrac{3}{8}})\bigg]\\
    &=n\beta\bigg[\frac{1}{2}-\frac{\beta}{2(1+\beta)^2} -\frac{(1-\beta)}{2}+O(n^{-\nfrac{3}{8}})\bigg]\\
    &=n\beta\bigg[ \frac{\beta}{2} - \frac{\beta}{2(1+\beta)^2} +O(n^{-\nfrac{3}{8}})\bigg]\\
    &=\frac{n\beta^2}{2}\bigg[ \frac{1}{2} - \frac{1}{2(1+\beta)^2} +O(n^{-\nfrac{3}{8}})\bigg].\label{expec_u_a_1}\numberthis
  \end{align*}
  \noindent
  This completes the proof of Equations~\eqref{eq_exp_ua2} and \eqref{eq_exp_ub}, as  well as Lemma~\ref{lem_utility_unconst_app}.
\end{proof}

\subsubsection{Proof of Lemma~\ref{lem_utility_const_app}}\label{sec_proof_lem_utility_const_app}
\begin{proof}
  Consider the observed utilities $\{w_{i}\}_{i\in G_a}$ and $\{w_{i}\}_{i\in G_b}$  of candidates in  $G_a$ and $G_b$.
  Let $Z_k$ be the $k$-th largest latent utility from $\{w_{i}\}_{i\in G_a}\cup \{w_{i}\}_{i\in G_b}$.\footnote{Note this $Z_k$ is different from the $Z_k$ defined in Section~\ref{sec_proof_lem_utility_unconst_app}.}

  Define $\cH$ to be the event that $N_b = \sum_{j\in G_b} \mathbb{I}[\hat{w}_j \geq Z_{n}]\leq \nfrac{n}{2}$.
  Notice that $\cH$ implies that $\sum_{i\in G_b, j\in [n]}(\piuncons)_{ij}\leq \nfrac{n}{2}$.

  \begin{proposition}\label{prop_pick_propotional}
    Given $\cH$, $\sum_{i\in G_b, j\in [n]}(\picons)_{ij} = \nfrac{n}{2}$, i.e., it ranks exactly $\nfrac{n}{2}$ candidates from $G_b$ in the top-$n$ positions.
  \end{proposition}
  \begin{proof}
    Since $\picons$ satisfies the constraint we have\\ $\sum_{i\in G_b, j\in [n]}(\picons)_{ij} \geq \nfrac{n}{2}$.
    Assume that $\sum_{i\in G_b, j\in [n]}(\picons)_{ij} = \nfrac{n}{2}+j$ for any $0<j\leq \nfrac{n}{2}$.
    Then we would create a ranking $\hat{x}$ such that
    $\mathcal{W}(\hat{x},v,\hat{w}) > \mathcal{W}(\picons,v,\hat{w})$ and $\sum_{i\in G_b, j\in [n]}(\hat{x})_{ij} = \nfrac{n}{2} + j - 1$.
    Contradicting the optimality of $\picons$.

    Let $w_1 > w_2 > \dots > w_n\in [0,1]$ and $\hat{v}_1 > \hat{v}_2 > \dots > \hat{v}_n\in [0,\beta]$ be the observed utilities of all candidates in $G_a$ and $G_b$ respectively.
    Conditioned on $\cH$, $\piuncons$ ranked $w_{(\nfrac{n}{2}-j)}$ in the top-$n$ positions but didn't rank $\hat{v}_{(\nfrac{n}{2}+j)}$.
    Therefore, we must have $w_{(\nfrac{n}{2}-j)} > \hat{v}_{(\nfrac{n}{2}+j)}$.

    Since the utilities are in strictly decreasing,
    we claim that $\picons$ ranks exactly $(\nfrac{n}{2}+j)$ candidates from $G_b$ with the observed utilities $\{\hat{v}_1,\hat{v}_2,\cdots,\hat{v}_{\nfrac{n}{2}+j}\}$.
    Similarly, $\picons$ ranks the candidates with observed utilities $\{w_1,w_2,\cdots,w_{\nfrac{n}{2}-j}\}$ from $G_a$.
    Construct $\hat{x}$ by swapping the positions of $\hat{v}_{\nfrac{n}{2}+j}$ and $w_{\nfrac{n}{2}-j+1}$ in $\picons$.
    Then we have
    \begin{align*}
      \mathcal{W}(\hat{x},v,\hat{w}) &= \mathcal{W}(\picons,v,\hat{w}) +\hat{v}_{(\nfrac{n}{2}+j)} - w_{(\nfrac{n}{2}-j)}\\
      &> \mathcal{W}(\picons,v,\hat{w}).
    \end{align*}
  \end{proof}

  For all $k\in [n]$ and $s\in \{a, b\}$, let $U_{k,s}$ be the latent utility of the $k$ candidates with the highest utility in $G_s$.
  Then from Proposition~\ref{prop_pick_propotional} the expected latent utility of the constrained ranking given $\cH$
  \begin{align*}
    \eE[U_{\text{cons}}\ | \ \cH]=\eE[U_{\nfrac{n}{2},a}+U_{\nfrac{n}{2},b}\ |\ \cH].
  \end{align*}

  \noindent However, $\cH$ is not independent of $U_{\nfrac{n}{2},a}$ and $U_{\nfrac{n}{2},b}$.
  Therefore, we cannot calculate $\eE[U_{\text{cons}}\ | \ \cH]$, directly from $\eE[U_{\nfrac{n}{2},a}]$ and $\eE[U_{\nfrac{n}{2},b}]$.
  We overcome this by showing that we show that conditioning on $\cH$ changes the expectation by a small amount for large $n$.
  Towards this, we first show that $\cH$ occurs with high probability, and then, we using the fact that the utility is bounded we get the required result.
  \begin{align*}
    \Pr[\cH]&=\Pr\bigg[N_b \leq \frac{n}{2}\bigg]\\
    &\hspace{-4.5mm}\stackrel{\text{Prop.}~\ref{thm_expected_cand}}{=}\  \Pr\bigg[\ N_b-\eE[N_b] \leq n\bigg(\frac{1}{2}-\frac{\beta}{1+\beta}+\ohalf\bigg)\ \bigg]\\
    &\geq \Pr\bigg[\ |N_b-\eE[N_b]| \leq n\bigg(\frac{1-\beta}{2(1+\beta)}+\ohalf\bigg)\ \bigg].
  \end{align*}
  Since $\beta\in(0,1)$ is a constant, we have that $n\big(\frac{1-\beta}{2(1+\beta)}+\ohalf\big)=\Theta(n)$.
  Therefore, from using Equation~\eqref{eq_conc_n_b} of Proposition~\ref{thm_expected_cand}  we have
  \begin{align*}
    \Pr[\cH]=\Pr[N_b \leq \nfrac{n}{2}]\ \stackrel{{\rm Prop.}\ref{thm_expected_cand}}{\geq}\ 1- \oexp.\label{eq_small_number_whp}\numberthis
  \end{align*}
  \noindent We have shown that $\cH$ occurs with high probability.
  Now, we can approximate $\eE[U_{\frac{n}{2},a}+U_{\frac{n}{2},b}\ |\ \cH]$ using the following equality
  \begin{align*}
    \eE[\ U_{\frac{n}{2},a}+U_{\frac{n}{2},b}]&= \eE[\ U_{\frac{n}{2},a}+U_{\frac{n}{2},b}\ |\ \cH] \Pr[\cH] + \eE[\ U_{\frac{n}{2},a}+U_{\frac{n}{2},b}\ |\ \cHC] \Pr[\cHC]\\
    &= \eE[\ U_{\frac{n}{2},a}+U_{\frac{n}{2},b}\ |\ \cH] \Pr[\cH] + O(n) \Pr[\cHC]\\
    &= \eE[\ U_{\frac{n}{2},a}+U_{\frac{n}{2},b}\ |\ \cH]\cdot \big(1- \oexp\big) + O(n)\cdot \oexp\\
    \eE[\ U_{\frac{n}{2},a}+U_{\frac{n}{2},b}\ |\ \cH]&= \bigg(\eE[\ U_{\frac{n}{2},a}+U_{\frac{n}{2},b}]- \onexp\big)\bigg) \cdot \big(1- \oexp\big)^{-1}\\
    &=\bigg(\eE[\ U_{\frac{n}{2},a}+U_{\frac{n}{2},b}]- \onexp\big)\bigg) \cdot \big(1+ \oexp\big)\\
    &=\eE[\ U_{\frac{n}{2},a}+U_{\frac{n}{2},b}] \cdot \big(1+ \oexp\big) -  \onexp\big).\label{approximation_of_utility}\numberthis
  \end{align*}
  Here, we can calculate unconditioned expectation $\eE[\ U_{\frac{n}{2},a}+U_{\frac{n}{2},b}]$ using the expectations of order statistics of the uniform distribution as follows
  \begin{align*}
    \eE[\ U_{\frac{n}{2},a}+U_{\frac{n}{2},b}] &=\eE[\ U_{\frac{n}{2},a}\ ]+\eE[\ U_{\frac{n}{2},b}\ ]\\
    &= \sum\nolimits_{k=1}^{\nfrac{n}{2}}\frac{n+1-k}{n+1} + \sum\nolimits_{k=1}^{\nfrac{n}{2}}\frac{n+1-k}{n+1}\\
    &=\frac{3n}{4}\bigg(1-\frac{1}{(n+1)}\bigg).
  \end{align*}
  Then using Equation~\eqref{approximation_of_utility} we have
  \begin{align*}
    \eE[\ U_{\frac{n}{2},a}+U_{\frac{n}{2},b}\ |\ \cH]&=  \left(\frac{3n}{4}\bigg(1-\frac{1}{(n+1)}\bigg)-\oexp\big)\right) \cdot \big(1+\oexp\big)\\
    &=\frac{3n}{4}\bigg(1-\frac{1}{(n+1)}\bigg)\cdot \big(1+ \oexp \big)-O\big(n e^{-n^{2\delta}}\big)\\
    &=\frac{3n}{4}\bigg(1-\frac{1}{(n+1)}\bigg)\cdot \big(1+\oexp\big).\label{eq_exp_with_cons_1}\numberthis
  \end{align*}
  \noindent Using this to approximate the utility we get
  \begin{align*}
    U_{\mathcal{D},1}(\alpha^\star,\beta)&=\eE[U_{\text{cons}}]= \eE[U_{\text{cons}}\ | \ \cH]\Pr[\cH] + \eE[U_{\text{cons}}\ | \ \overline{H} ]\Pr[\cHC]\\
    &= \eE[ U_{\frac{n}{2},a}+U_{\frac{n}{2},b}\ |\ \cH ]\Pr[\cH] + \eE[ U_{\frac{n}{2},a}+U_{\frac{n}{2},b}\ | \ \cHC ]\Pr[\cHC]\\
    &=\eE[ U_{\frac{n}{2},a}+U_{\frac{n}{2},b} \ | \ \cH]\Pr[\cH] + O(n)\Pr[\cHC]\\
    &\hspace{-1mm}\stackrel{\eqref{eq_small_number_whp}}{=}\ \eE[ U_{\frac{n}{2},a}+U_{\frac{n}{2},b} \ | \ \cH]\big(1-\oexp \big) + O(n)\cdot \oexp \\
    &\hspace{-1mm}\stackrel{\eqref{eq_exp_with_cons_1}}{=}\frac{3n}{4}\cdot \bigg(1-\frac{1}{n+1}\bigg)\cdot\big(1- \oexp \big)^2 + O(n)\cdot \oexp \\
    &= \frac{3n}{4}\cdot \bigg(1-\frac{1}{n+1}\bigg)\cdot\big(1 - 2\oexp + O\big(e^{-n^{4\delta}}\big) \big) + O(n)\cdot \oexp \\
    &= \frac{3n}{4}\cdot \bigg(1-\frac{1}{n+1}\bigg)\cdot\big(1- \oexp \big) + O(n)\cdot \oexp \\
    &= \frac{3n}{4}\cdot \bigg(1-\frac{1}{n+1}\bigg)\cdot\big(1- \oexp + \oexp)\\
  \end{align*}
  \begin{align*}
    &=\frac{3n}{4}\bigg(1-\frac{1}{n+1}\bigg)\bigg(1+ \oexp\bigg)\\
    &=\frac{3n}{4}\bigg(1-\frac{1}{n+1}+\oexp\bigg).
  \end{align*}
  This completes the proof of Lemma~\ref{lem_utility_const_app}.
\end{proof}

\subsubsection{Extending the proof to $m_a\neq m_b\neq n$}\label{sec_extending_proof}
The core ideas of proof do not change when we extend it to the setting with $m_a\neq m_b\neq n$.
In this section, we discuss the key-points where the proof changes or involves a non-trivial generalization of the techniques.
We begin by stating (without proof) straightforward generalizations of Proposition~\ref{fact_negative_exp} and Proposition~\ref{thm_expected_cand} which would be useful.
\begin{proposition}
  Let $\beta\in (0,1)$, $N\sim \text{Binomial}(m_a,1-\beta)$ and $r\in \cR$, then for any $\delta > 0$
  \begin{align*}
    \eE\bigg[  \frac{1}{r+N}  \bigg] &= \frac{1}{r+\eE[N]}\bigg[1+ \ohalf\bigg],\numberthis\\
    \eE\bigg[  \frac{1}{(r+N)^2}  \bigg] &= \frac{1}{(r+\eE[N])^2}\bigg[1+ \ohalf\bigg].\numberthis
  \end{align*}
\end{proposition}

\begin{proposition}\label{thm_expected_cand_2}
  {\bf (Expected number of candidates without constraints).}
  Given a number $\beta\in(0,1]$, representing the implicit bias, we have
  \begin{align}
    &N_{a_1}\stackrel{d}{=} \text{Binomial}(m_a,1-\beta),\label{eq_dist_na1}\\
    &\eE[N_b] \leq m_b\Paren{\frac{n-m_a(1-\beta)}{m_a\beta+m_b}}+\ohalf,\label{eq_exp_nb}\\
    &\Pr\big[ |N_b - \eE[N_b]| \leq n^{\nfrac{5}{8}} \big] \geq 1-\oexp.\label{eq_conc_nb}
  \end{align}
\end{proposition}

\medskip

\paragraph{Proof of Equation~\myeqref{eq_utility_with_const}{16}.}
The proof proceeds in two steps.
In the first step, we show that with high probability $\picons$ satisfies (similar to Section~\ref{sec_proof_lem_utility_const_app})
$$\sum_{i\in G_b, j\in [n]}(\picons)_{ij} = n\cdot\frac{m_b}{(m_a+m_b)}.$$
Towards this, we define an event $\cH^\prime$ as $N_b\leq n\cdot\nfrac{m_b}{(m_a+m_b)}$ and show that $\cH^\prime$ occurs with high probability.
Then we generalize Proposition~\ref{prop_pick_propotional} to prove the claim.
In the second step, we calculate $\eE\nolimits_{w}\big[\mathcal{W}(\picons, v, w) | \cH^\prime]$ conditioned on $\cH^\prime$.
Since, $\cH^\prime$ occurs with high probability and $U_{\mathcal{D},1}(\alpha^\star, \beta)$ is bounded,
it follows that
$$U_{\mathcal{D},1}(\alpha^\star, \beta)=\eE\nolimits_{w}\big[\mathcal{W}(\picons, v, w)]\approx \eE\nolimits_{w}\big[\mathcal{W}(\picons, v, w) | \cH^\prime].$$

\noindent {\bf Step 1.}
Define $\cH^\prime$ to be the event that $N_b\leq n\cdot\nfrac{m_b}{(m_a+m_b)}$,
i.e.,
$$\sum_{i\in G_b, j\in [n]}(\piuncons)_{ij} \leq n\cdot\frac{m_b}{(m_a+m_b)}.$$
Notice from Equation~\eqref{eq_exp_nb}, that $\eE[N_b]\leq n-\Theta(n)$.
With some simple manipulation we will get
\begin{align*}
  \Pr[\cH^\prime] &\coloneqq \Pr\Bracket{N_b \leq n\cdot\frac{m_b}{m_a+m_b}}\\
  &\geq \Pr[|N_b-\eE[N_b]| \leq \Theta(n)]\\
  &\geq 1-\oexp.
\end{align*}
We can generalize Proposition~\ref{prop_pick_propotional} into the following.
\begin{proposition}
  Given $\cH^\prime$, we have
  $$\sum_{i\in G_b, j\in [n]}(\picons)_{ij} = n\cdot\frac{m_b}{(m_a+m_b)},$$
  i.e., $\picons$ ranks exactly $n\cdot\frac{m_b}{(m_a+m_b)}$ candidates from $G_b$ in the top-$n$ positions.
\end{proposition}
This follows from the proof of Proposition~\ref{prop_pick_propotional} by replacing $n$ by $m_a$ and $m_b$ at when ever it refers to the size of the groups.
This will complete the first step.

\noindent {\bf Step 2.}
Define $\eta\coloneqq \nfrac{n}{m_a+m_b}$. We can calculate $\eE[U_{m_a\eta,a}+U_{m_b\eta,b}]$ as follows
\begin{align*}
  \eE[U_{m_a\eta,a}+U_{m_b\eta,b}]&= \sum_{k=1}^{\eta\cdot m_a}\frac{m_a+1-k}{m_a+1}+\sum_{k=1}^{\eta\cdot m_b}\frac{m_b+1-k}{m_b+1}\\
  & = n\bracket{1-\frac{n}{2(m_a+m_b)}+O(n^{-1})}.
\end{align*}
Where the first equality follows from the expectation of order statistics of the uniform distribution.
Finally, following the arguments used to derive Equation~\eqref{approximation_of_utility} we have that
\begin{align*}
  U_{\mathcal{D},1}(\alpha^\star, \beta)&=\eE\nolimits_{w}\big[\mathcal{W}(\picons, v, w)]\\
  & = \eE\nolimits_{w}\big[\mathcal{W}(\picons, v, w)\ | \ \cH]+\onexp\\
  & = \eE\nolimits_{w}\big[U_{m_a\eta,a}+U_{m_b\eta,b}\ | \ \cH]+\onexp\\
  & = \eE\nolimits_{w}\big[U_{m_a\eta,a}+U_{m_b\eta,b}]+\onexp\\
  & = n\Bracket{1-\frac{n}{2(m_a+m_b)}+O(n^{-1})}.
\end{align*}
This will complete the proof of Equation~\myeqref{eq_utility_with_const}{16}.

\medskip

\paragraph{Proof of Equation~\myeqref{eq_utility_without_const}{17}.}
Define
$$c\coloneqq \eE[N_{a_1}]\stackrel{\eqref{eq_dist_na1}}{=} m_a(1-\beta).$$
We consider two cases when $c\geq n+\Theta(n^{\nfrac{5}{8}})$ and $c=n-\omega(n^{\nfrac{5}{8}})$.
We discuss outline the proof for each case separately.

\noindent {\bf Case A ($c\geq n+\Theta(n^{\nfrac{5}{8}}))$:}
From Proposition~\ref{thm_expected_cand_2} we have that
$$\eE[N_{a_1}]\geq n+\Theta(n^{\nfrac{5}{8}}).$$
Define $\cE^\prime$ to be the event that $N_{a_1} \geq  n - \Theta(n^{\nfrac{5}{8}})$.
Applying Hoeffding's we find
\begin{align*}
  \Pr[\cE^\prime] = \Pr[N_{a_1} \geq  n - \Theta(n^{\nfrac{5}{8}})] &\geq \Pr[N_{a_1} - \eE[N_{a_1}] \geq  \Theta(n^{\nfrac{5}{8}}) ]\tag{Using $\eE[N_{a_1}]=c\geq n+\Theta(n^{\nfrac{5}{8}})$}\\
  &\geq 1-O\big(e^{-\frac{n^{\nfrac{5}{4}}}{m_a}}\big)\\
  &\geq  1-\oexp.\numberthis
\end{align*}
We can show that conditioned on $\cE^\prime$,  $$U_{a_2}=O(n^{-\nfrac{3}{8}})\ \ \text{and}\ \ U_{b}=O(n^{-\nfrac{3}{8}}).$$
Therefore $U_{a_2}$ and $U_{b}$ have small contributions to the utility.
Further, since $\cE^\prime$ occurs with a high probability and $U_{\mathcal{D},1}(0, \beta)$ is bounded, we have that %
$$U_{\mathcal{D},1}(0, \beta)= \eE\nolimits_{w}\big[\mathcal{W}(\piuncons, v, w)]\approx \eE\nolimits_{w}\big[\mathcal{W}(\piuncons, v, w)|\cE^\prime].$$
Finally, calculating $\eE\big[U_{a_1}|\cE^\prime]$ will give us the required equation.

\noindent Formally,
\begin{align*}
  U_{\mathcal{D},1}(0, \beta) &= \eE\nolimits_{w}\big[\mathcal{W}(\piuncons, v, w)]\\
  & =\eE\nolimits_{w}\big[\mathcal{W}(\piuncons, v, w)|\cE^\prime]\Pr[\cE^\prime]+O(n)\Pr[\overline{\cE^\prime}]\\
  & =\eE\big[U_{a_1}+U_{a_2}+U_{b}|\cE^\prime]\cdot\paren{1+\oexp}+\onexp\\
  & =\eE\big[U_{a_1}|\cE^\prime]\cdot\paren{1+\oexp}+O(n^{\nfrac{5}{8}})+\onexp.
\end{align*}
Conditioned on $\cE^\prime$ we can bound the utility $\eE\big[U_{a_1}|\cE^\prime]$ between the following
\begin{align*}
  \eE\big[U_{a_1}|\cE^\prime]&\leq \eE\big[U_{a_1}|N_{a_1}=n]\\
  &= \sum_{k=1}^{n}\frac{m_a+1-k}{m_a+1}\\
  &= n\bracket{1-\nfrac{n}{(2m_a)}+O(n^{-1})}\label{eq_prf_lwbnd}\numberthis\\
  \eE\big[U_{a_1}|\cE^\prime]&\geq \eE\big[U_{a_1}|N_{a_1}=n-\Theta(n^{\nfrac{5}{8}})]\\
  &= \sum_{k=1}^{n-\Theta(n^{\nfrac{5}{8}})}\frac{m_a+1-k}{m_a+1}\\
  &= n\bracket{1-\nfrac{n}{(2m_a)}+\Theta(n^{\nfrac{-3}{8}})}.\label{eq_prf_ubnd}\numberthis
\end{align*}
Combining Equation~\eqref{eq_prf_lwbnd} and Equation~\eqref{eq_prf_ubnd} we get
\begin{align*}
  \eE\big[U_{a_1}|\cE^\prime] &\stackrel{\eqref{eq_prf_lwbnd}, \eqref{eq_prf_ubnd}}{=} n\cdot\bracket{1-\nfrac{n}{(2m_a)}+\ohalf }\\
  U_{\mathcal{D},1}(0, \beta)&\hspace{4.2mm}= n\bracket{1-\nfrac{n}{(2m_a)}+\ohalf }\paren{1+\oexp}+O(n^{\nfrac{5}{8}}+ne^{-n^{\nfrac{1}{4}}})\\
  &\hspace{4.2mm}=n\bracket{1-\nfrac{n}{(2m_a)}+\ohalf }.
\end{align*}
This will complete the proof for Case A.

\medskip

\noindent {\bf Case B ($c=n-\omega(n^{\nfrac{5}{8}}))$:}
Note that this case encapsulates the setting with $m_a=m_b=n$.
Unlike the special case ($m_a=m_b=n$), here we could have  which implies $N_{a_2}=N_b=0$.
\newcommand{\cEpp}{\cE^{\prime\prime}}
\newcommand{\cEppc}{\overline{\cE^{\prime\prime}}}
Define $\cEpp$ as the event $N_{a_1}>n$.
We show that $\cEpp$ occurs with a low probability, and so has a small impact on $U_{\mathcal{D},1}(0, \beta)$.
Having handled this, we follow the proof for the special case ($m_a=m_b=n$), substituting $n$ by $m_a$ and $m_b$ when it is used to refer to the group sizes.

Formally, we proceed as follows
\begin{align*}
  \Pr[\cEpp] &=1-\Pr[N_{a_1}\leq n]\\
  &=1-\Pr[N_{a_1}-\eE[N_{a_1}]\leq \omega(n^{\nfrac{5}{8}})]\\
  &\leq \oexp\\
  U_{\mathcal{D},1}(0, \beta) &= \eE\nolimits_{w}\big[\mathcal{W}(\piuncons, v, w)]\\
  &= \eE\nolimits_{w}\big[\mathcal{W}(\piuncons, v, w)|\cEppc]\Pr[\cEppc]+\eE\nolimits_{w}\big[\mathcal{W}(\piuncons, v, w)|\cEpp]\Pr[\cEpp]\\
  &= \eE\nolimits_{w}\big[\mathcal{W}(\piuncons, v, w)|\cEppc]\cdot\paren{1-\oexp}+\onexp\\
  &= \eE\big[U_{a_1}+U_{a_2}+U_b|\cEppc]\cdot\paren{1-\oexp}+\onexp.
\end{align*}
Next, we calculate $\eE\big[U_{a_1}|\cEppc]$, $\eE\big[U_{a_2}|\cEppc]$, and $\eE\big[U_{b}|\cEppc]$.
We follow the proof in Section~\ref{sec_proof_lem_utility_unconst_app}.
To do so, we only need the additional fact that $n\leq \min(m_a, m_b)$.\footnote{This is for ease of notation, otherwise the error terms would involve $m_a$, $m_b$, and $n$.}
$\text{This will give us}$
\begin{align*}
  \eE\big[U_{a_1}|\cEppc] &= \eE\nolimits_{k}\bracket{\ \eE\big[U_{a_1}\ |\ \cEppc, N_{a_1}=k]\ }\\
  &= \eE\nolimits_{k}\bracket{N_{a_1}\cdot \nfrac{(1+\beta)}{2} \ |\  \cEppc}\\
  &= \frac{(1+\beta)}{2}\cdot\paren{m_a(1-\beta)+\onexp}\cdot\paren{1+\oexp}\\
  &= \frac{m_a}{2}(1-\beta^2)+\onexp,\numberthis\label{eq_prf_ua1}\\
  \eE\big[U_{a_2}|\cEppc] &= \eE\nolimits_{k}\bracket{\ \eE\big[U_{a_2}\ |\ \cEppc, N_{a_1}=k]\ }\\
  &=\frac{m_b}{2} -\frac{m_b}{2}\eE\Bracket{\frac{(m_a+m_b-n)^2}{(m_a+m_b-N_{a_1})^2} + O(n^{-1})}\\
  &=\frac{m_b}{2}\Bracket{ 1-\frac{(m_a+m_b-n)^2}{(m_a\beta+m_b)^2}\cdot\paren{1+\ohalf}},\numberthis\label{eq_prf_ua2}\\
  \eE\big[U_{a_2}|\cEppc] &= \eE\nolimits_{k}\bracket{\ \eE\big[U_{a_2}\ |\ \cEppc, N_{a_1}=k]\ }\\
  &=\frac{\beta}{2}(n-\eE[N_{a_1}|\cEppc])
  -\beta\eE\nolimits_{k}\bracket{ \eE\big[U_{a_2}| \cEppc, N_{a_1}=k] }\\
  &\hspace{3.75mm}+\frac{\beta}{2}(m_a+m_b-n)\Bracket{
  1-\eE\Bracket{\frac{(m_a+m_b+1-n)}{m_a+m_b-N_{a_1}+1}}
  }\\
  &= \frac{m_a\beta^2}{2}\Bracket{
  1-\frac{(m_a+m_b-n)^2}{(m_a\beta+m_b)^2}\cdot\paren{1+\ohalf}
  }.\numberthis\label{eq_prf_ub}
\end{align*}
Combining equations~\eqref{eq_prf_ua1}, \eqref{eq_prf_ua2}, and \eqref{eq_prf_ub} we get
\begin{align*}
  U_{\mathcal{D},1}(0, \beta)& = \eE\big[U_{a_1}+U_{a_2}+U_b\ |\ \cEppc]\cdot\paren{1-\oexp}+\onexp\\
  & = \frac{m_a}{2}(1-\beta^2)+  \frac{m_a\beta^2+m_a}{2}\Bracket{ 1-\frac{(m_a+m_b-n)^2}{(m_a\beta+m_b)^2}\cdot\paren{1+\ohalf} }.
\end{align*}
This will complete the proof of Case B and also of Equation~\myeqref{eq_utility_without_const}{17}.

\subsubsection{Details omitted from equations~\eqref{eq_exp_ub_conditioned} and \eqref{eq_exp_ua1_conditioned}}\label{sec_omitted_eqs}
\begin{align*}
  \eE[U_{b}  \ | \ N_{a_1}] &=  \frac{n}{2n-N_{a_1}} \bigg( n-N_{a_1} - \frac{(n-N_{a_1}+1)(n-N_{a_1})}{2(2n-N_{a_1}+1)}\bigg)\\
  &= n-\frac{n^2}{2n-N_{a_1}}-\frac{n(n-N_{a_1}+1)(n-N_{a_1})}{2(2n-N_{a_1})(2n-N_{a_1}+1)} \\
  &= n-\frac{n^2}{2n-N_{a_1}}-\frac{n(n-N_{a_1}+1)}{2(2n-N_{a_1}+1)}+\frac{n^2(n-N_{a_1}+1)}{2(2n-N_{a_1})(2n-N_{a_1}+1)} \\
  &= \frac{n}{2}-\frac{n^2}{2n-N_{a_1}}-\frac{n^2}{2(2n-N_{a_1}+1)}+\frac{n^2}{2(2n-N_{a_1})} - \frac{n^3}{2(2n-N_{a_1})(2n-N_{a_1}+1)} \\
  &= \frac{n}{2}-\frac{n^2}{2(2n-N_{a_1})}-\frac{n^2}{2(2n-N_{a_1}+1)} - \frac{n^3}{2(2n-N_{a_1})(2n-N_{a_1}+1)} \\
  &= \frac{n}{2}+\frac{n^2}{2(2n-N_{a_1}+1)(2n-N_{a_1})} - \frac{n^3}{2(2n-N_{a_1})(2n-N_{a_1}+1)} \\
  &= \frac{n}{2}\bigg[1-\frac{n^2}{(2n-N_{a_1})(2n-N_{a_1}+1)}+O(n^{-1})\bigg]\\
  &= \frac{n}{2}\bigg[1-\frac{n^2}{(2n-N_{a_1})^2}+\frac{n^2}{(2n-N_{a_1})^2(2n-N_{a_1}+1)}+O(n^{-1})\bigg]\\
  &= \frac{n}{2}\bigg[1-\frac{n^2}{(2n-N_{a_1})^2}+O(n^{-1})\bigg].\\
  \frac{1}{\beta}\eE[U_{a_2}  \ | \ N_{a_1}] &= \frac{n-N_{a_1}}{2n-N_{a_1}} \bigg( n-N_{a_1} - \frac{(n-N_{a_1}+1)(n-N_{a_1})}{2(2n-N_{a_1}+1)}\bigg)\\
  &=n-N_{a_1}-\frac{(n-N_{a_1}+1)(n-N_{a_1})}{2(2n-N_{a_1}+1)}-\eE[U_{b}  \ | \ N_{a_1}]\\
  &=n-N_{a_1}-\frac{(n-N_{a_1})}{2}+\frac{n(n-N_{a_1})}{2(2n-N_{a_1}+1)}-\eE[U_{b}  \ | \ N_{a_1}]\\
  &=n-N_{a_1}-\frac{(n-N_{a_1})}{2}+\frac{n}{2}-\frac{n(n+1)}{2(2n-N_{a_1}+1)}-\eE[U_{b}  | N_{a_1}]\\
  &=n-\frac{N_{a_1}}{2}-\frac{n}{2(2n-N_{a_1})}\bigg[n-\frac{n-N_{a_1}}{(2n-N_{a_1}+1)} \bigg]-\eE[U_{b}  | N_{a_1}]\\
  &=\frac{n}{2}\bigg[2-\frac{n}{(2n-N_{a_1})} + O(n^{-1}) \bigg]-\frac{N_{a_1}}{2}-\eE[U_{b}  \ | \ N_{a_1}].
\end{align*}

\subsection{Fact~\ref{fact_impossibility} and its proof}
The following fact shows that without any assumptions on $w$, any constraint $\mathcal{K}$, such that
$$\max_{x\in \mathcal{K}}\mathcal{W}(x,v,\hat{w})=\max_{x}\mathcal{W}(x,v,w)$$
must be a function of $w$.
Here, by $x\in\mathcal{K}$ we imply that $x$ $\text{satisfies $\mathcal{K}$.}$
\begin{fact}\label{fact_impossibility}
  {\bf (No fixed constraint can mitigate implicit bias for arbitrary utilities).}
  There exists a bias parameter $\{\beta_s\}_{s=1}^p\in[0,1)$, position-discount $v$, and two sets of utilities $\{w_i\}_{i=1}^{m}$ and $\{w^\prime_i\}_{i=1}^{m}$, such that, no constraint $\mathcal{K}$ can ensure the following together
  \begin{align*}
    \mathcal{W}(\tilde{x},v,{w})&=\max_{x}\mathcal{W}(x,v,w)\label{equation_1}\numberthis\\
    \mathcal{W}(\tilde{x}^\prime,v,{w}^\prime)&=\max_{x}\mathcal{W}(x,v,w^\prime),\label{equation_2}\numberthis
  \end{align*}
  where $\tilde{x}\coloneqq\argmax_{x\in \mathcal{K}}\mathcal{W}(x,v,\hat{w})$ and $\tilde{x}^\prime\coloneqq\argmax_{x\in \mathcal{K}}\mathcal{W}(x,v,\hat{w}^\prime)$.
\end{fact}
\begin{proof}
  Let $m=n=2$ and $v=\{2,1\}$, and suppose there are only two groups $G_a=\{1\}$, $G_b=\{2\}$ ($p=2$).
  The implicit bias is $\beta_a=1$ and $\beta_b=\nfrac{1}{4}$, and the two utilities are $w=\{2,1\}$ and $w^\prime=\{1,2\}$.
  Consider the optimal rankings
  \begin{align*}
    \overline{x}&=\argmax_{x}\mathcal{W}(x,v,w)=\begin{pmatrix}
    1 & 0 \\
    0 & 1
    \end{pmatrix},\\
    {\overline{x}^\prime}&=\argmax_{x}\mathcal{W}(x,v,w^\prime)=\begin{pmatrix}
    0 & 1 \\
    1 & 0
    \end{pmatrix}.
  \end{align*}
  $\overline{x}$ and $\overline{x}^\prime$ rank candidates in the decreasing order of their latent utilities, $w$ and $w^\prime$ respectively.
  Therefore, they are optimal $w$ and $w^\prime$.
  Since they are also the only two possible rankings, we can calculate their utilities to verify that they are unique maxima.

  Assume towards a contradiction that there exists a constraint $\mathcal{K}$ which satisfies Equations~\eqref{equation_1} and \eqref{equation_2}.
  Since $\overline{x}$ is a unique maximum, and the ranking $\tilde{x}\in \mathcal{K}$ from the fact satisfies
  $$\mathcal{W}(\tilde{x},v,{w})=\mathcal{W}(\overline{x},v,{w}).$$
  Therefore, $\tilde{x}=\overline{x}$, i.e., $\tilde{x}\in \mathcal{K}$.
  Similarly, we have $\overline{x}^\prime\in \mathcal{K}$.

  \noindent  Consider the ranking $\tilde{x}^\prime$ from the fact,
  \begin{align*}
    \tilde{x}^\prime\coloneqq\argmax_{x\in \mathcal{K}}\mathcal{W}(x,v,\hat{w}^\prime)=\begin{pmatrix}
    1 & 0 \\
    0 & 1
    \end{pmatrix}.
  \end{align*}
  The above ranking ranking candidates in decreasing order of $\hat{w}$.
  Therefore, it is optimal for $\hat{w}$.
  Since $\mathcal{W}(\tilde{x}^\prime, v,w^\prime)=4<5=\mathcal{W}(\overline{x}^\prime, v,w^\prime)$, we have a contradiction.
\end{proof}

\noindent{\bf Acknowledgment.} {Part of this research was supported by an AWS MLRA 2019 Award.}

\bibliographystyle{plain}
\bibliography{bib}

\appendix

\end{document}